\documentclass[11pt,a4paper]{article}
\usepackage[T1]{fontenc}
\usepackage[english]{babel}

\usepackage{ulem}
\usepackage{amsmath, nccmath}
\usepackage{amsfonts}
\usepackage{amsthm}
\usepackage{amssymb}
\usepackage{bbm}

\usepackage[shortlabels]{enumitem}
\usepackage{tikz}
\usepackage{tikz-cd}
\usepackage[colorlinks,linkcolor=blue]{hyperref}
\usepackage{graphicx}

\usepackage{float}
\usepackage{geometry}
\usepackage{subfigure}
\usepackage{multicol}
\usepackage{mathrsfs}
\usepackage{xcolor}
\usepackage{tcolorbox}
\usepackage{framed}
\usepackage{mathtools}
\usepackage{extarrows}

\setlist{noitemsep}
\allowdisplaybreaks

\tcbuselibrary{breakable}
\newcommand{\me}{\textrm{e}}
\newcommand{\ddp}{\displaystyle}
\newcommand{\mm}[1]{\mathrm{#1}}
\newcommand{\mb}[1]{\mathbb{#1}}

\newcommand{\mn}[1]{\mathbf{#1}}

\newcommand{\dd}{\mathrm{d}}

\newcommand{\ol}[1]{\overline{#1}}

\newtcolorbox{defbox}{colback=black!5!white, colframe=black!75!white, boxrule=0mm, leftrule=1mm, sharp corners, breakable}

\theoremstyle{plain}
\theoremstyle{definition}

\newtheorem{deef}{Definition}[section]

\theoremstyle{plain}
\newtheorem{lemm}{Lemma}[section]
\newtheorem{thrm}{Theorem}
\newtheorem{prop}[lemm]{Proposition}
\newtheorem{corr}[lemm]{Corollary}

\theoremstyle{definition}
\newtheorem{exxx}{Example}[section]
\newtheorem{def7}{Remark}[section]
\theoremstyle{remark}

\newtheorem*{notation}{Notation}

\newtheoremstyle{exercise}
{}{}{}{}
{\sffamily\bfseries}
{.}
{ }{\thmname{#1}\thmnumber{\phantom{a}#2}\thmnote{\textnormal{\phantom{i}(#3)}}}

\theoremstyle{exercise}

\numberwithin{exe}{section}

\theoremstyle{definition}

\usepackage{titlesec}
\titleformat*{\subsection}{\Large\sffamily\bfseries}
\titleformat*{\section}{\LARGE\sffamily\bfseries}
\titleformat*{\subsubsection}{\large\sffamily\bfseries}

\usepackage{textcomp}

\newcommand{\fk}[1]{\mathfrak{#1}}

\newcommand{\tts}{\textstyle}

\newcommand{\ii}{\mathrm{i}}

\newcommand{\bnom}[2]{\genfrac{[}{]}{0pt}{}{#1}{#2}}

\newcommand{\ank}[1]{\left\langle #1 \right\rangle}
\newcommand{\bank}[1]{\big\langle #1 \big\rangle}
\newcommand{\Bank}[1]{\Big\langle #1 \Big\rangle}
\newcommand{\sank}[1]{\langle #1 \rangle}

\newcommand{\wh}[1]{\widehat{#1}}

\DeclareMathOperator{\spn}{Span}

\DeclareMathOperator{\ttr}{tr}
\DeclareMathOperator{\diag}{diag}

\DeclareMathOperator{\supp}{supp}

\DeclareMathOperator{\Mor}{Mor}

\DeclareMathOperator{\vol}{vol}
\DeclareMathOperator{\ddiv}{div}

\DeclareMathOperator{\PI}{PI}
\DeclareMathOperator{\DN}{DN}

\DeclareMathOperator{\detz}{det_{\zeta}}
\DeclareMathOperator{\detf}{det_{\mathrm{Fr}}}

\newcommand{\nrm}[1]{\left\lVert#1\right\rVert}
\newcommand{\smx}[1]{\begin{smallmatrix}#1\end{smallmatrix}}
\newcommand{\mss}[1]{\mathscr{#1}}

\newcommand{\lto}{\longrightarrow}

\newcommand{\one}{\mathbbm{1}}

\newcommand{\bnrm}[1]{\big\lVert#1\big\rVert}
\newcommand{\snrm}[1]{\lVert#1\rVert}

\renewcommand{\tilde}{\widetilde}
\renewcommand{\ge}{\geqslant}
\renewcommand{\le}{\leqslant}
\newcommand{\defeq}{\overset{\mathrm{def}}{=}}
\newcommand{\heueq}{\overset{\mathrm{heu}}{=}}

\usepackage{pdfpages}
\usepackage{blkarray}

\usepackage{fancyhdr}
\numberwithin{equation}{section}

\newenvironment{thrmbis}[1]
  {\renewcommand{\thethrm}{\ref{#1}$'$}%
   \addtocounter{thrm}{-1}%
   \begin{thrm}}
  {\end{thrm}}

  \newenvironment{thrmbisbis}[1]
  {\renewcommand{\thethrm}{\ref{#1}$''$}%
   \addtocounter{thrm}{-1}%
   \begin{thrm}}
  {\end{thrm}}

\title{The Bayes Principle and Segal Axioms for $P(\phi)_2$,\\
with application to Periodic Covers}
\author{Jiasheng Lin}
\date{}

\begin{document}

\maketitle

\abstract{
We construct a $P(\phi)_2$ Gibbs state on infinite volume periodic surfaces (namely, with discrete ``time translations'') by analogy with 1-dimensional spin chains and establish the mass gap for our Gibbs state, there are no phase transitions. We also derive asymptotic properties of the $P(\phi)_2$ partition function on certain towers of cyclic covers of large degrees that converge to the periodic surface in some appropriate sense. This gives the first construction of an interacting Quantum Field Theory on surfaces of infinite genus with a mass gap. 
The main ingredient in our approach is to reconcile the so-called $P(\phi)_2$ model from classical constructive quantum field theory (CQFT) with Riemannian version of the axioms proposed by G. Segal \cite{Segal} in the 90's. We show the $P(\phi)_2$ model satisfies these axioms, appropriately adjusted. One key ingredient in our proof is to use what we call ``the Bayes principle'' of conditional probabilities in the infinite dimensional setting. We also give a precise statement and full proof of the locality of the $P(\phi)_2$ interaction. 
}

\tableofcontents

\newpage

     \section{Introduction}

\noindent

\noindent In the classical approach to quantum field theory (QFT),
a central role is given to the representation of the Lorentz group as one can easily recognize if one looks at the Wightman axioms which were among the 
first attempt to axiomatize QFT.
It was later realized by Feynman, Symanzik \cite{Symanzik}, Nelson \cite{NelsonEuclidean} in the 60's that one could describe most objects of QFT using the functional integral, in imaginary time, which started the Euclidean approach to QFT. This allows to recover many computations and predictions of interest in quantum field theory.  

In the 80's, motivated by works of Atiyah and Witten on the relation of QFT with geometry and topology, there was an attempt to give an axiomatic definition of QFT based on geometry, that would capture the main properties of the
functional integral representation.
The two important notions at the foundation of QFT are the concept of \textbf{locality} and \textbf{unitarity}. The first notion of locality is intuitively captured by the functional integral representation whereas the second notion will be encoded in certain symmetries of our manifolds under reflection called reflection positivity. The notion of reflection posivity plays a central role in constructive quantum field theory since the works of Osterwalder and Schrader, Glimm and Jaffe~\cite{GJ} since it allows to recover relativistic Quantum Field Theory on Lorentzian manifolds from functional integrals and also plays an important role in statistical physics, for instance in the study of phase transitions.

The main idea behind the geometric axioms for QFT due to Atiyah, Segal \cite{Segal} and more recently Kontsevich and Segal \cite{KS21} can be summarized as follows: instead of viewing QFT as representation of some space-time symmetry group, we rather view them loosely as linear representations of some geometric bordism category.
In the following subsection we give 
the example of the 1 dimensional spin chain which contains the key phenomenology of the QFT model we will describe in the present work.  It can be understood as linear representations of some discrete geometric bordism category and we will also explain why our spin chain admits a unique Gibbs measure (in the sense of Ruelle \cite{Ruelle}) and the shift map acting on the chain is exponentially mixing for the Gibbs measure.

\subsection{Example of 1D Spin Chain and its Transfer Operator}
\label{sec-intro-spin-chain}

\noindent Our example illustrates the gluing properties by the transfer operator for discrete spin systems.
Consider $\mb{Z}_N:=\mathbb{Z}/N\mathbb{Z}$ as a $1$-dimensional chain of size $N$, consider the space $\mathbb{R}^N$ now as the space of maps $\mb{Z}_N \lto \mathbb{R}$, namely a discrete \textsf{path space}. The lattice site is denoted by $i\in \{1,\dots,N\}$ and call $\sigma\in\mb{R}^N$ a \textsf{spin configuration}, whose value at the site $i$ reads $\sigma(i)\in \mathbb{R}$. We impose \textsf{periodic boundary condition}, which means~$\sigma(N+1)\equiv \sigma(1)$. Thus the chain could be considered as circular.
We are given an action functional on the configuration space $\mb{R}^N$ that reads
\begin{eqnarray}
S_N(\sigma)\defeq\sum_{i=1}^{N} \vert \sigma(i+1)-\sigma(i)\vert^2 + \sum_{i=1}^N P(\sigma(i))
\label{eqn-spin-chain-inter}
\end{eqnarray}
where $P$ is a polynomial bounded from below, the interaction is nearest neighbour.
Given a configuration~$\sigma\in \mb{R}^N$,~$S_N(\sigma)$ may be thought of as its ``energy'' and the statistical behaviour of the system is described by the probability measure called \textsf{Gibbs measure},
\begin{eqnarray}
 \dd\mu_{P(\sigma)}(\sigma)\defeq \frac{1}{\mathcal{Z}(N)}\me^{-S_N(\sigma)} \dd^N\sigma, \quad\textrm{with}\quad
\mathcal{Z}(N)=\int_{\mathbb{R}^N} \me^{-S_N(\sigma)} \dd^N\sigma 
\label{eqn-spin-chain-gibbs-part-func}
\end{eqnarray}
called the \textsf{partition function} of the system.

Now we would like to express this partition function $\mathcal{Z}(N)$ in terms of elementary building blocks.
The main idea is to slice the action functional as
\begin{eqnarray*}
S_N(\sigma)=\sum_{i=1}^{N} \Big[ \vert \sigma(i+1)-\sigma(i)\vert^2 + \frac{1}{2}(P(\sigma(i+1))+P(\sigma(i))) \Big]
\end{eqnarray*} 
so exponentiating gives
\begin{eqnarray}
\exp\left(-S_N(\sigma) \right)=  \prod_{i=1}^{N} K(\sigma(i+1),\sigma(i)) 
\label{eqn-spin-chain-inter-expression}
\end{eqnarray}
where 
\begin{equation}
    K(x,y)= \me^{-\vert x-y\vert^2-\frac{1}{2}(P(x)+P(y))} 
    \label{eqn-spin-chain-trans-kernel}
\end{equation}
which is the Schwartz kernel of an operator on $L^2(\mathbb{R})$ that is smoothing.

\begin{deef}
  Define the \textsf{transfer operator} $T$ to be exactly the operator with kernel~$K(x,y)$, that is,
  \begin{equation}
    (TF)(x)=\int K(x,y)F(y) \dd y,
    \label{}
  \end{equation}
  for any function(al)~$F\in L^2(\mb{R})$.
\end{deef}

Then we see immediately from (\ref{eqn-spin-chain-inter-expression}) that
\begin{eqnarray}
  \mathcal{Z}(N)=\int_{}^{}\prod_{i=1}^{N} \big[K(\sigma(i+1),\sigma(i)) \dd\sigma(i) \big]= \ttr_{L^2(\mathbb{R})}( T^{N} ),
  \label{eqn-spin-chain-part-func-trace}
\end{eqnarray}
remembering that~$\sigma(N+1)\equiv \sigma(1)$, since for a smoothing operator~$A: L^2(\mb{R}) \lto \mathcal{S}(\mathbb{R})$ we have~$\ttr_{L^2(\mb{R})}(A)=\int_{}^{}K_A(x,x)\dd x$ with~$K_A$ being the integral kernel.

More generally we would like to express the kernel of~$T^N$ in terms of~$\exp(-S(\sigma))$ using the relation (\ref{eqn-spin-chain-inter-expression}). This means instead of letting the boundary condition be periodic we let $\sigma(1)=\sigma_{\mm{in}}$, $\sigma(N+1)=\sigma_{\mm{out}}$ given two boundary conditions $(\sigma_{\mm{in}},\sigma_{\mm{out}})\in \mathbb{R}^2 $. Then the kernel~$K_N$ of~$T^N$ is\footnote{Formula (\ref{eqn-intro-kernel-compo-transfer}) below is an analogue of what is called in the quantum information context the ``\textsf{replica trick}''.}
\begin{align}
  K_N(\sigma_{\mm{out}},\sigma_{\mm{in}})&=\int_{}^{}K(\sigma_{\mm{out}},\sigma(N))\cdots K(\sigma(2),\sigma_{\mm{in}}) \prod_{i=2}^N \dd\sigma(i) \label{eqn-intro-kernel-compo-transfer} \\
  &=\int_{\mb{R}^{N-1}}^{} \me^{-S_N(\sigma|\sigma_{\mm{in}},\sigma_{\mm{out}})} \prod_{i=2}^N \dd\sigma(i), \label{eqn-intro-kernel-path-integral}
\end{align}
with the \textsf{conditioned} interaction~$S_N(\sigma|\sigma_{\mm{in}},\sigma_{\mm{out}})$ defined as
\begin{equation}
  S_N(\sigma|\sigma_{\mm{in}},\sigma_{\mm{out}})\defeq \sum_{i=1}^{N} \vert \sigma(i+1)-\sigma(i)\vert^2 + \sum_{i=2}^N P(\sigma(i)) +\frac{1}{2}(P(\sigma_{\mm{in}})+P(\sigma_{\mm{out}})),
  \label{eqn-intro-interact-condition}
\end{equation}
with~$\sigma(1)\equiv \sigma_{\mm{in}}$ and~$\sigma(N+1)\equiv\sigma_{\mm{out}}$. Now if~$N_1$,~$N_2$ are two integers and we define the kernels~$K_{N_1}$,~$K_{N_2}$ using (\ref{eqn-intro-kernel-path-integral}) and (\ref{eqn-intro-interact-condition}) with~$N$ replaced respectively by~$N_1$,~$N_2$, then it follows ``trivially'' from the composition property~$T^{N_2}\circ T^{N_1}=T^{N_2+N_1}$ that
\begin{equation}
  K_{N_2+N_1}(\sigma_{\mm{out}},\sigma_{\mm{in}})=\int_{}^{}K_{N_2}(\sigma_{\mm{out}},\sigma)K_{N_1}(\sigma,\sigma_{\mm{in}}) \dd\sigma.
  \label{eqn-intro-composition}
\end{equation}

However, such a relation becomes remarkable (rather than trivial) if we do not have the ``unit'' transfer operator~$T$ to start with; that is, if we do not have (\ref{eqn-intro-kernel-compo-transfer}) but define~$K_{N_1}$ and~$K_{N_2}$ \textit{directly} with an expression of the form (\ref{eqn-intro-kernel-path-integral}) and (\ref{eqn-intro-interact-condition}). 
This corresponds to the idea of a \textsf{path integral} in quantum mechanics and quantum field theory. Alternatively one could consider the Gibbs measure (\ref{eqn-spin-chain-gibbs-part-func}) and take~$K_N$ as the \textit{transition probability} of a certain stochastic process. Then (\ref{eqn-intro-composition}) is the \textsf{Chapman-Kolmogorov} equation which relies heavily on the fact that the underlying process is \textsf{Markovian}. In a sense, for both interpretations a crucial condition is that the interaction~$S(\sigma)$ be \textsf{local}; that is, very roughly speaking, if one chops the sites~$[1,N_1+N_2]:=\{1,2,\dots,N_1+N_2\}$ into~$[1,N_1]\sqcup [N_1+1,N_1+N_2]$ then~$S_{[1,N_1+N_2]}(\sigma)\approx S_{[1,N_1]}(\sigma|_{[1,N_1]})+S_{[N_1+1,N_1+N_2]}(\sigma|_{[N_1+1,N_1+N_2]})$.

The main result of this article concerns a 2-dimensional ``continuum'' version of this story where lattice sites are replaced by the continuum of points on a 2D surface (considered as space-time) and a configuration is replaced by a distribution. See the section below for a more precise description. In the final section, we also show that when the space-time admits a periodic translation symmetry then a more precise analogy with the spin chain described above can be restored, in particular, there exists a \textsf{Gibbs state} in the \textsf{thermodynamic limit}. Further discussion of the above example continues in section \ref{sec-spin-chain-cont}.

\subsection{Main Results}

This article will prove three main results whose preliminary versions are stated as theorems \ref{thrm-intro-main-2}, \ref{thrm-intro-main-3} and \ref{thrm-intro-main-1} below, among which the first two are a consequence of the third, whose proof constitutes the most part of this paper.
The first theorem below generalises to the interacting cases, certain results of Naud~\cite{Nauddet} on the asymptotics of the free energy of free bosons on certain large degree covers.

\begin{thrm}[asymptotics of the partition function on large degree periodic covers]\label{thrm-intro-main-2}

 Let $(M,g)$ denotes a compact Riemannian surface, $M_\infty\lto M$ a Riemannian $\mathbb{Z}$-cover of $M$ and $\gamma$ the generator of the deck group.
We can think of $M_\infty$ as the infinite composition of some given cobordism $\Omega$, which serves as the ``fundamental domain'' of the deck transformations on $M_\infty$. 
For every $N\in \mb{N}$, denote by $M_N:=M_\infty/\gamma^N $ the cyclic cover of degree $N$ of $M$ and $\Delta_N$ is the Laplace-Beltrami operator on $M_N$.
We define the partition function of the $P(\phi)_2$ theory on $M_N$ heuristically as
\begin{equation}
  Z_N\heueq\int_{\mathcal{D}'(M_N)}^{} \me^{-\int_{M_N}^{}P(\phi(x))\dd V(x)}\me^{-\frac{1}{2}\int_{M_N}^{}(|\nabla \phi|_g^2 +m^2\phi^2 )\dd V}[\mathcal{L}\phi]
  \label{eqn-heu-part-func-cover}
\end{equation}
where $\dd V$ is the Riemannian area form, $\mathcal{D}'(M_N)$ is the space of real distributions on $M_N$ and $[\mathcal{L}\phi]$ is the non-existent Lebesgue measure on $\mathcal{D}'(M_N)$.

Then the renormalized sequence of free energies
$
\frac{1}{N}\log\left(Z_N \right)    
$ has a limit $\lambda_0$ when $N\rightarrow +\infty$, moreover this limit $\lambda_0$ can be interpreted as the leading eigenvalue of some transfer operator $U_\Omega$ which is the quantization of the cobordism $\Omega$ mentioned below.
\end{thrm}

Rigorous sense of the expression (\ref{eqn-heu-part-func-cover}) will be made in sections \ref{sec-gaus-fields} and \ref{sec-nelson-main}.

\begin{def7}
The above theorem can be thought of as an interacting QFT version of a result of Naud on the asymptotics of zeta determinants on large degree random covers of compact surfaces~\cite{Nauddet}, the main difference is that he treats random covers whereas our sequence of covers is deterministic and he only deals with the partition function of free fields whereas we treat the interacting case.
\end{def7}

In subsection~\ref{sec-perron-frob-gibbs}, we use the above $P(\phi)_2$ measures defined on the towers of cyclic covers $M_N$ to produce a $P(\phi)_2$-Gibbs state on the periodic surface $M_\infty$ of infinite volume. 
\begin{thrm}[Mass gap for $P(\phi)_2$ on $M_\infty$]\label{thrm-intro-main-3}
In the notations from Theorem~\ref{thrm-intro-main-2},
for observables $F$ with compact support on $M_\infty$, one can define expectations with respect to the $P(\phi)_2$ Gibbs state informally as:
\begin{eqnarray*}
\mathbb{E}\left[F \right]  \heueq \lim_{N\rightarrow +\infty} \frac{1}{Z_N} \int_{\mathcal{D}'(M_N)}^{} F(\phi)\me^{-\int_{M_N}^{}P(\phi(x))\dd V(x)}\me^{-\frac{1}{2}\int_{M_N}^{}(|\nabla \phi|_g^2 +m^2\phi^2 )\dd V}[\mathcal{L}\phi].  
\end{eqnarray*}
The generator of the deck group $\gamma$ acts on $M_\infty$ by diffeomorphism which induces a shift map $\tau$ on observables which moves the support by $\gamma$.
The Gibbs state defined above is exponentially mixing under the shift map, for compactly supported $L^2$ observables~:
\begin{eqnarray*}
\mathbb{E}[ \tau^kF G]=\mathbb{E}[ F ]\mathbb{E}[ G]+ \mathcal{O}(\alpha^k)    
\end{eqnarray*}
for some $\alpha<1$ (see corollary~\ref{cor-segal-trans-exp-decay}).
\end{thrm}

This means in physical terminology that our 
$P(\phi)_2$-Gibbs state has a \textbf{mass gap}.
These are immediate consequences of the constructions of theorem \ref{thrm-intro-main-1}
that we state below~:

\begin{thrm}
  [Segal Axioms for~$P(\phi)_2$, partial formal statement] \label{thrm-intro-main-1} For each finite disjoint union~$\Sigma$ of Riemannian circles (each component circle characterised by its radius) there exists a finite measure~$\mu_{\Sigma}$ on the space~$\mathcal{D}'(\Sigma)$ of real distributions over~$\Sigma$, and hence the Hilbert space~$\mathcal{H}_{\Sigma}:= L^2(\mathcal{D}'(\Sigma),\mu_{\Sigma})$, such that the following holds: if~$\Omega$ is a Riemannian surface whose boundary has two components~$\partial\Omega=\Sigma_{\mm{in}}\sqcup \Sigma_{\mm{out}}$, denote by~$\mu_{\mm{in}}$,~$\mu_{\mm{out}}$ respectively the measures on~$\mathcal{D}'(\Sigma_{\mm{in}})$,~$\mathcal{D}'(\Sigma_{\mm{out}})$ and put~$\mathcal{H}_{\mm{in}}:=L^2(\mathcal{D}'(\Sigma_{\mm{in}}),\mu_{\mm{in}})$,~$\mathcal{H}_{\mm{out}}:=L^2(\mathcal{D}'(\Sigma_{\mm{out}}),\mu_{\mm{out}})$, then there exists an operator~$U_{\Omega}:\mathcal{H}_{\mm{in}}\lto \mathcal{H}_{\mm{out}}$ given by
  \begin{equation}
    (U_{\Omega}F)(\varphi_{\mm{out}})=\int_{}^{}\mathcal{A}_{\Omega}(\varphi_{\mm{in}},\varphi_{\mm{out}}) F(\varphi_{\mm{in}})\,\dd\mu_{\mm{in}}(\varphi_{\mm{in}}),\quad F\in \mathcal{H}_{\mm{in}},
    \label{}
  \end{equation}
  where formally
\begin{equation}
  \mathcal{A}_{\Omega}(\varphi_{\mm{in}},\varphi_{\mm{out}})\heueq\int_{\{\phi|\partial\Omega =(\varphi_{\mm{in}},\varphi_{\mm{out}})\}}^{}\me^{-\int_{\Omega}^{}\frac{1}{2}(|\nabla \phi|_g^2 +m^2\phi^2 )+P(\phi(x))\dd V_{\Omega}(x)}[\mathcal{L}\phi],
  \label{eqn-heu-segal-amplitude}
\end{equation}
and the formal integration is over a ``subspace'' of~$\mathcal{D}'(\Omega)$, subject to the boundary condition~$\phi|_{\Sigma_{\mm{in}}}=\varphi_{\mm{in}}$ and~$\phi|_{\Sigma_{\mm{out}}}=\varphi_{\mm{out}}$, with respect to the (non-existent) Lebesgue measure~$[\mathcal{L}\phi]$ on the space of such distributions, with~$P$ being a polynomial bounded from below.

Moreover, the correspondence~$\Omega\longmapsto U_{\Omega}$ between surfaces and operators satisfy the following \textsf{composition/functorial property}. Let~$\Omega_1$,~$\Omega_2$ be two Riemannian surfaces of the kind as above where the ``out'' boundary of~$\Omega_1$ is isometric to the ``in'' boundary of~$\Omega_2$ via an isometry~$\rho$, we glue them along~$\rho$ and obtain the surface~$\Omega_2\cup_{\rho}\Omega_1$. Then
\begin{equation}
  U_{\Omega_2\cup_{\rho}\Omega_1}=U_{\Omega_2}\circ U_{\Omega_1},
  \label{eqn-intro-state-compo}
\end{equation}
for the operators~$U_{\Omega_2\cup_{\rho}\Omega_1}$,~$U_{\Omega_2}$ and~$U_{\Omega_1}$ defined as above.
\end{thrm}

\begin{def7}
By a \textsf{real distribution} $\phi$ we mean that if we pair~$\phi$ with a potentially complex test function~$f$, then
  \begin{equation}
    \ol{\phi(f)}=\phi(\ol{f}),
    \label{}
  \end{equation}
  the bar denoting complex conjugation. 
\end{def7}

A lot of technicalities go into the full precise description of the above statement, see subsection \ref{sec-segal-descript} for details and in particular theorem {\color{blue} 1$''$} for the final precise statement. In order to emphasise the main difficulties dealt with by this article regarding the above theorem, let us restate it as follows.

\begin{thrmbis}{thrm-intro-main-1}
  A rigorous definition can be given to the formal expression (\ref{eqn-heu-segal-amplitude}) along the lines of the classical construction of the~$P(\phi)_2$ measure \`a la Nelson \cite{Nel66}, and out of this definition, the property (\ref{eqn-intro-state-compo}), along with others required by Segal's axioms (see definition \ref{def-segal-2}), can be proved.
\end{thrmbis}

\subsection{Main Novelties, Comments, Organisation}

We would like to point out \textbf{three main novelties} of this article. Firstly, the relationship between theorems \ref{thrm-intro-main-2}, \ref{thrm-intro-main-3} and \ref{thrm-intro-main-1} above offers a new perspective that certain asymptotic results in pure geometry can be viewed as a consequence of an underlying Segal Quantum Field Theory whose ``transfer operators''~$U_{\Omega}$ satisfy certain nice properties. Besides, of course, theorem \ref{thrm-intro-main-2} itself is new in that the numbers~$Z_N$ are now partition functions of \textit{interacting} field theories instead of regularized determinants of geometric operators, which were studied previously and in our setting corresponds to free field theories. Moreover, in the process of proving theorem \ref{thrm-intro-main-2} one also obtains an infinite-volume~$P(\phi)_2$-theory (in one direction) over a curved periodic surface which has never been considered before in the literature. Moreover, the Segal view point gives a simple conceptual proof of the mass gap for the $P(\phi)_2$ state for any polynomial $P$ of degree $\geqslant 4$ bounded from below on a space of infinite volume, infinite genus. Theorem \ref{thrm-intro-main-3} seems to be one of the first mass gap results for interacting QFT on curved spaces.

A key difficulty in proving theorem \ref{thrm-intro-main-1} goes into proving (\ref{eqn-intro-state-compo}) for the free theory, namely the case with~$P=0$, and the proof for~$P\ne 0$ builds on that case. Although the same result as theorem \ref{thrm-intro-main-1} was to an extent obtained by Pickrell in \cite{Pickrell}, this article offers a novel viewpoint based on a simple conceptual idea involving a symmetry in the successive conditioning of random variables (see remark \ref{rem-bayes-descrip} for a heuristic description), whereas \cite{Pickrell} proceeded by an analogy with the finite dimensional case where computations are done using the theory of \textsf{oscillator semigroups}. A parallel treatment of the same difficulty also appears in Guillarmou-Kupiainen-Rhodes-Vargas \cite{GKRV} where the method is based on explicit computations of Dirichlet forms, yet different from both this article and \cite{Pickrell}.

To go from~$P=0$ to~$P\ne 0$, a key ingredient is to show that the~$P(\phi)$-interaction functional of the field ---~$\int_{M}^{}P(\phi(x))\,\dd x$ heuristically --- is \textit{local}, as alluded to heuristically at the end of subsection \ref{sec-intro-spin-chain}, which in fact enables the amplitude (\ref{eqn-heu-segal-amplitude}) to be defined rigorously in the first place. No precise description of this locality existed in the literature and no explicit construction was given to prove it except for some remarks in~\cite{Abdesselam, Duch}. The present article fills this gap in subsection \ref{sec-locality}, based on a strengthening of Nelson's argument presented in subsection \ref{sec-nelson-reg-indep} that allows a freedom in choosing the regulators in the renormalisation process, producing the same measure in the limit. This latter result is also new by itself.

\paragraph{Dirichlet-to-Neumann operators and reflection positivity.} A fundamental fact used repeatedly in this article is that the induced probability law of the Gaussian Free Field (GFF) on a closed manifold upon restricting onto a hypersurface is given by the inverse of the (jumpy) Dirichlet-to-Neumann operator. There exists several ways to see this and we offer a new approach based on an elementary relation between the operators~$\tau_{\Sigma}$,~$\Delta_M+m^2$,~$\PI_M^{\Sigma}$ and~$\DN_M^{\Sigma}$ (see ``notations'' below), namely
\begin{equation}
  j_{\Sigma}\defeq \tau_{\Sigma}^*=(\Delta_M+m^2)\PI_M^{\Sigma}(\DN_M^{\Sigma})^{-1},
  \label{}
\end{equation}
where~``$*$'' denotes the distributional adjoint. Though elementary, such a relation was never written down elsewhere and here we derive it directly from the Green-Stokes formula (see subsection \ref{sec-green-stokes}). Later, we also put forward another elementary relation
 \begin{equation}
    \PI_{\Omega}^{\partial\Omega}(\DN_{\Omega}^{\partial\Omega})^{-1}(\PI_{\Omega}^{\partial\Omega})^*=C_N-C_D,
    \label{}
  \end{equation}
  where~$C_D$,~$C_N$ respectively are the Green operators (of the massive Laplacian) with Dirichlet and Neumann conditions. This gives a \textbf{direct relation} between the positivity of the Dirichlet-to-Neumann map and the reflection positivity (RP) of the GFF on closed Riemannian manifolds with reflection symmetry (see subsection \ref{sec-RP-first}). This intimate relation between Dirichlet-to-Neumann operators and RP was never explicitly noticed before.

\paragraph{Organisation.} In section \ref{sec-gaus-fields} we recall the preliminaries on which our construction will be based; these include the definition and properties of the (massive) Gaussian Free Field (GFF) measure on $\mathcal{D}'(M)$ (subsection \ref{sec-mas-GFF}), the functional determinants (subsection \ref{sec-det}) and a formula computing the Radon-Nikodym densities between mutually absolutely continuous Gaussian measures on spaces of distributions (subsection \ref{sec-quad-pert-rad-niko-main}), of which we include a full proof based on the method of \cite{GJ} section 9.3. 

In section \ref{sec-nelson-main} we retrace the classical argument of Nelson \cite{Nel66} and construct the interacting $P(\phi)$ functional measure on $\mathcal{D}'(M)$ for $M$ a closed Riemannian surface, giving the partition function $Z_M$; the new result which says the interaction functional (and hence the limiting measure) is independent of the regulator for the renormalisation process, when the regulator is picked from a specific class, is proposition \ref{prop-nelson-main}.

Section \ref{sec-trace-poisson} aims to derive the behavior of a Gaussian field under the trace map (restriction to hypersurface); our method is based on elementary relations between several geometric operators as mentioned above (lemma \ref{lemm-DN-trick}), and as above it leads to the somewhat new perspective on reflection positivity (subsection \ref{sec-RP-first}). Section \ref{sec-markov-main} discusses the Markov decomposition of the GFF (subsection \ref{sec-markov-main}) that culminates in the \textit{Bayes principle} (see remark \ref{rem-bayes-descrip} and proposition \ref{prop-bayes-gff}) for the GFF and proving locality of the $P(\phi)$ interaction (subsection \ref{sec-locality}). 

Section \ref{sec-segal-main} formulates precisely and proves theorem \ref{thrm-intro-main-1} while section \ref{sec-period-cover} does it for theorem \ref{thrm-intro-main-2} and \ref{thrm-intro-main-3}.

\begin{def7}[Bayes principle]\label{rem-bayes-descrip}
Here we explain in briefly the new idea that proves Segal's gluing (composition) axiom in the case $P=0$. In the simplest form it amounts to the following observation: suppose two real random variables~$X$ and~$Y$ have joint probability law~$\mb{P}_{(X,Y)}$; then we have two equal expressions for their joint density~$p(x,y)=(\dd\mb{P}_{(X,Y)}/\dd \mathcal{L}_{\mb{R}}^{\otimes 2}) (x,y)$, namely
\begin{equation}
  p(x,y)=\frac{\dd(\pi^x_*\mb{P}_{(X,Y)})}{\dd \mathcal{L}_{\mb{R}}}(x)\frac{\dd \mb{P}_{Y|X=x}}{\dd \mathcal{L}_{\mb{R}}}(y)=\frac{\dd(\pi^y_*\mb{P}_{(X,Y)})}{\dd \mathcal{L}_{\mb{R}}}(y)\frac{\dd \mb{P}_{X|Y=y}}{\dd \mathcal{L}_{\mb{R}}}(x),
  \label{eqn-bayes-ordinary}
\end{equation}
where~$\pi^x$,~$\pi^y$ are respectively the projections onto the~$x$- and~$y$-axes,~$\mb{P}_{Y|X=x}$ denotes the conditional law of~$Y$ knowing ``$X=x$'', and vice versa for~$\mb{P}_{X|Y=y}$ (more rigorously these are expressed using \textit{transition kernels}). In other words, one could evaluate~$p(x,y)$ by \textit{conditioning} in two alternative ways: on~$X$ or on~$Y$. In what follows we put forward a version of (\ref{eqn-bayes-ordinary}) for the Gaussian Free Field (GFF) measure on~$\mathcal{D}'(M)$ for~$M$ a closed Riemannian surface, where the projections~$\pi^{x,y}$ are replaced by trace maps (restrictions) onto embedded Riemannian circles. By the spacial Markov property of the GFF, the expressions for the conditional laws simplify nicely. Since the transition amplitudes (\ref{eqn-heu-segal-amplitude}) are given by the square roots of the densities of the trace-induced measures with respect to fixed background Gaussian functional measures on the circles (corresponding to~$p(x,y)$ above), we obtain a proof of Segal's gluing axioms for the GFF and, by the \textit{locality} of the~$P(\phi)$ interaction, the result extends immediately to the interacting case for this interaction. 
\end{def7}

\begin{def7}[role of symmetries and time translation]\label{rem-sym-time-trans}
  One primary aim of the present work is to reconcile the model (\ref{eqn-def-mes-gibbs-heu}) with the axiomatics proposed by G. Segal \cite{Segal}. However, a significant difference between our model and the theories \cite{Segal} principally tries to target is that we do not have and consider symmetries beyond the simplest one of reflecting across a hypersurface. In particular the model (\ref{eqn-def-mes-gibbs-heu}) will not be conformally invariant, and not even the free case ($P=0$) as~$m>0$ and we do not include any ``conformal anomaly'' accounting effects of a conformal change in the metric. A linear term involving the Gauss curvature could have been included in the exponential of (\ref{eqn-heu-segal-amplitude}) --- our situation corresponds to the case of ``minimal coupling'', see discussion of Jaffe and Ritter \cite{JR} around equation (7).
  
  In fact, we lack a basic symmetry which is present in the Euclidean ($\mb{R}^2$) case of the same model (\ref{eqn-def-mes-gibbs-heu}): since we deal with general compact surfaces, there are no distinguished ``time translations'' (in the form of a generating Killing vector field). In the Euclidean case, the ``canonical'' time direction helps to fix a definitive Hilbert space of quantum states, and eventually allows one to obtain an explicit expression of the Hamiltonian operator generating time evolution of the theory. This would also be the case if $\Omega$ is a cylinder equipped with cylindrical metric, for example, but for general $\Omega$ the notion of a ``Hamiltonian operator'' is not very well-defined. It is nevertheless interesting to think of a definition perhaps in terms of categorical limits as in Kontsevich and Segal \cite{KS21} section 3. Basic constructions of a Euclidean (Riemannian) field theory on curved space-times equipped with a time translation have been considered in Jaffe and Ritter \cite{JR2}.

  Last but not least, in the last section \ref{sec-period-cover} of this article we explore a crucial property of the operator $U_{\Omega}$ representing the ``evolution'' across the cobordism $\Omega$ (more suggestively it could be written as $\me^{-[\Omega]H}$), namely that it is \textsf{Perron-Frobenius}. This would imply that it nevertheless has many properties in common with an actual time evolution (i.e. $\me^{-t H}$), especially it has a \textsf{spectral gap}.
\end{def7}

\begin{def7}
  In hindsight, the author feels that a proof following (\ref{eqn-bayes-ordinary}) is a natural one regarding the intuition behind the path integral as ``summing over histories'', and (\ref{eqn-bayes-ordinary}) simply says that such a sum ``conditioned'' on different intermediate points in history\footnote{In our case, history is circular!} produce equivalent results (the proof of (\ref{eqn-bayes-ordinary}) itself is trivial!). However, a still more ideal version of the proof would perhaps be one that puts the free field and interacting field on equal footing, where Segal's axioms follow from a generalized version of ``Nelson's axioms'' as explained on section IV.1 of Simon \cite{Sim2}, based on the Markov property and uses a ``nonlinear version'' of the transition operator~$\mathcal{M}_{M,2}^1$ (see section \ref{sec-bayes-gff}).
\end{def7}

\subsection{Notations}

\noindent In this paper, unless stated otherwise,

``$\heueq$'' means ``heuristically equal to'';

$(M,g)=$ closed Riemannian surface with metric~$g$;

$(\Omega,g)=$ compact Riemannian surface with metric~$g$, boundary~$\partial\Omega$, seen as isometrically embedded in~$M$, the induced metric on~$\partial\Omega$ is still denoted~$g$;

$\Sigma =$ disjoint union of Riemannian circles (of different radii), usually~$\Sigma =\partial\Omega$; hypersurface in $M$;

$\Omega^{\circ}=$ interior of $\Omega$;

$m=$ mass parameter,~$m>0$, \textit{fixed throughout paper};

$\mathcal{D}'(M)$,~$\mathcal{D}'(\Omega^{\circ})=$ \textsf{real} distributions on~$M$ and~$\Omega$;

$C^{\infty}(M)$,~$C_c^{\infty}(\Omega^{\circ})=$ smooth functions on~$M$, smooth compact support functions on~$\Omega^{\circ}$;

$\Delta=$ Laplacian on~$(M,g)$, $\Delta f:=-\ddiv (\nabla f)$, therefore it is defined to be \textsf{nonnegative}, the same applies below;

$\Delta_{\Sigma}=$ Laplacian on~$\Sigma$;

$\Delta_{\Omega,D}=$ Laplacian on~$\Omega$ with (zero) Dirichlet boundary condition (see also remark \ref{rem-Dir-Neu-cond-mean});

$\mn{D}_{\Sigma}=(\Delta_{\Sigma}+m^2)^{1/2}$;

$\tau_{\Sigma}=$ trace operator (restriction) onto the hypersurface $\Sigma$;

$(Q,\mathcal{O})=$ general probability sample space with~$\sigma$-algebra~$\mathcal{O}$;

$C=$ general positive self-adjoint elliptic pseudodifferential operator on~$M$,~$\Omega$,~$\Sigma$, which is Hilbert-Schmidt on the corresponding~$L^2$ spaces;

$\mu_C^{\mm{cond}}=\mu_{C^{-1}}^{\mm{cond}}=$ Gaussian measure on~$\mathcal{D}'(M)$,~$\mathcal{D}'(\Omega^{\circ})$, or~$\mathcal{D}'(\Sigma)$, equipped with their Fr\'echet Borel~$\sigma$-algebra, with covariance~$\mb{E}_C[\phi(f)\phi(h)]=\ank{f,C h}_{L^2}$, under some conditions;

$\mu_{\mm{GFF}}^M=$ massive GFF measure on $\mathcal{D}'(M)$; $\mu_{\mm{GFF}}^{\Omega,D}=$ massive Dirichlet GFF measure on $\mathcal{D}'(\Omega^{\circ})$;

$\mb{E}_{B}^A=$ expectation under $\mu_{B}^A$;

$\ank{-,-}_{L^2}=L^2$-inner product, pairing between $\mathcal{D}'(M)$ and $C^{\infty}(M)$, or between $W^s(M)$ and $W^{-s}(M)$, or any other pairing which is an extension of the $L^2$-inner product;

$\phi(f)=$ the random variable~$\phi\mapsto \ank{\phi,f}_{L^2}$ indexed by~$f\in C^{\infty}(M)$,~$C_c^{\infty}(\Omega^{\circ})$ or~$C^{\infty}(\Sigma)$;

$W^s(M)=$ the~$L^2$ Sobolev space on~$M$, with inner product~$\ank{-,-}_{W^s(M)}:=\ank{-,(\Delta+m^2)^{s} -}_{L^2}$;

$W^s_{A}(M)$, $W^s_U(M)$, $W^s(U)=$ see appendix \ref{sec-app-sobo};

$p_{M\setminus A}^{\perp}$,~$p_{M\setminus A}$,~$P_{ A}$, and~$P_{ A}^{\perp}=$ Sobolev projections defined in lemma \ref{lemm-sobo-decomp};

$\Psi^r(M)$, $\Psi^r(\Sigma)=$ pseudodifferential operators ($\Psi$DOs) on~$M$ or $\Sigma$ with order~$r$;

$[\mathcal{L}\phi]=$ hypothetical Lebesgue measure on a space of distributions;

$\nrm{\cdot}_{\mm{tr}}$, $\nrm{\cdot}_{\mm{HS}}$, $\nrm{\cdot}_{L^2}=$ trace norm, Hilbert-Schmidt norm, operator norm acting on $L^2$ or $L^2$-norm on function;

$\PI_{\Omega}^{\Sigma,B}=$ Poisson integral operator extending from~$\Sigma$ to~$\Omega$ (definition \ref{def-pi-more-general}), with boundary conditions~$B$ imposed on boundary components \textit{other than}~$\Sigma$ (see also remark \ref{rem-Dir-Neu-cond-mean});

$\DN_{\Omega}^{\Sigma,B}=$ Dirichlet-to-Neumann operator or the jumpy version defined using~$\PI_{\Omega}^{\Sigma,B}$ (jumpy version is understood when~$\Sigma$ is in the interior of~$\Omega$ rather than a boundary component);

$\PI_{M}^{\Sigma}$, $\DN_{M}^{\Sigma}=$ Poisson integral operator extending from~$\Sigma$ to~$M$ (see (\ref{eqn-def-pi-emb-hyp-in-closed-case})), and jumpy Dirichlet-to-Neumann operator defined using $\PI_M^{\Sigma}$ (definition \ref{def-jp-DN-map}).

$\mu_{\DN}^{\Sigma,\Omega,B}=$ Gaussian measure on~$\mathcal{D}'(\Sigma)$ with covariance~$(\DN_{\Omega}^{\Sigma,B})^{-1}$;

$\mu_{2\DN}^{\partial\Omega,\Omega}=$ Gaussian measure on~$\mathcal{D}'(\partial\Omega)$ with covariance~$\frac{1}{2}(\DN_{\Omega}^{\partial\Omega})^{-1}$;

$\mu_{2\mn{D}}^{\Sigma}=$ Gaussian measure on~$\mathcal{D}'(\Sigma)$ with covariance~$\frac{1}{2}(\mn{D}_{\Sigma})^{-1}$;

$\mathcal{M}_{M,2}^1= \tau_{\Sigma_2}\PI_{M}^{\Sigma_1}$ is the \textsf{transition operator} defined in section \ref{sec-bayes-gff}.

  \subsection{Acknowledgement}

 \noindent The author thanks foremostly his thesis advisor Nguyen-Viet Dang for raising attention to the work~\cite{Pickrell} which results in this article and many discussions, especially in suggesting the use of the regulator in Dyatlov and Zworski \cite{DZ} and pointing out the wonderful analogy with spin chain treated in the final section. He also thanks Colin Guillarmou and Tat-Dat Tô for helpful discussions at an early stage of the project, Douglas Pickrell for giving valuable comments on a preliminary version of the work, and Sylvie Paycha and Bin Zhang for providing related reference and interesting discussion following a talk given by the author on this work. Finally, he thanks his laboratory IMJ-PRG for supporting his research.

  \section{Gaussian Fields on a Riemannian Manifold}
  \label{sec-gaus-fields}

\noindent \textbf{In this and the next section} we define rigorously a probability measure on~$\mathcal{D}'(M)$, as well as certain variants on~$\mathcal{D}'(\Omega^{\circ})$, which heuristically bears the form
\begin{equation}
  \dd\mu_P(\phi)\heueq\frac{1}{Z}\overbrace{\me^{-\int_{M}^{}P(\phi(x))\dd V_M (x)}}^{A} \underbrace{\me^{-\frac{1}{2}\int_{M}^{}(|\nabla \phi|_g^2 +m^2\phi^2 )\dd V_M}[\mathcal{L}\phi]}_{B},
  \label{eqn-def-mes-gibbs-heu}
\end{equation}
where~$P$ is a polynomial bounded from below, such as $P(\phi)=\phi^4-\phi^2$,~$[\mathcal{L}\phi]$ denotes the nonexistent Lebesgue measure on~$\mathcal{D}'(M)$ and
\begin{equation}
  Z_M\heueq\int_{\mathcal{D}'(M)}^{} \me^{-\int_{M}^{}P(\phi(x))\dd V_M(x)}\me^{-\frac{1}{2}\int_{M}^{}(|\nabla \phi|_g^2 +m^2\phi^2 )\dd V_M}[\mathcal{L}\phi]
  \label{}
\end{equation}
is the normalization factor, also called the \textsf{partition function}. 

The idea is that while~$[\mathcal{L}\phi]$ is nonexistent, together with the factor~$\exp(-\frac{1}{2}\int_{M}^{}(|\nabla \phi|_g^2 +m^2\phi^2 )\dd V_M)$ the expression~$B$ can be given a rigorous meaning as a Gaussian probability measure scaled by a (finite) volume constant, and the measure (\ref{eqn-def-mes-gibbs-heu}) can be constructed if, after defining part $A$ rigorously, one proves that it is~$L^1$ with respect to the measure~$B$. 

\textbf{In this section} we construct the measure $B$. In analogy with the expression of Gaussian measures on $\mb{R}^n$, one should define heuristically
\begin{equation}
  \me^{-\frac{1}{2}\int_{M}^{}(|\nabla \phi|_g^2 +m^2\phi^2 )\dd V_M}[\mathcal{L}\phi] \heueq \textrm{``}\det\textrm{''}(\Delta+m^2)^{-\frac{1}{2}} \dd\mu_{\mm{GFF}}^M (\phi)
  \label{eqn-heu-gaussian-mea-det-volume}
\end{equation}
where~$\mu_{\mm{GFF}}^M$ is a Gaussian measure on~$\mathcal{D}'(M)$ with \textit{covariance operator}~$(\Delta+m^2)^{-1}$, and ``$\det$'' is an infinite dimensional generalization of the determinant of a matrix. The Gaussian measure and the determinant are two issues to be treated separately (in sections \ref{sec-mas-GFF} and \ref{sec-det}), both being rather classical, the former called \textsf{Gaussian Free Field (measure)} and the latter called the \textsf{$\zeta$-regularized determinant}.

  \label{sec3.1}

 \subsection{The Massive Gaussian Free Field}\label{sec-mas-GFF}

 \subsubsection{Definition and Representations}
In this subsection we explain the point of view adopted in this article of the massive Gaussian Free Field (GFF), and refer to Sheffield \cite{Shef} and Powell and Werner \cite{PW} for more information. 

Let~$(M,g)$ be a closed Riemannian manifold \textsf{of dimension~$d$} with metric~$g$, and~$\Omega \subset M$ an open domain with smooth boundary~$\partial\Omega$, both equipped with the metric induced from~$g$ (same notation). Fix~$m>0$ as the mass parameter. We say that
  the \textsf{massive Gaussian Free Field (GFF)} with mass~$m$ on~$M$ is the \textit{Gaussian random process indexed by}~$C^{\infty}(M)$, consisting of random variables~$\{\phi(f)~|~f\in C^{\infty}(M)\}$ such that
  \begin{equation}
    \mb{E}[\phi(f)\phi(h)]=\ank{f,(\Delta+m^2)^{-1}h}_{L^2(M)},\quad \mb{E}[\phi(f)]\equiv 0,
    \label{eqn-gff-cov-closed}
  \end{equation}
  for any~$f$,~$h\in C^{\infty}(M)$. Similarly, the \textsf{Dirichlet massive Gaussian Free Field} with mass~$m$ on~$\Omega$ is the \textit{Gaussian random process indexed by}~$C_c^{\infty}(\Omega^{\circ})$, consisting of random variables~$\{\phi(f)~|~f\in C_c^{\infty}(\Omega^{\circ})\}$ such that
  \begin{equation}
    \mb{E}[\phi(f)\phi(h)]=\bank{f,(\Delta_{\Omega,D}+m^2)^{-1}h}_{L^2(\Omega)}=\bank{P_{M\setminus \Omega^{\circ}}^{\perp}f, P_{M\setminus \Omega^{\circ}}^{\perp}h}_{W^{-1}},\quad \mb{E}[\phi(f)]\equiv 0,
    \label{}
  \end{equation}
  for any~$f$,~$h\in C_c^{\infty}(\Omega^{\circ})$. See appendix \ref{sec-app-sobo} and in particular lemma \ref{lemm-diri-green-op-quad-form}.

There exist many choices of sample spaces (``$Q$-spaces'' in the terminology of \cite{Sim2}) on which to realize those Gaussian processes. These realizations are all equivalent in the sense of \textsf{isomorphism of measure algebras} (see \cite{Sim2} section I.2). For the sake of concreteness, we point out that one choice for the sample space is~$\mathcal{D}'(M)$ (or~$\mathcal{D}'(\Omega^{\circ})$), as formulated in the following proposition, whose proof parallels the case on Euclidean space with minor modification.

\begin{prop}
  [Bochner-Minlos, \cite{BHL} theorem 5.11, page 266] \label{prop-boch-min} There exists a Borel probability measure~$\mu_{\mm{GFF}}^M$ on the Fr\'echet space~$\mathcal{D}'(M)$ such that~$\phi\mapsto \ank{\phi,f}_{L^2(M)}=:\phi(f)$,~$f\in C^{\infty}(M)$, realizes the random variable~$\phi(f)$ of the massive GFF on~$M$. 
  
  Similarly, there exists a Borel probability measure~$\mu_{\mm{GFF}}^{\Omega,D}$ on the Fr\'echet space~$\mathcal{D}'(\Omega^{\circ})$ such that~$\phi\mapsto \ank{\phi,f}_{L^2(\Omega)}$,~$f\in C_c^{\infty}(\Omega^{\circ})$, realizes the random variable~$\phi(f)$ of the Dirichlet massive GFF on~$\Omega$.\hfill~$\Box$
\end{prop}

\begin{notation}
  Denote by~$\mb{E}_{\mm{GFF}}^M$ and~$\mb{E}_{\mm{GFF}}^{\Omega,D}$ the expectations under~$\mu_{\mm{GFF}}^M$ and~$\mu_{\mm{GFF}}^{\Omega,D}$ respectively.
\end{notation}

\begin{def7}\label{rem-mass-gau-field-four}
  Alternatively, based on the spectral theory of~$\Delta$, let~$\left\{ \varphi_j \right\}_{j=0}^{\infty}$ be its complete orthonormal eigenfunctions with (real nonnegative) eigenvalues~$\left\{ \lambda_j \right\}_{j=0}^{\infty}$,~$0=\lambda_0<\lambda_1\le \lambda_2\le \dots\le \lambda_j\le \cdots$, counted with multiplicity. Then the \textsf{GFF with mass~$m$} on~$M$ could also be represented as the random formal series
  \begin{equation}
    \phi=\sum_{j=0}^{\infty}\xi_j \varphi_j
    \label{eqn3.3}
  \end{equation}
  where the sequence~$(\xi_j)_j$ consists of i.i.d.\ real-valued random variables with $\xi_j$ being the standard centered Gaussian on~$\mb{R}$ with variance~$(\lambda_j+m^2)^{-1}$. Similarly, the Dirichlet GFF on~$\Omega$ could also be so represented using eigenfunctions of~$\Delta_{\Omega,D}$, which are complete for~$L^2(\Omega)$. One way to see that this construction is equivalent to the previous one is by appealing to \cite{Shu} page 92 proposition 10.2.
\end{def7}

We also consider a slightly more general situation  where~$(\Delta+m^2)^{-1}$ is replaced by an operator $C$.

\begin{deef}
  We say that a bounded self-adjoint positive elliptic pseudodifferential operator~$C$ of order~$-s$ on~$M$ ($s>0$) is a \textsf{Gaussian covariance operator} of order~$-s$.
\end{deef}

Following similar reasoning as proposition \ref{prop-boch-min} or remark \ref{rem-mass-gau-field-four}, one obtains a measure~$\mu_C$ on~$\mathcal{D}'(M)$.

\begin{deef}
  The \textsf{Gaussian Field} on~$M$ with \textsf{covariance operator}~$C$ is the \textit{Gaussian random process indexed by}~$C^{\infty}(M)$, consisting of random variables~$\{\phi(f)~|~f\in C^{\infty}(M)\}$ such that
  \begin{equation}
    \mb{E}[\phi(f)\phi(h)]=\ank{f,C h}_{L^2(M)},\quad \mb{E}[\phi(f)]\equiv 0,
    \label{}
  \end{equation}
  for any~$f$,~$h\in C^{\infty}(M)$. We denote the corresponding measure on~$\mathcal{D}'(M)$ by~$\mu_C$ and the expectation with respect to this measure by~$\mb{E}_C$.
\end{deef}

\begin{def7}
  For~$C$ satisfying the assumptions, the inner product~$(f,h)\mapsto \ank{f, Ch}_{L^2(M)}$ defines an equivalent norm for~$W^{-s}(M)$.
\end{def7}

\subsubsection{Essential Properties}

Now we collect some properties of the Gaussian random fields constructed above. We start with the following classical fact (the proofs are in parallel with the Euclidean case for which one could refer to \cite{Shef}).
\begin{lemm}\label{lemm-ghs-for-gff}
  The \textsf{Gaussian Hilbert space} of~$\mu_{\mm{GFF}}^M$ is~$W^{-1}(M)$, and that of~$\mu_{\mm{GFF}}^{\Omega,D}$ is~$W^{-1}(\Omega)$. The \textsf{Cameron-Martin space} of~$\mu_{\mm{GFF}}^M$ is~$W^{1}(M)$, and that of~$\mu_{\mm{GFF}}^{\Omega,D}$ is~$W^{1}_{\Omega}(M)$.\hfill~$\Box$
\end{lemm}

\begin{def7}\label{rem-use-cam-mar-pairing-not-l2}
  There is a tacit assumption in the way we defined our fields: we expected the random variable~$\phi(f)$ to come from the distributional pairing ($L^2$-pairing) between~$\mathcal{D}'(M)$ and~$C^{\infty}(M)$ (respectively,~$\Omega$). Other pairings may also be used. For example, let~$\phi(f)$ be~$\ank{-,f}_{W^1(M)}$ instead of~$\ank{-,f}_{L^2(M)}$. One then needs to alter the covariances accordingly. Indeed, they are related by
  \begin{equation}
    \mb{E}\big[\bank{\phi,f}_{W^1}\bank{\phi,h}_{W^1}\big]=\mb{E}\big[
    \bank{\phi,(\Delta+m^2)f}_{L^2}\bank{\phi,(\Delta+m^2)h}_{L^2}\big]=\bank{f,h}_{W^1(M)}.
    \label{}
  \end{equation}
  This way of definition is noticeably used by Sheffield \cite{Shef}. The Gaussian Hilbert space in this case is~$W^1(M)$ (respectively,~$W^1_{\Omega}(M)$), and the Cameron-Martin spaces are the same. See also remark \ref{rem-natural-sobo-cam-mar}.
\end{def7}

Similarly,

\begin{lemm}\label{lemm-GHS-cam-mar-of-cov-C}
  The Gaussian Hilbert space of~$\mu_C$ is~$W^{-s}(M)$, equipped with~$\ank{-,C-}_{L^2}$, and the Cameron-Martin space is~$W^s(M)$, equipped with~$\sank{-,C^{-1}-}_{L^2}$. \hfill~$\Box$
\end{lemm}

Finally we say a word about the supports of the measures in the closed manifold case. This can be proved either using lemma \ref{lemm-quadratic-pert-trace} and evaluating the expectation of $\snrm{\phi}_{W^s(M)}^2$ or a spectral representation like in remark \ref{rem-mass-gau-field-four}.

\begin{lemm}\label{lemm-supp-gaus-meas}
  We have~$\mu_C(W^{-\delta}(M))=1$ for any~$\delta>\frac{1}{2}(d-s)$.\hfill~$\Box$
\end{lemm}

\begin{def7}
  We point out that~$\mu_C(W^{-\delta}(M))=1$ for any~$\delta>0$ in the following two cases:
  \begin{enumerate}[(i)]
    \item $\dim M=2$ and~$\mu_C=\mu_{\mm{GFF}}^M$;
    \item $\dim M=1$ and~$C$ has order~$-1$.
  \end{enumerate}
\end{def7}

Last but not least, we make the following innocent but useful observation.

\begin{lemm}\label{lemm-gaus-field-disj-indep}
  If~$\Omega=\Omega_1\sqcup \Omega_2$ (possibility of non-empty boundary in either or both components), then GFFs (indeed, Gaussian fields) over~$\Omega_1$ and~$\Omega_2$ are independent and~$\mu_{\mm{GFF}}^{\Omega,B}=\mu_{\mm{GFF}}^{\Omega_1,B}\otimes \mu_{\mm{GFF}}^{\Omega_2,B}$ where~$B=D$ when the corresponding~$\Omega$,~$\Omega_1$ or~$\Omega_2$ has boundary and~$B=\varnothing$ when either is closed. \hfill~$\Box$
\end{lemm}

 \subsection{Determinants}\label{sec-det}

 \noindent In this section we discuss (two) generalizations of the notion of the determinant (of a matrix) to infinite dimensional operators. General references include Kontsevich and Vishik \cite{KV}, Shubin \cite{Shu} sections 9-13, Simon \cite{Sim1} chapter 3 and finally Gohberg, Goldberg and Krupnik \cite{GGK}. See also Dang \cite{Dang} for a quick acquaintance of the physical-geometric context and Quine, Heydari and Song \cite{QHS} for an interesting discussion of zeta-regularization of infinite products.

\subsubsection{Zeta-regularized and Fredholm Determinants}\label{sec-det-def}

\noindent The zeta-regularized determinant was first introduced by Ray and Singer \cite{RS}. The first step is to define the \textsf{zeta function} of a (rather special) pseudodifferential operator~$A$ over a manifold $M$ with or without boundary,
\begin{equation}
  \zeta_A(z)\defeq \ttr_{L^2}(A^{-z})
  \label{eqn-zeta-def}
\end{equation}
as a function of the complex variable~$z$ and study its meromorphic extension over a region that includes~$z=0$. For our purposes ~$A$ is~$\Delta+m^2$,~$(\Delta+m^2)^{1/2}$ or a Dirichlet-to-Neumann operator. These are positive elliptic~$\Psi$DOs of positive order such that the principal symbol~$\sigma_A(x,\xi)$ is \textsf{strictly positive} whenever~$\xi\ne 0$. In particular their spectra are in~$\mb{R}_+$ and does not intersect~$\mb{B}_{\rho}(0)\subset\mb{C}$ for some~$\rho>0$. This enables one to define the complex power~$A^{-z}$ using the Cauchy integral representation
\begin{equation}
  A^{-z}\defeq \frac{\ii}{2\pi}\int_{\gamma}^{}\me^{-z\log \lambda}(A-\lambda)^{-1}\dd\lambda,
  \label{eqn-complex-pow-cauchy}
\end{equation}
where~$\gamma$ is the contour (with parametrization traversing in order)
\begin{equation}
  \gamma=\{r\me^{\ii\pi}~|~r>\rho\}\cup\{\rho\me^{\ii\theta}~|~-\pi<\theta<\pi\}\cup\{r\me^{-\ii \pi}~|~r>\rho\},
  \label{}
\end{equation}
and~$\log\lambda$ taken to be the principal branch defined on~$\mb{C}\setminus (-\infty,0]$ with $\log 1=0$.

\begin{prop}[\cite{Shu} proposition 10.1, theorems 10.1, 13.1, 13.2, also \cite{Seeley}]\label{prop-def-det-zeta} We have
\begin{enumerate}[(i)]
    \item For the operators~$A$ under consideration, we have the bound
  \begin{equation}
    \nrm{(A-\lambda)^{-1}}_{L^2}\le c|\lambda|^{-1}
    \label{eqn-resolve-op-norm-bound}
  \end{equation}
  for~$\lambda\in \gamma$, and the integral (\ref{eqn-complex-pow-cauchy}) defines~$A^{-z}$ as a holomorphic function valued in bounded~$L^2$ operators, for~$\fk{Re}(z)>0$. It continues as such a holomorphic operator function to all~$z\in\mb{C}$ via
  \begin{equation}
    A^{-z}\defeq A^k A^{-z-k},
    \label{eqn-complex-pow-cont}
  \end{equation}
  where~$k$ is any integer with~$\fk{Re}(z)>-k$ so that~$A^{-z-k}$ is defined by (\ref{eqn-complex-pow-cauchy}), and the definition (\ref{eqn-complex-pow-cont}) does not depend on~$k$.
  \item If~$A$ has order~$s$ then~$A^{-z}$ defined as above is a classical~$\Psi$DO of order~$-zs$. In particular it is trace class on~$L^2(M)$ when~$\fk{Re}(z)>d/s$, $d=\dim M$, for which~$\zeta_A(z)$ is well-defined by (\ref{eqn-zeta-def}). Moreover, for these~$z$,
\begin{equation}
  \zeta_A(z)=\sum_{j=0}^{\infty}\lambda_j^{-z},
  \label{eqn-zeta-eigenval}
\end{equation}
where~$\{\lambda_j\}_{j=0}^{\infty}$ are the eigenvalues of~$A$, and the sum converges absolutely, and uniformly in~$z$ over~$\{\fk{Re}(z)>d/s+\varepsilon\}$ for any~$\varepsilon>0$. 
\item Finally,~$\zeta_A(z)$ can be meromorphically continued over~$\mb{C}$ with simple poles possible at~$\{\frac{d}{s},\frac{d-1}{s},\frac{d-2}{s},\cdots\}\setminus \mb{Z}_{\le 0}$, and holomorphic elsewhere. In particular, it is holomorphic at~$z=0$.\hfill~$\Box$
\end{enumerate}
\end{prop}

\begin{deef}
  For an operator~$A$ under consideration, we define its \textsf{zeta-regularized determinant} as
  \begin{equation}
    \detz A\defeq \exp\left( -\partial_z\zeta_A(0) \right),
    \label{}
  \end{equation}
  where~$\zeta_A(0)$ is the \textsf{zeta function} of~$A$ given by (\ref{eqn-zeta-def}) and (\ref{eqn-complex-pow-cauchy}).
\end{deef}

\begin{def7}
  An alternative way of defining the zeta function and its meromorphic continuation is to use the heat kernel and the Mellin transform. See Gilkey \cite{Gilkey} section 1.12. This way of definition also gives (\ref{eqn-zeta-eigenval}) over the same region, defining therefore the same function as ours.
\end{def7}

\begin{def7}\label{rem-zeta-det-positive}
  From (ii) of proposition \ref{prop-def-det-zeta} we see that if~$A$ is self-adjoint and strictly positive, then~$\zeta_A(z)$ is real-valued for~$z\in (d/s,+\infty)$. But~$\fk{Im}(\zeta_A)$ is real analytic and hence~$\zeta_A$ remains real-valued on~$\mb{R}$ before crossing a pole, and by (iii) it is in particular real-valued on an interval around~$0$. Thus~$\partial_z \zeta_A(0)$ is real and~$\detz A$ is positive.
\end{def7}

Now we move on to the second notion of determinant. Let~$\mathcal{H}$ be a Hilbert space and~$A\in \mathcal{L}(\mathcal{H})$. Denote by~$\Lambda^k \mathcal{H}$ and~$\Lambda^k A$, respectively, the~$k$-th exterior product of~$\mathcal{H}$ and~$A$ (see Simon \cite{Sim1} section 1.5).

\begin{prop}[\cite{Sim1} lemma 3.3]\label{prop-def-det-fred}
  If~$A$ is trace class on~$\mathcal{H}$, then~$\Lambda^k A$ is also trace class on~$\Lambda^k \mathcal{H}$ with bound on trace norm
  \begin{equation}
    \bnrm{\Lambda^k A}_{\mm{tr}}\le \frac{1}{k!}\bnrm{A}_{\mm{tr}}^k.
    \label{}
  \end{equation}
  In particular, putting
  \begin{equation}
    \detf(\one +zA)\defeq \sum_{k=0}^{\infty} z^k\ttr_{\Lambda^k\mathcal{H}}(\Lambda^k A)
    \label{eqn-det-fred-func}
  \end{equation}
  for~$z\in\mb{C}$ defines an entire function, and
  \begin{equation}
    |\detf(\one+z A)|\le \exp(|z|\nrm{A}_{\mm{tr}}).\quad \Box
    \label{}
  \end{equation}
\end{prop}

\begin{deef}
  Let~$A$ be a trace class operator on the Hilbert space~$\mathcal{H}$. Then the determinant~$\detf(\one+A)$ given by (\ref{eqn-det-fred-func}) for~$z=1$ is called the \textsf{Fredholm determinant} of~$\one+A$.
\end{deef}

\begin{lemm}[\cite{Sim1} theorems 3.4, 3.7, 3.8] \label{lemm-det-fred-cont} ~
\begin{enumerate}[(i)]
    \item The map~$A\longmapsto \detf(\one+A)$ defines a continuous function on the trace ideal~$\mathcal{J}_1$ with~$\nrm{\cdot}_{\mm{tr}}$. More precisely,
  \begin{equation}
    |\detf(\one+A)-\detf(\one+B)|\le \nrm{A-B}_{\mm{tr}}\exp(\nrm{A}_{\mm{tr}}+\nrm{B}_{\mm{tr}}+1).
    \label{}
  \end{equation}
  \item If~$A$,~$B\in \mathcal{L}(\mathcal{H})$ are such that both~$AB$ and~$BA$ are of trace class, then we have
  \begin{equation}
    \detf(\one+AB)=\detf(\one +BA). \quad~\Box
    \label{}
  \end{equation}
\end{enumerate}
\end{lemm}

\subsubsection{Factorization Lemma}

\begin{lemm}\label{lemm-det-factor}
  Suppose~$A$ and~$K$ are~$\Psi$DOs such that both~$A$ and~$A(\one+K)$ satisfy the assumptions of proposition \ref{prop-def-det-zeta} and that~$\detz(A)$ and~$\detz(A(\one+K))$ are defined. Suppose moreover~$K$ is trace class and there exists smoothing operators~$\{K_i\}_{i=1}^{\infty}$ such that~$AK_i\to AK$ in~$\nrm{\cdot}_{\mm{tr}}$ (in particular, $AK$ is also trace class). Then
  \begin{equation}
    \detz(A(\one+K))=\detz(A)\detf(\one +K).
    \label{eqn-det-factor}
  \end{equation}
\end{lemm}

\begin{proof}
  We start from Kontsevich and Vishik \cite{KV} proposition 6.4 and take for granted that (\ref{eqn-det-factor}) holds with~$K_i$ in place of~$K$. Our assumptions are tailor-made so that as~$i\to \infty$,
  \begin{equation}
    \detz(A(\one+K_i))\lto \detz(A(\one+K)).
    \label{}
  \end{equation}
  Indeed, by (\ref{eqn-resolve-op-norm-bound}) we have
\begin{equation}
  \nrm{(A(\one+K)-\lambda)^{-1}A(K_i-K)(A(\one+K_i)-\lambda)^{-1}}_{\mm{tr}}\le c|\lambda|^{-2}\bnrm{A(K_i-K)}_{\mm{tr}},
  \label{eqn-diff-resolv-bound}
\end{equation}
with~$c$ independent of~$i$ since a fortiori~$AK_i\to AK$ under~$\nrm{\cdot}_{L^2}$. This in particular shows when~$\fk{Re}(z)>-1$ the integral expression for~$A(\one+K)^{-z}-A(\one+K_i)^{-z}$ is a converging Bochner integral valued in the trace class ideal~$\mathcal{J}_1$ (note~$|\lambda^{-z}|\asymp |\lambda|^{-\fk{Re}(z)}$ as~$\lambda\to -\infty$) and since~$\ttr_{L^2}$ is a continuous functional on~$\mathcal{J}_1$,
\begin{equation}
  \big|\ttr_{L^2}(A(\one+K)^{-z}-A(\one+K_i)^{-z})\big|\lesssim \int_{\gamma}^{} |\lambda|^{-\fk{Re}(z)-2}\bnrm{A(K_i-K)}_{\mm{tr}}\dd\lambda.
  \label{eqn-diff-zeta-bound}
\end{equation}
Now by (iii) of proposition \ref{prop-def-det-zeta} there is~$\frac{1}{2}>\delta>0$ so that~$\zeta_{A(\one+K_i)}$,~$\zeta_{A(\one+K)}$ are both holomorphic over~$\ol{\mb{B}}_{\delta}(0)$ (for example,~$\delta<1/|s|$ where~$s$ is the order of~$A$). Thus by (\ref{eqn-diff-zeta-bound}) and Cauchy's estimate
\begin{align*}
  \big|\zeta_{A(\one+K)}'(0)-\zeta_{A(\one+K_i)}'(0)\big|&\le \frac{1}{\delta}\sup_{|z|=\delta}\big|\zeta_{A(\one+K)}(z)-\zeta_{A(\one+K_i)}(z)\big| \\
  &\lesssim \frac{1}{\delta}\sup_{|z|=\delta}\int_{\gamma}^{} |\lambda|^{-\fk{Re}(z)-2}\bnrm{A(K_i-K)}_{\mm{tr}}\dd\lambda \\
  &\lesssim \bnrm{A(K_i-K)}_{\mm{tr}}\int_{\gamma}^{}|\lambda|^{-3/2}\dd\lambda \\
  &\lesssim \bnrm{A(K_i-K)}_{\mm{tr}}.
\end{align*}
This shows~$|\zeta_{A(\one+K)}'(0)-\zeta_{A(\one+K_i)}'(0)|\to 0$, as we have desired.
\end{proof}

\subsubsection{The Gluing Formula of Burghelea-Friedlander-Kappeler}\label{sec-BFK}

\noindent Let~$(M,g)$ is a closed Riemannian manifold and~$\Sigma\subset M$ an embedded closed hypersurface with induced metric. Assume proposition \ref{prop-stoc-decomp-closed} and decompose~$\phi=\phi_{\Sigma}+\phi_{M\setminus\Sigma}^D$ corresponding to~$\mu_{\mm{GFF}}^M=\mu_{\mm{GFF}}^{M\setminus\Sigma,D}\otimes \mu_{\DN}^{\Sigma,M}$. In view of equation (\ref{eqn-heu-gaussian-mea-det-volume}), and in parallel
\begin{equation}
  \me^{-\frac{1}{2}\ank{\varphi,\DN \varphi}_{L^2}}[\mathcal{L}\varphi]\heueq\detz(\DN_M^{\Sigma})^{-\frac{1}{2}}\dd\mu_{\DN}^{\Sigma,M},
  \label{}
\end{equation}
if we assume a ``formal Fubini theorem'' with respect to the heuristic expressions involving~$\mathcal{L}$, namely
\begin{equation}
  \int_{}^{}\me^{-\frac{1}{2}\sank{\phi,(\Delta+m^2)\phi}_{L^2}}[\mathcal{L}\phi] \heueq\iint \me^{-\frac{1}{2}\sank{\phi_{\Sigma},(\Delta+m^2)\phi_{\Sigma}}_{L^2}}\me^{-\frac{1}{2}\sank{\phi_{M\setminus\Sigma}^D,(\Delta+m^2)\phi_{M\setminus\Sigma}^D}_{L^2}} [\mathcal{L}\phi_{\Sigma}]\otimes [\mathcal{L}\phi_{M\setminus\Sigma}^D],
  \label{}
\end{equation}
then we are led to the following relation of the corresponding determinants (volumes) which were first rigorously proved by Burghelea, Friedlander and Kappeler \cite{BFK}.

\begin{prop}
  [\cite{BFK} theorem B, \cite{Lee} theorem 1.1] Let~$(M,g)$ is a closed Riemannian surface and~$\Sigma\subset M$ an embedded closed hypersurface with induced metric. Then
  \begin{equation}
    \detz(\Delta_M+m^2)=\detz(\Delta_{M\setminus\Sigma,D}+m^2)\detz(\DN_M^{\Sigma}). \quad\Box
    \label{eqn-bfk-non-disec}
  \end{equation}
\end{prop}

The following version where~$\Sigma$ dissects~$M$ such that~$M\setminus \Sigma=M^{\circ}_+\sqcup M^{\circ}_-$ is also useful.

\begin{corr}
  In the situation as above, we have
  \begin{equation}
    \detz(\Delta_M+m^2)=\detz(\Delta_{M_+,D}+m^2)\detz(\Delta_{M_-,D}+m^2)\detz(\DN_M^{\Sigma}).\quad\Box
    \label{eqn-bfk-disec}
  \end{equation}
\end{corr}

As we base our analysis on background Gaussian probability measures, the formulae (\ref{eqn-bfk-non-disec}) and (\ref{eqn-bfk-disec}) constitute separate ingredients (constants) that needs to be ``tuned'' for the final gluing result to hold exactly. In fact, it is also reasonable to consider ``projective gluing'' which allows the freedom for an arbitrary (nonzero) constant to appear in the equation (see remark \ref{rem-segal-proj} and Segal \cite{Segal} page 460).

\subsection{Quadratic Perturbation = Radon-Nikodym Density}\label{sec-quad-pert-rad-niko-main}

\noindent Let~$C$ be a Gaussian covariance operator of order $-s$ on a closed Riemannian manifold~$\Sigma$, and denote by~$\mu_C$ the Gaussian measure on~$\mathcal{D}'(\Sigma)$ with covariance~$\ank{-, C-}_{L^2(\Sigma)}$. Let~$V$ be another bounded formally self-adjoint operator on~$L^2(\Sigma)$ (it could be given by a real symmetric Schwartz kernel). In this section we look at the Gibbs measure
\begin{equation}
  \dd\mu (\varphi) \defeq \frac{\me^{-\frac{1}{2}\ank{\varphi,V\varphi}_{L^2}}\dd\mu_C(\varphi)}{\int_{}^{}\me^{-\frac{1}{2}\ank{\varphi,V\varphi}_{L^2}}\dd\mu_C},
  \label{eqn-quadratic-gibbs-meas}
\end{equation}
which is a Gaussian measure (see proposition \ref{prop-quadratic-perturbation}).

From another perspective we consider Radon-Nikodym densities between mutually absolutely continuous Gaussian measures on~$\mathcal{D}'(\Sigma)$. See Bogachev \cite{Bogachev} section 6.4 for a general treatment from this perspective. We shall reproduce a proof following Glimm and Jaffe \cite{GJ} section 9.3 for reader's convenience and adaptation to the current situation.

A principal corollary of the results of this section is the following.
\begin{corr}\label{corr-rad-niko-dense}
  Let~$\Sigma$ be the disjoint union of Riemannian circles, embedded in an ambient Riemannian surface~$M$ (with or without boundary). Let~$\mu_{\mm{DN}}^{\Sigma,M}$ and~$\mu_{2\mn{D}}^{\Sigma}$ be the two Gaussian measures constructed on~$\mathcal{D}'(\Sigma)$ with covariance operators~$(\DN_M^{\Sigma})^{-1}$ and~$(2\mn{D}_{\Sigma})^{-1}$ (if~$M$ has boundary, specify the boundary condition to be $B$ as in section \ref{sec-dn-map}). Then~$\mu_{\mm{DN}}^{\Sigma,M}$ and~$\mu_{2\mn{D}}^{\Sigma}$ are mutually absolutely continuous with Radon-Nikodym density given by
  \begin{equation}
    \frac{\dd\mu_{\mm{DN}}^{\Sigma,M}}{\dd\mu_{2\mn{D}}^{\Sigma}}(\varphi)=({\tts \det_{\zeta}} (2\mn{D}_{\Sigma}))^{-\frac{1}{2}}({\tts \det_{\zeta}} \DN_M^{\Sigma})^{\frac{1}{2}} \me^{-\frac{1}{2}\sank{\varphi,(\DN_M^{\Sigma}-2\mn{D}_{\Sigma})\varphi}_{L^2(\Sigma)}}.
    \label{}
  \end{equation}
\end{corr}

The proof is at the end of this section. First we come back to the general case.

\begin{prop}\label{prop-quadratic-perturbation}
  Let~$V:L^2(\Sigma)\lto L^2(\Sigma)$ be as above and moreover assume
  \begin{equation}
    C^{-1}+V\textrm{ is positive,}
    \label{eqn-quadratic-perturbation-cov}
  \end{equation}
  and that
  \begin{equation}
    \wh{V}\defeq C^{\frac{1}{2}}V C^{\frac{1}{2}} \textrm{ is trace class.}
    \label{eqn-quadratic-perturbation-trace}
  \end{equation}
  Then
  \begin{enumerate}[(i)]
  \item the random variable~$\ank{\varphi,V\varphi}_{L^2}$ can be defined in $L^1(\mu_C)$ and $\mb{E}_C[\ank{\varphi,V\varphi}_{L^2}]=\ttr(\wh{V})$,
  \item $Z:= \mb{E}_C[\me^{-\frac{1}{2}\ank{\varphi,V\varphi}_{L^2}}] ={\tts \det_{\mm{Fr}}}(\one+ \wh{V})^{-\frac{1}{2}}$, and
  \item the Gibbs measure (\ref{eqn-quadratic-gibbs-meas}) is Gaussian with covariance~$(C^{-1}+V)^{-1}$.
\end{enumerate}
\end{prop}

Note that since~$C^{-1}+V$ is positive and~$C$ is also positive,~$\one + \wh{V}=C^{\frac{1}{2}}(C^{-1}+V)C^{\frac{1}{2}}$ is positive.

\begin{lemm}\label{lemm-quadratic-pert-trace}
  There exist an orthonormal basis~$\{f_j\}_{j=1}^{\infty}$ of the Gaussian Hilbert space~$W^{-s}(\Sigma)$ of~$\mu_C$ equipped with~$\ank{-,C-}_{L^2}$, such that
  \begin{equation}
    \ank{\varphi,V \varphi}_{L^2}=\sum_{j=1}^{\infty} \lambda_j \varphi(f_j)^2,
    \label{eqn-quadratic-expression}
  \end{equation}
  for all~$\varphi$ belonging to the Cameron-Martin space~$W^s(\Sigma)$, where~$\{\lambda_j\}$ are the eigenvalues of~$\wh{V}$ on~$L^2(\Sigma)$, and the series converges absolutely in~$L^1(\mu_C)$. Thus we \textsf{define} the random variable~$\ank{\varphi,V\varphi}_{L^2}$ with this converging series. Consequently, (i) of proposition \ref{prop-quadratic-perturbation} holds.
\end{lemm}

\begin{proof}
  The key is to seek~$\{C^{\frac{1}{2}}f_j\}$ as complete~$L^2$-orthonormal eigenfunctions of~$\wh{V}$ with eigenvalues~$\{\lambda_j\}$, which exist since~$\wh{V}$ is self-adjoint and trace class on~$L^2(\Sigma)$. Now~$\{f_j\}\subset W^{-s}(\Sigma)$ and is complete orthonormal with respect to~$\ank{-,C-}_{L^2}$. This also means the random variables~$\{\varphi(f_j)\}$ are mutually independent. Note for~$\varphi\in W^s(\Sigma)$,~$C^{-\frac{1}{2}}\varphi \in L^2(\Sigma)$, and
  \begin{align*}
    \bank{\varphi,V \varphi}_{L^2}&=\bank{C^{-\frac{1}{2}}\varphi,\wh{V} C^{-\frac{1}{2}}\varphi}_{L^2} \\
    &=\sum_{j=1}^{\infty} \lambda_j \big|\bank{C^{-\frac{1}{2}}\varphi,C^{\frac{1}{2}}f_j}_{L^2}\big|^2 \\
    &=\sum_{j=1}^{\infty} \lambda_j \varphi(f_j)^2.
  \end{align*}
  Since~$\wh{V}$ is trace class,
  \begin{equation}
    \sum_{j=1}^{\infty} \mb{E}_C\big[ |\lambda_j| \varphi(f_j)^2\big]=\sum_{j=1}^{\infty}|\lambda_j|<\infty,
    \label{}
  \end{equation}
  and we obtain the result.
\end{proof}

\begin{def7}
  We thus defined~$\ank{\varphi,V\varphi}_{L^2}$ as a \textsf{Wiener quadratic form} and it lies in the second Wiener chaos of~$\mu_C$, a fortiori in~$L^2(\mu_C)$. See Bogachev \cite{Bogachev} pages 257-261 for more information and in particular proposition 5.10.16 for the same result in the context of Malliavin calculus. In fact,
  \begin{equation}
    VCf_j=C^{-\frac{1}{2}} \wh{V} C^{\frac{1}{2}}f_j=\lambda_j f_j,
    \label{}
  \end{equation}
  and hence~$V$ is trace class on~$W^s(\Sigma)$ with eigenbasis~$\{Cf_j\}$. 
\end{def7}

\begin{def7}
  We have~$\lambda_j>-1$ by the assumption (\ref{eqn-quadratic-perturbation-cov}).
\end{def7}

To treat (ii) and (iii) of proposition \ref{prop-quadratic-perturbation} we adopt some approximations.

\begin{lemm}
  Proposition \ref{prop-quadratic-perturbation} (ii) and (iii) is true in the case~$V$ has finite rank.
\end{lemm}

\begin{proof}
  In this case the series in (\ref{eqn-quadratic-expression}) is finite with~$\{f_j\}_{j=1}^N$ for some~$N\in\mb{N}$. Thus
  \begin{align*}
    Z&=\frac{1}{(2\pi)^{N/2}}\int_{\mb{R}^N}^{}\me^{-\frac{1}{2}\sum_{j} \lambda_j x_j^2} \me^{-\frac{1}{2}\sum_j x_j^2}\dd^N x \\
    &=\prod_j (1+\lambda_j)^{-\frac{1}{2}}=\det(\one +\wh{V})^{-\frac{1}{2}},
  \end{align*}
  by projecting onto~$\mb{R}^N$ via~$\varphi\mapsto (\varphi(f_1),\dots,\varphi(f_N))=:(x_1,\dots,x_N)$. 
   For the covariance, we orthogonally decompose~$W^{-s}(\Sigma)$ as
  \begin{equation}
    W^{-s}(\Sigma)=\spn\{f_j~|~ 1\le j\le N\}\oplus \spn\{f_j~|~ 1\le j\le N\}^{\perp},
    \label{eqn-quadratic-orth-gauss}
  \end{equation}
  let~$\Pi_0$ and~$\Pi_1$ be the corresponding orthogonal projections (in order), and for any~$f\in C^{\infty}(\Sigma)$ write
  \begin{equation}
    \varphi(f)=\alpha_1\varphi(f_1)+\cdots +\alpha_N\varphi(f_N)+\varphi(\Pi_1 f),
  \end{equation}
  then~$\varphi(\Pi_1 f)$ is independent of both~$\varphi(f_j)$, $1\le j\le N$, and~$\ank{\varphi,V\varphi}_{L^2}$. It is now clear that (\ref{eqn-quadratic-gibbs-meas}) is Gaussian because we could now express~$\mu(\{\phi(f)\in A\})$ for any Borel set $A\subset \mb{R}$ as a Gaussian integral over~$\mb{R}^{N+1}$. Apply $C^{\frac{1}{2}}$ to (\ref{eqn-quadratic-orth-gauss}) we get the~$L^2$-orthogonal decomposition
  \begin{equation}
    L^2(\Sigma)=\underbrace{C^{\frac{1}{2}}(\spn\{f_j~|~ 1\le j\le N\})}_{=:L^2(\Sigma)_0}\oplus \underbrace{C^{\frac{1}{2}}(\spn\{f_j~|~ 1\le j\le N\}^{\perp})}_{=:L^2(\Sigma)_1}.
    \label{}
  \end{equation}
  Clearly~$\wh{V}$ leaves this decomposition invariant and is zero on~$L^2(\Sigma)_1$. Hence~$(\one+\wh{V})^{-1}$ is block-diagonal,
  \begin{equation}
    (\one+\wh{V})^{-1}=
    \def\arraystretch{1.3}
    \begin{blockarray}{ccl}
      \{C^{\frac{1}{2}}f_j\}_1^N & L^2(\Sigma)_1& \\
      \begin{block}{(cc)c}
	\diag_{j=1}^N\{(\lambda_j+1)^{-1}\}&0& \{C^{\frac{1}{2}}f_j\}_1^N \\
	0&\one &L^2(\Sigma)_1\\
      \end{block}
    \end{blockarray}
    \label{}
  \end{equation}
   Thus
  \begin{align*}
    Z^{-1}\mb{E}_C[\varphi(f)^2 \me^{-\frac{1}{2}\ank{\varphi,V \varphi}}]&=Z^{-1}\mb{E}_C[\varphi(\Pi_0 f)^2 \me^{-\frac{1}{2}\sum_j \lambda_j \varphi(f_j)^2}] +Z^{-1}\mb{E}_C[\varphi(\Pi_1 f)^2]\mb{E}_C[\me^{-\frac{1}{2}\ank{\varphi,V\varphi}}] \\
    &=\sum_j \alpha_j^2(\lambda_j+1)^{-1}+\bank{\Pi_1 f,C \Pi_1 f}_{L^2} \\
    &=\bank{C^{\frac{1}{2}}\Pi_0 f, (\one +\wh{V})^{-1}C^{\frac{1}{2}}\Pi_0 f}_{L^2}+\bank{\Pi_1 f,C^{\frac{1}{2}}(\one +\wh{V})^{-1}C^{\frac{1}{2}} \Pi_1 f}_{L^2} \\
        &=\ank{f,(C^{-1}+V)^{-1}f}_{L^2},
  \end{align*}
  where we perform again a Gaussian integral on~$\mb{R}^N$ in the second line.
\end{proof}

\begin{proof}[Proof of proposition \ref{prop-quadratic-perturbation} (ii) and (iii).]
For general~$V$ we impose spectral cut-off at~$N$,
\begin{equation}
  V_N(\varphi)\defeq \sum_{j=1}^N \lambda_j \ank{\varphi,f_j}_{L^2}f_j.
  \label{eqn-quadratic-V-spectral-cutoff}
\end{equation}
Then $C^{\frac{1}{2}}V_N C^{\frac{1}{2}}=:\wh{V}_N \to \wh{V}$ under the trace norm~$\nrm{\cdot}_{\ttr}$ (acting on~$L^2(\Sigma)$). Hence~$\ank{\varphi,V_N \varphi}_{L^2}\to \ank{\varphi,V\varphi}_{L^2}$ in~$L^1(\mu_C)$ by lemma \ref{lemm-quadratic-pert-trace} and after passing to a subsequence
\begin{equation}
  \me^{-\frac{1}{2}\ank{\varphi,V_N \varphi}_{L^2}}\to \me^{-\frac{1}{2}\ank{\varphi,V\varphi}_{L^2}}
  \label{}
\end{equation}
in~$L^1(\mu_C)$ and (ii) for $V$ follows, since $\detf(\one+\wh{V}_N)\to \detf(\one+\wh{V})$ by lemma \ref{lemm-det-fred-cont}.
  To prove (iii), we note that~$\wh{V_N}\to \wh{V}$ a fortiori under the operator norm, then~$C^{\frac{1}{2}}(\one+\wh{V}_N)^{-1}C^{\frac{1}{2}}\to C^{\frac{1}{2}}(\one +\wh{V})^{-1}C^{\frac{1}{2}}$ in norm and
  \begin{align*}
    \frac{\mb{E}_C[\me^{\ii\varphi(f)}\me^{-\frac{1}{2}\ank{\varphi,V\varphi}_{L^2}}]}{\mb{E}_C[\me^{-\frac{1}{2}\ank{\varphi,V\varphi}_{L^2}}]}&=\lim_{N\to\infty} \frac{\mb{E}_C[\me^{\ii\varphi(f)}\me^{-\frac{1}{2}\ank{\varphi,V_N\varphi}_{L^2}}]}{\mb{E}_C[\me^{-\frac{1}{2}\ank{\varphi,V_N\varphi}_{L^2}}]}\\
    &=\lim_{N\to\infty} \exp\Big( -\frac{1}{2}\ank{f,(C^{-1}+V_N)^{-1}f}_{L^2} \Big) \\
    &=\exp\Big( -\frac{1}{2}\ank{f,(C^{-1}+V)^{-1}f}_{L^2} \Big),
  \end{align*}
  for~$f\in C^{\infty}(\Sigma)$, showing that (\ref{eqn-quadratic-gibbs-meas}) is Gaussian with the right covariance.
\end{proof}

\begin{proof}[Proof of corollary \ref{corr-rad-niko-dense}.]
Remember now that~$\dim \Sigma=1$.
  Write for short~$\mn{D}:=2\mn{D}_{\Sigma}$ and~$\DN:=\DN_M^{\Sigma}$. Setting~$C=\mn{D}^{-1}$ and~$V=\DN-\mn{D}$ in proposition \ref{prop-quadratic-perturbation}, we are left to prove the determinant identity
  \begin{equation}
    \detz(\mn{D})\detf(\one+\mn{D}^{-\frac{1}{2}}(\DN-\mn{D})\mn{D}^{-\frac{1}{2}})=\detz(\DN).
    \label{}
  \end{equation}
  Indeed, this is now immediate as
  \begin{equation}
    \textrm{LHS}=\detz(\mn{D})\detf(\one+\mn{D}^{-1}(\DN-\mn{D}))=\textrm{RHS}
    \label{}
  \end{equation}
  by lemma \ref{lemm-det-fred-cont} (ii) and lemma \ref{lemm-det-factor}. We point out~$\mn{D}\mn{D}^{-1}(\DN-\mn{D})=V$ can be approximated in~$\nrm{\cdot}_{\mm{tr}}$ by smoothing operators (by lemma \ref{lemm-dn-prop} (iv),~$V=\DN-\mn{D}$ is~$L^2$-trace class), so the conditions of lemma \ref{lemm-det-factor} is satisfied. Indeed, as above,~$\{Cf_j\}\subset W^1(\Sigma)$ are $W^1$-complete eigenfunctions of~$V$. If the corresponding eigenvalue~$\lambda_j\ne 0$, then the bootstrap argument shows~$Cf_j\in C^{\infty}(\Sigma)$ since~$V\in \Psi^{<0}(\Sigma)$. Thus~$V_N$, where~$V_N$ is as in (\ref{eqn-quadratic-V-spectral-cutoff}), approximates~$V$ in~$\nrm{\cdot}_{\mm{tr}}$ and is smoothing. This concludes the proof.
\end{proof}

\begin{def7}
  Though it is probably true, we do not claim~$\DN-\mn{D}$ is elliptic.
\end{def7}

  \section{Variants of Nelson's Argument}\label{sec-nelson-main}

 \noindent The goal of this section is to define part~$A$ of (\ref{eqn-def-mes-gibbs-heu}) which culminates in \hyperref[thrm-nelson]{Nelson's theorem}. The principal obstacle in achieving this is the fact that powers of a distribution such as~$\phi^2$,~$\phi^4$ etc., are generally not defined. Even as a random variable under~$\mu_{\mm{GFF}}^M$, we have (formally)~$\mb{E}_{\mm{GFF}}^M[\phi(x)\phi(y)]\heueq G_{(\Delta+m^2)}(x,y)$ from (\ref{eqn-gff-cov-closed}) for~$x\ne y$ but this implies~$\mb{E}_{\mm{GFF}}^M[\phi(x)\phi(x)]= G_{(\Delta+m^2)}(x,x)=\infty$. This necessitates a procedure of \textsf{renormalization} which \textit{subtracts away}~$\infty$ and makes~$\mb{E}_{\mm{GFF}}^M[\phi(x)\phi(x)]<\infty$.

Here the natural renormalization strategy is provided by the Gaussian probability theory. Let~$\{K_{\varepsilon}~|~\varepsilon>0\}$ be a family of \textit{smoothing} operators\footnote{that is, each $K_{\varepsilon}$ maps $\mathcal{D}'(M)\lto C^{\infty}(M)$.} on~$M$, for which~$K_{\varepsilon}\to \one$ as~$\varepsilon\to 0$ (in a sense to be specified later), and consider the mollified random field~$\phi_{\varepsilon}:=K_{\varepsilon}\phi$. Instead of~$\phi_{\varepsilon}(x)^4$, say, we look at
\begin{align*}
  {:}\phi_{\varepsilon}(x)^4{:}&\defeq \phi_{\varepsilon}(x)^4-6\mb{E}[\phi_{\varepsilon}(x)^2] \phi_{\varepsilon}(x)^2+ 3\mb{E}[\phi_{\varepsilon}(x)^2]^2 \\
  &\defeq \phi_{\varepsilon}(x)^4-6C_{\varepsilon}(x) \phi_{\varepsilon}(x)^2+ 3C_{\varepsilon}(x)^2.
\end{align*}
It happens that for any~$\chi\in C^{\infty}(M)$, the integral~$\int_{M}^{}\chi(x){:}\phi_{\varepsilon}(x)^4{:}\dd V_M(x)$ converges as a random variable in~$L^2(\mu_{\mm{GFF}}^M)$ to a definitive limit, and defines~$\int_{M}^{}\chi(x){:}\phi(x)^4{:}\dd V_M(x)$ as a random variable in~$L^2(\mu_{\mm{GFF}}^M)$. Note that~$C_{\varepsilon}(x)\to \infty$ as~$\varepsilon\to 0$, so we have subtracted ``infinities''.

\begin{def7}\label{rem-locality-crucial}
\textbf{The crucial point} in our adaptation of Nelson's argument is to realize the \textit{locality} of the interaction~$\int_{M}^{}\chi(x){:}P(\phi(x)){:}\dd V_M(x)$ (see section \ref{sec-locality}). To this end we must allow a sufficiently large class of \textit{regulators}~$K_{\varepsilon}$ (in particular, \textit{local} ones) and show that they define the same interaction (proposition \ref{prop-nelson-main}). In addition, the \textit{Wick ordering} also needs to be local (see section \ref{sec-change-wick}), so that the interaction on a domain with boundary could be defined without reference to the ambient closed manifold where this domain ``caps''. See Brunetti, Fredenhagen, Verch \cite{BFV} and Guo, Paycha, Zhang \cite{GPZ} for more information and perspective on locality.
\end{def7}

\begin{def7}
The method adopted here is restricted to dimension two. In three dimensions, the target measure bearing the heuristic form (\ref{eqn-def-mes-gibbs-heu}) becomes mutually singular with respect to~$\mu_{\mm{GFF}}^M$ and hence cannot be expressed as an integrable function multiplied by~$\mu_{\mm{GFF}}^M$. A recent phenomenal method to treat this case is developed in the framework of stochastic PDEs, called \textsf{stochastic quantization}, providing an alternative to older results outlined in \cite{GJ} section 23.1. See the introductions in \cite{GH}, \cite{HS}, \cite{AK}, \cite{MWX} and \cite{BDFT2} for reviews of the literature and pedagogical discussions.
\end{def7}

\subsection{Regularizations}

\noindent In this subsection we describe an admissible class of regulators~$K_{\varepsilon}$ which would eventually produce the same random variable~$\int_{M}^{}\chi(x){:}\phi(x)^4{:}\dd V_M(x)$ as will be proved in the next subsection. Basically, they are smoothing operators such that~$K_{\varepsilon}\to\one$ in~$\Psi^{\delta}(M)$ in the \textit{symbol sense} for any~$\delta>0$ (see definition below). A compact notation is to say~$K_{\varepsilon}\to \one $ in~$\Psi^{0+}(M)$.

\begin{deef}
  Let~$r\in \mb{R}$ we say that operators~$K_{\varepsilon}\to K$ in~$\Psi^r(M)$ \textsf{in the symbol sense} if for any coordinate chart~$\kappa: U\lto \mb{R}^d$ and cut-off~$\chi\in C_c^{\infty}(U)$, the full symbol of~$\chi K_{\varepsilon}\chi$ (considered acting on~$C_c^{\infty}(\kappa(U))$), converges to that of~$\chi K\chi$ in the~$\mathcal{S}^r_{1,0}(\kappa(U)\times \mb{R}^d)$ topology.
\end{deef}

Now we describe the first candidate for~$K_{\varepsilon}$ satisfying the above assumption (the proof is in appendix \ref{app-symbol}). This was introduced in Dyatlov and Zworski \cite{DZ} and has the advantage of being \textsf{local}, realizing the locality of the~$P(\phi)$ interaction eventually in section \ref{sec-locality}. Consider~$\psi\in C_c^{\infty}((-1,1))$ with~$0\le \psi\le 1$ and equal to~$1$ near~$0$. For~$\varepsilon>0$ we define the operator
\begin{equation}
  E_{\varepsilon}u(x)\defeq \int_{M}^{} E_{\varepsilon}(x,y)u(y)\dd V_g(y),\quad \textrm{with }E_{\varepsilon}(x,y)=\frac{1}{F_{\varepsilon}(x)}\psi\left( \frac{d_g(x,y)}{\varepsilon} \right).
  \label{}
\end{equation}
Here~$F_{\varepsilon}(x)=\int_{}^{}\psi(d(x,y)/\varepsilon) \dd y$ so that~$\int_{}^{}E_{\varepsilon}(x,y)\dd y=1$, and~$d_g$ denotes the Riemannian distance. One observes that~$E_{\varepsilon}(x,y)$ is smooth for each~$\varepsilon>0$ so~$E_{\varepsilon}:\mathcal{D}'(M)\lto C^{\infty}(M)$. Observe also that~$\varepsilon^d/C\le F_{\varepsilon}(x)\le C\varepsilon^d$ for some~$C>0$, and this~$C$ could be made dependent neither on~$\varepsilon$ nor on~$x$ as~$M$ is compact.

\begin{lemm}\label{lemm-dyat-zwor-symbol-conv}
  For any~$\delta>0$ we have~$E_{\varepsilon}\to \one$ in~$\Psi^{\delta}(M)$ in the symbol sense.
\end{lemm}

\begin{proof}
    See appendix \ref{app-symbol}.
\end{proof}

Note that~$K^*_{\varepsilon}(x,y)=K_{\varepsilon}(y,x)$ for real smoothing operators (and their symbols are related in a simple manner),~$K_{\varepsilon}\to\one$ in the symbol sense is equivalent to~$K_{\varepsilon}^*\to \one$ in the symbol sense.

\begin{lemm}
  \label{lemm-conv-symbol-top}
  Let~$\{E'_{\varepsilon'}\}_{\varepsilon'>0}$ be another family of smoothing operators such that~$E'_{\varepsilon'}\to \one$ in~$\Psi^{\delta'}(M)$ in the symbol sense for any~$\delta'>0$. Then the net~$E_{\varepsilon}^*(\Delta+m^2)^{-1}(E'_{\varepsilon'}-E_{\varepsilon})$,~$(\varepsilon',\varepsilon)\in\mb{R}_+\times\mb{R}_+$ (we say~$(\varepsilon',\varepsilon)\prec (\varepsilon'_1,\varepsilon_1)$ iff~$\varepsilon'>\varepsilon'_1$ and~$\varepsilon>\varepsilon_1$), converges to zero in~$\Psi^{-2+\delta}(M)$ in the symbol sense for any~$\delta>0$.
\end{lemm}

\begin{proof}
  Note that following essentially the same arguments as above the~$\Psi^{\delta/2}(M)$ seminorms of~$E_{\varepsilon}$ can be bounded uniformly in~$\varepsilon$. This said, the result follows essentially from the continuity of the twisted product (composition product) of symbols as a map~$\mathcal{S}^r_{1,0}\times \mathcal{S}^{r'}_{1,0}\lto \mathcal{S}^{r+r'}_{1,0}$ with respect to the symbol topologies (see Folland \cite{Folland2} page 105 theorem 2.47).
\end{proof}

We shall consider another set of seminorms on~$\Psi^r(M)$ in the case~$-d<r<0$ which suits better our purposes. They are defined as follows. Let~$\mathcal{M}\subset C^{\infty}(M\times M,T(M\times M))$ denote the~$C^{\infty}(M\times M)$-module of smooth vector fields tangent to the diagonal in~$M\times M$. We fix a finite coordinate cover~$\{U_i\}_{i\in 1}^N$ of~$M$, with charts~$\kappa_i:U_i\lto \mb{R}^d$, and a partition of unity~$\{\chi_i\}$ subordinate to this cover.

\begin{deef}\label{def-kernel-diag-top}
  For any~$K\in C^{\infty}(M\times M\setminus \mm{diag})$,~$1\le i\le N$ and~$L_1$, \dots,~$L_p\in \mathcal{M}$, we define the seminorms
  \begin{equation}
    p_{i,L_1,\dots,L_p}(K)\defeq \sup_{(x,y)\in U_i\times U_i} \big|(\kappa_i\times \kappa_i)_*\left( (\chi_i\otimes \chi_i)L_1\cdots L_p K \right)(x,y)\big|\cdot d_g(x,y)^{d+r},
    \label{}
  \end{equation}
  while on~$U_i\times U_j$,~$i\ne j$, which does not touch the diagonal, we use the~$C^{\infty}(U_i\times U_j)$ seminorms. By the \textsf{kernel topology} on~$\Psi^r(M)$,~$-d<r<0$, we mean the topology induced by these seminorms on the Schwartz kernels~$K_A$ of~$A\in \Psi^r(M)$. Here~$d_g$ is the distance function.
\end{deef}

\begin{prop}\label{prop-kernel-top-equiv-symbol}
  In the case~$-d<r<0$, the above kernel topology is equivalent to the topology induced by symbols~$\mathcal{S}^r_{1,0}(T^*M)$ on~$\Psi^r(M)$. In particular, if~$A_{\varepsilon}\to A$ as~$\varepsilon\to 0$ in~$\Psi^r(M)$ in the symbol sense then~$A_{\varepsilon}\to A$ also in the above kernel topology.
\end{prop}

\begin{proof}
    Essentially in Taylor \cite{Taylor2} page 6, proposition 2.2, page 7, proposition 2.4 and page 10, proposition 2.7. See also Bailleul, Dang, Ferdinand and Tô \cite{BDFT} proposition 6.9 for a more detailed treatment.
\end{proof}

Finally, we observe that the heat operator~$K_{\varepsilon}=\me^{-\varepsilon(\Delta+m^2)}$ is also a valid candidate:

\begin{lemm}
  [\cite{Dang} lemma 4.15] We have~$\me^{-\varepsilon(\Delta+m^2)}\to \one$ in~$\Psi^{\delta}(M)$ in the symbol sense for any~$\delta>0$. \hfill~$\Box$
\end{lemm}

Some properties of the heat operator is summed up in appendix \ref{app-symbol}.

\subsection{Integrability of Interaction and Regularization Independence}\label{sec-nelson-reg-indep}

\begin{prop}
  \label{prop-nelson-main} Let~$(K_{\varepsilon})_{\varepsilon>0}$ be any family of real smoothing operators such that~$K_{\varepsilon}\to \one$ in~$\Psi^{\delta}(M)$ in the symbol sense for any~$\delta>0$. Define~$\phi_{\varepsilon}(x)$ and~${:}P(\phi_{\varepsilon}(x)){:}$ as above and let $\chi\in C_c^\infty(M)$ be a test function. Put
  \begin{equation}
    S_{M,\varepsilon,\chi}(\phi):=\int_{M}^{} \chi(x) {:}P(\phi_{\varepsilon}(x)){:} \dd V_M.
    \label{}
  \end{equation}
  This is a random variable on~$\mathcal{D}'(M)$ equipped with~$\mu_{\mm{GFF}}^M$. Then~$\{S_{M,\chi,\varepsilon}\}$ converges in~$L^2(\mu_{\mm{GFF}}^M)$ as~$\varepsilon\to 0$, and the limit is independent of the specific smoothing~$(K_{\varepsilon})$ chosen, provided they have the convergence property described above.
  
  More precisely, for any other smoothing family~$(\tilde{K}_{\varepsilon'})$ satisfying the same condition and defining the random variable~$\tilde{S}_{M,\varepsilon',\chi}$, we have a quantitative bound of the form
\begin{eqnarray*}
  \mathbb{E}\big[\big| S_{M,\varepsilon,\chi}-\tilde{S}_{M,\varepsilon^\prime,\chi}\big|^2\big]\le \fk{O}(\vert \varepsilon-\varepsilon^\prime\vert) C_n  \Vert \chi\Vert_{L^4}^2,
\end{eqnarray*}  
where~$\fk{O}(|\varepsilon-\varepsilon'|)$ is a function going to zero as~$|\varepsilon-\varepsilon'|\to 0$, depending only on~$M$.
\end{prop}

A particularity of the dimension two is seen in the following elementary lemma.

\begin{lemm}\label{lemm-2d-log-blow-up-can}
  If~$\dim M=2$ and~$g$ is a smooth Riemannian metric on~$M$ then
  \begin{equation}
    \iint_{M\times M}|\log(d_g(x,y))|^p \dd V_{M\times M}\le C_p<\infty
    \label{}
  \end{equation}
  for any~$1\le p<\infty$ and
  \begin{equation}
    \iint_{M\times M} d_g(x,y)^{-\delta}\dd V_{M\times M}\le C_{\delta}<\infty
    \label{}
  \end{equation}
  for any~$0<\delta<1$,~$d_g$ denoting the distance function.\hfill~$\Box$
\end{lemm}

Consequently, since in two dimensions~$G(x,y)=\mathcal{O}(\log(d(x,y)))$ as~$y\to x$ (see lemma \ref{lemm-tadpole-mass}), a simple argument with partition of unity shows
\begin{equation}
  \iint_{M\times M} \chi(x)\chi(y) |G(x,y)|^p \dd V_{M\times M}\le C_p \Vert \chi\Vert_{L^\infty}^2\vol(\text{supp}(\chi))^2 <\infty
  \label{eqn-upper-bound-green-p}
\end{equation}
for some~$C_p>0$, for all~$1\le p<\infty$. When~$p=2$ this is the familiar fact that in two dimensions the Green operator~$(\Delta+m^2)^{-1}$ is Hilbert-Schmidt, which may also be shown using Weyl's law.

We will denote
\begin{align}
  \mb{E}_{\mm{GFF}}^M\big[\phi_{\varepsilon}(x)\phi_{\varepsilon}(y)\big]&=\bank{\delta_x,K_{\varepsilon}^*(\Delta+m^2)^{-1}K_{\varepsilon}\delta_y}_{L^2}\defeq G_{\varepsilon,\varepsilon}(x,y),\\
  \mb{E}_{\mm{GFF}}^M\big[\phi_{\varepsilon}(x)\phi_{\varepsilon'}(y)\big]&=\bank{\delta_x,K_{\varepsilon}^*(\Delta+m^2)^{-1}\tilde{K}_{\varepsilon'}\delta_y}_{L^2}\defeq G_{\varepsilon,\varepsilon'}(x,y), \\
  \mb{E}_{\mm{GFF}}^M\big[\phi_{\varepsilon'}(x)\phi_{\varepsilon'}(y)\big]&=\bank{\delta_x,\Tilde{K}_{\varepsilon'}^*(\Delta+m^2)^{-1}\tilde{K}_{\varepsilon'}\delta_y}_{L^2}\defeq G_{\varepsilon',\varepsilon'}(x,y).
  \label{}
\end{align}
Note that~$G_{\varepsilon,\varepsilon'}$ has a \textit{different} regulator on the second variable! For fixed~$\varepsilon>0$, we have by lemma \ref{lemm-wick-feyn},
\begin{equation}
    {:}\phi_{\varepsilon}(x)^{2n}{:}=G_{\varepsilon,\varepsilon}( x,x)^n h_{2n}\big(\phi_{\varepsilon}(x)\big/G_{\varepsilon,\varepsilon}(x,x)^{\frac{1}{2}} \big)\ge -b_1 G_{\varepsilon,\varepsilon}(x,x)^n,
    \label{eqn-lower-bound-hermite}
  \end{equation}
  for some constant~$b_1>0$ independent of~$\varepsilon$, since the even-degree Hermite polynomial~$h_{2n}$ is bounded below.

\begin{proof}[Proof of proposition \ref{prop-nelson-main}.]
   We shall prove the proposition for the case~$P(\theta)=\theta^{2n}$,~$n\in\mb{N}$, the general case is similar. For fixed~$\varepsilon$,~$\varepsilon'>0$, we compute
  \begin{align*}
    \nrm{S_{M,\varepsilon}-S_{M,\varepsilon'}}_{L^2(\mu_{\mm{GFF}})}^2&=\mb{E}\bigg[ \left| \int_{M}^{} \chi {:}\phi_{\varepsilon}(x)^{2n}{:}\dd V_M -\int_{M}^{} \chi  {:}\phi_{\varepsilon'}(x)^{2n}{:} \dd V_M \right|^2 \bigg] \\
    &=\mb{E}\left[  \iint_{M\times M}^{}{:}\chi(x) \phi_{\varepsilon}(x)^{2n}{:}~{:}\phi_{\varepsilon}(y)^{2n}{:}\chi(y)\dd V_M\otimes \dd V_M \right] \\
    &\quad \quad-2\mb{E}\left[  \iint_{M\times M}^{}\chi(x)  {:}\phi_{\varepsilon}(x)^{2n}{:}~{:}\phi_{\varepsilon'}(y)^{2n}{:} \chi(y) \dd V_M\otimes \dd V_M \right] \\
    &\quad\quad +\mb{E}\left[  \iint_{M\times M}^{}  \chi(x){:}\phi_{\varepsilon'}(x)^{2n}{:}~{:}\phi_{\varepsilon'}(y)^{2n}{:}\chi(y)  \dd V_M\otimes \dd V_M \right] \\
    &=\iint_{M\times M}^{}\Big( \mb{E}\left[\chi(x) {:}\phi_{\varepsilon}(x)^{2n}{:}~{:}\phi_{\varepsilon}(y)^{2n}{:}\chi(y) \right] -2  \mb{E}\left[\chi(x) {:}\phi_{\varepsilon}(x)^{2n}{:}~{:}\phi_{\varepsilon'}(y)^{2n}{:} \chi(y)\right] \\
    &\quad\quad+\mb{E}\left[\chi(x) {:}\phi_{\varepsilon'}(x)^{2n}{:}~{:}\phi_{\varepsilon'}(y)^{2n}{:} \chi(y)\right] \Big)
    \dd V_{M\times M} \tag{Tonelli}\\
    &=(2n)!\iint_{M\times M}\chi(x)\chi(y) \left( G_{\varepsilon,\varepsilon}(x,y)^{2n}-2G_{\varepsilon,\varepsilon'}(x,y)^{2n}+G_{\varepsilon',\varepsilon'}(x,y)^{2n} \right) \dd V_{M\times M} \tag{lemma \ref{lemm-wick-feyn}}
  \end{align*}
  We will control the integral~$\iint_{M\times M}\chi(x)\chi(y)\left( G_{\varepsilon,\varepsilon}(x,y)^{2n}-G_{\varepsilon,\varepsilon'}(x,y)^{2n} \right)\dd V_{M\times M}$ for $\chi\in C^\infty(M)$ and $\chi\geqslant 0$. 
  Indeed,
\begin{align*}
  |\textrm{this integral}|&\le \iint_{M\times M}\chi(x)\chi(y) \left|G_{\varepsilon,\varepsilon}(x,y)-G_{\varepsilon,\varepsilon'}(x,y)\right| \cdot \big|G_{\varepsilon,\varepsilon}(x,y)^{2n-1} \\
  &\quad\quad +G_{\varepsilon,\varepsilon}(x,y)^{2n-2}G_{\varepsilon,\varepsilon'}(x,y)+\dots+G_{\varepsilon,\varepsilon'}(x,y)^{2n-1} \big| \dd V_{M\times M} \\
  &\le 2n  \Vert\chi\Vert^2_{L^4} \left( C_{8n-4} \right)^{\frac{1}{4}}
  \times \left( \iint_{M\times M} \left|G_{\varepsilon,\varepsilon}(x,y)-G_{\varepsilon,\varepsilon'}(x,y)\right|^2 \right)^{\frac{1}{2}} 
\end{align*}
Remember that~$G_{\varepsilon,\varepsilon'}(x,y)$ is the kernel of~$K^*_{\varepsilon}(\Delta+m^2)^{-1}\tilde{K}_{\varepsilon'}$ and~$G_{\varepsilon,\varepsilon'}(x,y)-G_{\varepsilon,\varepsilon}(x,y)$ is the kernel of~$K^*_{\varepsilon}(\Delta+m^2)^{-1}(\tilde{K}_{\varepsilon'}-K_{\varepsilon})$. By lemma \ref{lemm-conv-symbol-top}, definition \ref{def-kernel-diag-top} and proposition \ref{prop-kernel-top-equiv-symbol},
\begin{equation}
  |G_{\varepsilon,\varepsilon'}(x,y)|\le C_{M,\delta}d_g(x,y)^{-\delta},
  \label{}
\end{equation}
uniformly in~$(\varepsilon,\varepsilon')$ and by lemma \ref{lemm-conv-symbol-top},
\begin{equation}
  |G_{\varepsilon,\varepsilon'}(x,y)-G_{\varepsilon,\varepsilon}(x,y)| \le \fk{O}_{M,\delta}(|\varepsilon-\varepsilon'|)d_g(x,y)^{-\delta}
  \label{}
\end{equation}
for any~$\delta>0$. If we restrict moreover to $\delta<1$, then we prove our result, thanks to lemma \ref{lemm-2d-log-blow-up-can}.
\end{proof}

\subsection{Integrability of the Exponential of Interaction}

\noindent In this subsection we will adopt the heat regulator~$K_{\varepsilon}=\tilde{K}_{\varepsilon}:=\me^{-\varepsilon(\Delta+m^2)}$.

\begin{lemm}\label{lemmB6}
  Let~$\deg P=2n$. Then~$S_{M,\chi,\varepsilon}$,~$\varepsilon>0$, and hence the resulting limit~$S_{M,
  \chi}$, is in the~$(2n)$-th Wiener chaos of the GFF, that is,~$\ol{\mathcal{P}}_{2n}(\mathcal{H})$, where~$\mathcal{H}=W^{-1}(M)$ is the Gaussian Hilbert space of the GFF.
\end{lemm}

\begin{proof}
  We shall show it for~$P(\theta)=\theta^4$ and the general case is similar. Here we use the spectral representation of remark \ref{rem-mass-gau-field-four} (using the notations thereof) and take~$K_{\varepsilon}:=\me^{-\varepsilon(\Delta+m^2)}$. Thus we can write
  \begin{align*}
    \phi_{\varepsilon}(x)^4&=\bigg( \sum_{j=0}^{\infty}\ank{\varphi_j, K_{\varepsilon}\delta_x}_{L^2} \xi_j\bigg)^4 =\bigg( \sum_{j=0}^{\infty} \me^{-\varepsilon(\lambda_j+m^2)}\varphi_j(x) \xi_j\bigg)^4 \\
    &=\sum_{j,k,\ell,p} \me^{-\varepsilon(\lambda_j+\lambda_k+\lambda_{\ell}+\lambda_p+4m^2)} \varphi_j(x)\varphi_k(x)\varphi_{\ell}(x)\varphi_p(x)\xi_j\xi_k\xi_{\ell}\xi_p,
  \end{align*}
  the series converging absolutely in~$L^2(\mu_{\mm{GFF}}^M)$. Now each individual term is clearly in~$\ol{\mathcal{P}}_4(\mathcal{H})$. Since
  \begin{equation}
    \int_{}^{}|\varphi_j(x)\varphi_k(x)\varphi_{\ell}(x)\varphi_p(x)|\dd V_M(x) \lesssim \vol(M)\cdot [\textrm{polynomial in }\lambda_j,\lambda_k,\lambda_{\ell},\lambda_p\textrm{ with fixed degree}]
    \label{}
  \end{equation}
  as one has~$\sup_M |\varphi_j|\lesssim (1+\lambda_j)^2$ in two dimensions which follows essentially from the Sobolev embedding (see Sogge \cite{Sogge2} page 43 equation (3.1.12)), thus
  \begin{equation}
    \int_{}^{}\chi(x)\phi_{\varepsilon}(x)^4 \dd V_M(x)=\sum_{j,k,\ell,p} \me^{-\varepsilon(\lambda_j+\lambda_k+\lambda_{\ell}+\lambda_p+4m^2)} C_{\chi,j,k,\ell,p}\xi_j\xi_k\xi_{\ell}\xi_p
    \label{}
  \end{equation}
  with the series converging absolutely in~$L^2(\mu_{\mm{GFF}}^M)$ and the result is in~$\ol{\mathcal{P}}_4(\mathcal{H})$.
\end{proof}

\begin{prop}
  [hypercontractivity, \cite{Janson} theorem 5.10, \cite{Sim2} theorem I.22] \label{propB7} Let~$\mathcal{H}\subset L^2(Q,\mathcal{O},\mb{P})$ be a Gaussian Hilbert space on some probability space~$(Q,\mathcal{O},\mb{P})$, and let~$n\ge 1$,~$2\le p<\infty$. Then 
  \begin{equation}
    \mb{E}[|X|^p]^{\frac{1}{p}}\le (p-1)^{\frac{n}{2}}\mb{E}[X^2]^{\frac{1}{2}}
    \label{}
  \end{equation}
  for all~$X\in\ol{\mathcal{P}}_n (\mathcal{H})$.  \hfill~$\Box$
\end{prop}

Combining lemma \ref{lemmB6} and proposition \ref{propB7}, we have

\begin{corr}\label{cor-nelson-all-Lp-cutoff-limit}
  The convergence of~$\{S_{M,\chi,\varepsilon}\}$, as well as the limit~$S_{M,\chi}$, is in~$L^p(\mu_{\mm{GFF}})$ for all~$1\le p<\infty$. Moreover, if $\chi\rightarrow 0$ in $L^4(M)$ then 
$$\lim_{\varepsilon\rightarrow 0}\lim_{\chi\rightarrow 0} S_{M,\varepsilon,\chi}=0  $$
as random variable in $L^p(\mu_{\mm{GFF}}^M)$ for all $1\le p<\infty$.
Moreover $\mathbb{E}[\me^{-S_{M,\varepsilon,\chi}}] $ remains uniformly bounded along the limit. In particular $S_{M,\chi}$ is defined for $\chi\in L^4(M)$. \hfill~$\Box$
\end{corr}

We single out a calculus computation which will be used in the sequel:

\begin{lemm}\label{lemmB9}
  Let~$a$,~$b$ be positive real numbers. Then the real function~$\alpha(x):=x^{(bx)} a^x$,~$x>0$, attains its minimum value~$\me^{-be^{-1}a^{-1/b}}$ at~$x=e^{-1}a^{-1/b}$. \hfill ~$\Box$
\end{lemm}

\begin{thrm}
  [Nelson] \label{thrm-nelson} We have for $\chi\in L^4(M)$,
  \begin{equation}
    \me^{-S_{M,\chi}}\in L^1(\mu_{\mm{GFF}})
    \label{}
  \end{equation}
  and hence
  \begin{equation}
  Z_M={\tts \det_{\zeta}}(\Delta_g+m^2)^{-\frac{1}{2}}\int_{\mathcal{D}'(M)}^{}\me^{-\int_{M}^{}{:}P(\phi){:} \dd V_M}\dd\mu_{\mm{GFF}}^M(\phi)<\infty.
    \label{eqn-def-part-func-rigor}
  \end{equation}
\end{thrm}

\begin{proof}
  For any~$\varepsilon>0$, by (\ref{eqn-upper-bound-green-p}) and (\ref{eqn-lower-bound-hermite}),
  \begin{equation}
    S_{M,\chi,\varepsilon}\ge -b\vol(\text{supp}(\chi)) \Vert \chi\Vert_{L^\infty} \sup_x (|G_{\varepsilon,\varepsilon}(x,x)|^n).
    \label{}
  \end{equation}
From formula (\ref{eqnB9}), for~$\varepsilon$ small,
  \begin{equation}
    G_{\varepsilon,\varepsilon}(x,x)=\int_{2\varepsilon}^{\infty}p_t(x,x)\dd t=\bigg(\underbrace{\int_{2\varepsilon}^{1}}_{A}+\underbrace{\int_{1}^{\infty}}_{B}\bigg) p_t(x,x)\dd t.
    \label{}
  \end{equation}
  Now by (iii) of lemma \ref{lemm-heat} part~$A$ is~$\mathcal{O}(\log(2\varepsilon))$; since our field is \textit{massive} ($m>0$), by (iv) of lemma \ref{lemm-heat} part~$B$ is \textit{bounded}. Therefore one has overall~$G_{\varepsilon,\varepsilon}(x,x)=\mathcal{O}(\log(2\varepsilon))$.
  As a result~$S_{M,\chi,\varepsilon}\ge -b_2 |\log(2\varepsilon)|^n$ for~$\varepsilon$ small.

  Now we compute that
  \begin{align*}
    \mb{P}\big( \me^{-S_{M,\chi}}\ge \me^{b_2|\log(2\varepsilon)|^n +1}\big)&=\mb{P}\left( S_{M,\chi}\le -b_2|\log(2\varepsilon)|^n-1 \right) \\
    &\le \mb{P}\left( |S_{M,\chi}-S_{M,\chi,\varepsilon}|\ge 1 \right) \\
    &\le \nrm{S_{M,\chi}-S_{M,\chi,\varepsilon}}_{L^p(\mu_{\mm{GFF}})}^p \tag{Chebyshev} \\
    &\le (p-1)^{\frac{np}{2}} C_1^p \varepsilon^{\frac{p}{2}}\Vert \chi\Vert_{L^4}^p \tag{proposition \ref{propB7} and \ref{prop-nelson-main}} \\
    &\lesssim \Vert \chi\Vert_{L^4}^pp^{\frac{n}{2}p}(C_1\varepsilon^{\frac{1}{2}})^p,
  \end{align*}
  for all $2\le p <\infty$. The last line as a function of~$p$ has the form dealt with in lemma \ref{lemmB9} and attains a minimum of~$\exp\big( -C_2(\varepsilon^{\frac{1}{2}}\Vert \chi\Vert_{L^4})^{-1/n} \big)$, with some absorbed constant~$C_2>0$ which does not depend on $\chi$. Thus we obtain
  \begin{equation}
    \mb{P}\big( \me^{-S_{M,\chi}}\ge \me^{b_2|\log(2\varepsilon)|^n +1}\big)\lesssim \exp\big( -C_2\big(\varepsilon^{\frac{1}{2}}\Vert \chi\Vert_{L^4}\big)^{-\frac{1}{n}} \big).
    \label{}
  \end{equation}
  Now
we may conclude with the formula
\begin{eqnarray*}
\mathbb{E}\big[e^{-S_{M,\chi}} \big]=\int_0^\infty \mb{P}\left( \me^{-S_{M,\chi}}\ge  t\right)\dd t=\int_0^1 \mb{P}\left( \me^{-S_{M,\chi}}\ge  t\right)\Big| \frac{\dd t}{\dd\varepsilon} \Big| \dd\varepsilon
\end{eqnarray*}  
where the last integral involves a change of variable~$t:=\me^{b_2|\log(2\varepsilon)|^n +1}$. This gives
\begin{eqnarray*}
\mathbb{E}\big[e^{-S_{M,\chi}} \big]\lesssim 1+\int_0^1 \exp\big( -C_2\big(\varepsilon^{\frac{1}{2}}\Vert \chi\Vert_{L^4}\big)^{-\frac{1}{n}} \big) C_3\varepsilon^{-1}n\vert\log(\varepsilon)\vert^{n-1}
e^{c\vert\log(\varepsilon)\vert^n-1}
  \dd\varepsilon
\end{eqnarray*}  
which is finite since integrable near $\varepsilon=0$.   
  Moreover, we see that the bound is uniform when $\Vert \chi\Vert_{L^4}\leqslant C_0$ for some given $C_0>0$.
\end{proof}

\subsection{Change of Wick Ordering}\label{sec-change-wick}

\noindent In order for the proof of proposition \ref{prop-nelson-main} to work as it is written one has to insist on the Wick ordering~${:}\bullet{:}$ provided by~$\mu_{\mm{GFF}}^M$, since we desire convergence in~$L^2(\mu_{\mm{GFF}}^M)$ and with a different Wick ordering the Feynman rules (lemma \ref{lemm-wick-feyn}) are not exact. Nevertheless, in order to define the interaction over a domain~$\Omega$ independently of its embedding in an ambient manifold~$M$, one must employ a Wick ordering independent of~$M$, or in order words, one that is \textit{local}.

Let~$d$ denote the Riemannian distance function of~$M$. This function is local in the sense that~$d(x,y)$ (as $y\to x$) depends only on the restriction of the Riemannian metric on any geodesic convex neighborhood containing~$x$ and~$y$. The local Wick ordering~${:}\bullet{:}_0$ is provided by the \textsf{log-measure}~$\mu_{\log}^M$ which is the Gaussian measure on~$\mathcal{D}'(M)$ with covariance
\begin{equation}
  \mb{E}_{\log}[\phi(f)\phi(h)]=\int_{M}^{}f(x)\Big(-\frac{1}{2\pi}\log(m\cdot d(x,y))\Big) h(y)\dd V_M(x)\dd V_M(y),
  \label{}
\end{equation}
for~$f$,~$h\in C^{\infty}(M)$, thanks to lemma \ref{lemm-2d-log-blow-up-can}. Here~$m$ is the mass used for~$(\Delta+m^2)^{-1}$. We denote~$-\frac{1}{2\pi}\log(md(x,y)):=C_0(x,y)$ and the corresponding operator by~$C_0$. We emphasize here that~$\mu_{\log}^M$ is used only as a tool to produce a linear change of random variables with deterministic coefficients, no random variables will be actually defined on~$\mu_{\log}^M$.

\begin{lemm}
  If~$C_1$,~$C_2$ are two covariance operators on~$M$, then
  \begin{equation}
    {:}\phi(f)^n{:}_{C_1} =\sum_{j=0}^{\lfloor n/2 \rfloor} \frac{n!}{(n-2j)!j!2^j} \ank{f,(C_2-C_1)f}_{L^2(M)}^j {:}\phi(f)^{n-2j}{:}_{C_2},
    \label{}
  \end{equation}
  for~$f\in C^{\infty}(M)$.
\end{lemm}

\begin{proof}
  Follows readily from Wick's theorem.
\end{proof}

The reason why the new Wick ordering works is the following. Let~$G_{(\Delta+m^2)}(x,y)$ denote the integral kernel of~$(\Delta+m^2)^{-1}$.

\begin{lemm}\label{lemm-tadpole-mass}
  For each~$x\in M$, the limit
  \begin{equation}
    \lim_{y\to x}\left( G_{(\Delta+m^2)}(x,y)-C_0(x,y) \right)\defeq \delta G(x)
    \label{}
  \end{equation}
  exists, and that~$\delta G\in L^p(M)$ for all~$1\le p<\infty$.
\end{lemm}

\begin{proof}
  [Remarks for proof.] The function~$\delta G$ is called in our context the (point-splitting) \textsf{tadpole function} (see Kandel, Mnev and Wernli \cite{KMW} section 5.4, in particular lemma 5.20 for a precise expression), which can be seen as a \textsf{renormalized} diagonal value of the Green function~$G_{(\Delta+m^2)}(x,x)$. The asymptotic of the Green function along the diagonal is a classical subject and we have in fact~$G_{(\Delta+m^2)}(x,y)-C_0(x,y)\in C^1(M\times M)$. The function~$\delta G$ is also important in the context of conformal geometry where it is called the \textsf{mass function}, if more precisely we do not include the constant~$-\log m/2\pi$ in~$C_0$ but rather in~$\delta G$. See Hermann and Humbert \cite{HH}, Ludewig \cite{Ludewig} or Schoen and Yau \cite{SY} for more information.
\end{proof}

It follows that
\begin{equation}
  \ank{E_{\varepsilon}\delta_x,\left( (\Delta+m^2)^{-1}-C_0 \right)E_{\varepsilon}\delta_x}_{L^2(M)} \lto \delta G(x)
  \label{}
\end{equation}
as~$\varepsilon\to 0$ and through the limiting process of proposition \ref{prop-nelson-main} we find
\begin{equation}
  \int_{M}^{}\chi(x){:}\phi(x)^{2n}{:}_{0}\dd V_M(x) =\sum_{j=0}^{ n } \frac{(2n)!}{(2n-2j)!j!2^j} \int_{}^{}\chi(x)\delta G(x)^j {:}\phi(x)^{2n-2j}{:}_{\mm{GFF}}\dd V_M(x),
  \label{}
\end{equation}
which exist as a random variable in~$L^p(\mu_{\mm{GFF}}^M)$ for all~$2\le p<\infty$ since each term on the RHS are such by proposition \ref{prop-nelson-main} and corollary \ref{cor-nelson-all-Lp-cutoff-limit}.

  \section{Geometric Operators and Induced Laws}\label{sec-trace-poisson}

  \noindent In this section we obtain a series of rather elementary relations between various geometric-analytic operators on $M$ and on $\Omega$. The moral is that, the so-called ``sharp-time localization'' map $j_{\Sigma}$ (see lemma \ref{lemm-DN-trick}), induced probability laws of Gaussian fields under the trace $\tau_{\Sigma}$ (see section \ref{sec-ind-law}), and finally the Green-Stokes formula, are largely different aspects of the same thing.
  
  \subsection{Summary of Operators Concerned}  \label{sec-dn-map}

   Let~$(M,g)$ be a closed Riemannian manifold, and~$\Sigma\subset M$ a smooth embedded hypersurface (codimension one submanifold). The map
\begin{equation}
  \left.
  \begin{array}{rcl}
    \tau_{\Sigma}:C^{\infty}(M) & \lto & C^{\infty}(\Sigma),\\
    f &\longmapsto & f|_{\Sigma},
  \end{array}
  \right.
  \label{}
\end{equation}
is called the \textsf{trace map} from~$M$ onto~$\Sigma$. If~$(\Omega,g)$ is a compact Riemannian manifold with boundary~$\partial\Omega$, we know that for~$f\in C^{\infty}(\partial\Omega)$, the (Helmholtz) boundary value problem
\begin{equation}
  \left\{
  \begin{array}{ll}
    (\Delta_{\Omega}+m^2)u=0,&\textrm{in }\Omega,\\
    u|_{\partial\Omega}=f,&\textrm{on }\partial\Omega,
  \end{array}
  \right.
  \label{eqn-PI-pde-bdy}
\end{equation}
admits a unique solution~$u\in C^{\infty}(\ol{\Omega})$ which is extendably smooth upto~$\partial\Omega$. The solution operator
  \begin{equation}
    \left.
    \begin{array}{rcl}
      \PI_{\Omega}^{\partial\Omega}:C^{\infty}(\partial\Omega) &\lto &C^{\infty}(\ol{\Omega}),\\
      f&\longmapsto & u\textrm{ solving (\ref{eqn-PI-pde-bdy})},
    \end{array}
    \right.
    \label{}
  \end{equation}
  is called in this paper the \textsf{Poisson integral operator} from~$\partial\Omega$ to~$\Omega$, \textit{with mass}~$m>0$. We also need a variant of this operator that works for embedded hypersurfaces in closed manifolds, as $M$ and $\Sigma$ above. Pick~$f\in C^{\infty}(\Sigma)$. This time we look at the boundary value problem
\begin{equation}
  \left\{
  \begin{array}{ll}
    (\Delta_{M}+m^2)u=0,&\textrm{in }M\setminus\Sigma,\\
    u|_{\Sigma}=f,&\textrm{on }\Sigma.
  \end{array}
  \right.
  \label{eqn-PI-pde-hyp}
\end{equation}
Indeed, one views~$M\setminus \Sigma$ as a manifold with two boundaries~$\Sigma\sqcup \Sigma$, and as a result one obtains a unique solution~$u$ which is smooth on~$M\setminus \Sigma$ and one-sidedly smooth upto~$\Sigma$ respectively on its two sides. In this case we denote the solution operator by
  \begin{equation}
    \left.
    \begin{array}{rcl}
      \PI_{M}^{\Sigma}:C^{\infty}(\Sigma) &\lto &C^{\infty}(\ol{M\setminus\Sigma}),\\
      f&\longmapsto & u\textrm{ solving (\ref{eqn-PI-pde-hyp})},
    \end{array}
    \right.
    \label{eqn-def-pi-emb-hyp-in-closed-case}
  \end{equation}
  also called the \textsf{Poisson integral operator}, from~$\Sigma$ to~$M$, \textit{with mass}~$m>0$.

\begin{lemm}
  [\cite{Taylor1} page 334, 361, \cite{Eskin} page 57, example 13.3] \label{lemm-trace-prop} Let~$(M,g)$,~$\Sigma$, $\Omega$ and $\partial\Omega$ be as above. Then
  \begin{enumerate}[(i)]
    \item the map~$\tau_{\Sigma}$ extends uniquely to a continuous operator~$\tau_{\Sigma}:W^s(M)\lto W^{s-\frac{1}{2}}(\Sigma)$ for each~$s>\frac{1}{2}$;
    \item the map~$\tau_{\Sigma}:W^s(M)\lto W^{s-\frac{1}{2}}(\Sigma)$,~$s>\frac{1}{2}$, is surjective;
    \item the map~$\PI_{\Omega}^{\partial\Omega}$ extends uniquely to a continuous operator $\PI_{\Omega}^{\partial\Omega}:W^s(\partial \Omega)\lto W^{s+\frac{1}{2}}(\Omega)$ for each~$s\ge -\frac{1}{2}$.
  \end{enumerate}
\end{lemm}

\begin{def7}\label{rem-Dir-Neu-cond-mean}
  In this paper, a function~$u\in C^{\infty}(\ol{\Omega})$ is said to satisfy the \textsf{Dirichlet condition} (respectively \textsf{Neumann}) if~$u|_{\partial\Omega}=0$ (respectively~$(\partial_{\nu}u)|_{\partial\Omega}=0$,~$\nu$ the outward unit normal). The~$f$ appearing in (\ref{eqn-PI-pde-bdy}) is called a \textsf{Dirichlet datum}.
\end{def7}

\begin{deef}\label{def-pi-more-general}
   More generally if~$\Omega$ has boundaries and~$\Sigma$ is \textit{either} one component of~$\partial\Omega$ \textit{or} embedded in the \textit{interior} of~$\Omega$, then we denote by~$\PI_{\Omega}^{\Sigma,B}f$ the solution with its restriction equal to~$f$ on~$\Sigma$ and boundary condition ``$B$'' (Dirichlet or Neumann) on all components of~$\partial\Omega$ \textit{except}~$\Sigma$. Such notations raise no ambiguity when the situation is understood from context.
\end{deef}

\begin{lemm}\label{lemm-pi-hyp-reg}
  Let~$(M,g)$ and~$\Sigma$ be as above. Then~$\PI_{M}^{\Sigma}$ extends uniquely to a continuous operator
  \begin{equation}
    \PI_{M}^{\Sigma}:W^{\frac{1}{2}}(\Sigma)\lto W^{1}(M).
    \label{}
  \end{equation}
\end{lemm}

\begin{proof}
  Let~$f\in W^{\frac{1}{2}}(\Sigma)$ and~$u:=\PI_M^{\Sigma}f$. We know that~$u\in W^1(M\setminus\Sigma)\subset \mathcal{D}'(M\setminus\Sigma)$, this means~$u\in L^2(M\setminus \Sigma)=L^2(M)$, and~$\nabla u\in L^2(T(M\setminus \Sigma))$, as a distribution over~$M\setminus \Sigma$. The problem is to compute~$\nabla u$ as a distribution over~$M$. For this, one picks a testing vector field~$X\in C^{\infty}(M,TM)$ and applies the Green-Stokes formula to get
  \begin{align*}
    \ank{\nabla u,X}_{L^2(M,TM)}&\defeq -\ank{u,\ddiv X}_{L^2(M)}\\
    &=\int_{M\setminus\Sigma}^{}\ank{\nabla u,X}_g \dd V_M -\int_{\Sigma}^{} f\ank{X,\nu}_g \dd V_{\Sigma}-\int_{\Sigma}^{}f\ank{X,-\nu}_g\dd V_{\Sigma} \\
    &=\int_{M\setminus\Sigma}^{}\ank{\nabla u,X}_g \dd V_M.
  \end{align*}
  Here~$\nu$ is any one of the two possible unit normal vector fields along~$\Sigma$. This shows that, nevertheless,
  \begin{equation}
    \nabla^M u=\nabla^{M\setminus\Sigma}u,
    \label{}
  \end{equation}
  and hence~$\nrm{u}_{W^1(M)}=\nrm{\nabla u}_{L^2(T(M\setminus\Sigma))}+\nrm{u}_{L^2(M)}\approx\nrm{u}_{W^1(M\setminus\Sigma)}$. 
\end{proof}

  \begin{deef}
  Let~$(\Omega,g)$ be a compact Riemannian manifold with boundary~$\partial \Omega\ne\varnothing$, and~$\PI_{\Omega}^{\partial \Omega}:C^{\infty}(\partial \Omega)\to C^{\infty}(\ol{\Omega})$ the Poisson operator defined previously. Put
  \begin{equation}
    \left.
    \begin{array}{rcl}
       \DN_{\Omega}^{\partial \Omega}:C^{\infty}(\partial \Omega)&\lto & C^{\infty}(\partial \Omega),\\
       f&\longmapsto & \partial_{\nu}(\PI_{\Omega}^{\partial \Omega} f),
    \end{array}
    \right.
    \label{}
  \end{equation}
  where~$\nu=$ outward unit normal along $\partial \Omega$, called the \textsf{Dirichlet-to-Neumann operator} on $\partial\Omega$ with respect to $\Omega$.
\end{deef}

  \begin{deef} \label{def-jp-DN-map}
  Let~$(M,g)$ be a closed Riemannian manifold,~$\Sigma\subset M$ an embedded hypersurface, and~$\PI_{M}^{\Sigma}:C^{\infty}(\partial \Omega)\to C^{\infty}(M\setminus \Sigma)$ the hypersurface Poisson operator. Put
  \begin{equation}
    \left.
    \begin{array}{rcl}
       \DN_{M}^{\Sigma}:C^{\infty}(\partial \Omega)&\lto & C^{\infty}(\partial \Omega),\\
       f&\longmapsto & \partial_{\nu}(\PI_{M}^{\Sigma} f)|_{\Sigma_-}+\partial_{-\nu}(\PI_{M}^{\Sigma} f)|_{\Sigma_+},
    \end{array}
    \right.
    \label{}
  \end{equation}
  where~$\nu$ is any one of the two unit normal vector fields along~$\Sigma$, extended over a cylindrical neighborhood of~$\Sigma$. Here~$\Sigma_{-}$ and~$\Sigma_+$ means that we are taking one-sided derivatives, respectively, from the backward-time and forward-time directions with regard to the flow of~$\nu$. We call~$\DN_{M}^{\Sigma}$ the \textsf{jumpy Dirichlet-to-Neumann operator} on $\Sigma$ with respect to $M$.
\end{deef}

\begin{deef}\label{def-DN-boundary}
   Similarly if~$\Omega$ has boundaries and~$\Sigma$ is \textit{either} one component of~$\partial\Omega$ \textit{or} embedded in the \textit{interior} of~$\Omega$, then we denote by~$\DN_{\Omega}^{\Sigma,B}$ the corresponding operator with~$\PI_{M}^{\Sigma}$ replaced by~$\PI_{\Omega}^{\Sigma,B}$ in the definition (see definition \ref{def-pi-more-general}).
\end{deef}

\begin{def7}\label{rmk-dn-jump}
  If we see~$M\setminus\Sigma$ as a manifold with boundary~$\Sigma\sqcup \Sigma$, then~$\DN_M^{\Sigma}f$ is also the sum over~$\Sigma$ of the two \textit{outward} unit normal derivatives of~$\PI_M^{\Sigma}f$ along~$\Sigma\sqcup\Sigma$. Intuitively,~$\DN_M^{\Sigma}f$ describes the ``jump'' of~$\nabla\PI_M^{\Sigma}f$ across~$\Sigma$.
\end{def7}

We summarize in the following lemma the essential properties of~$\DN_{\Omega}^{\partial\Omega}$ and~$\DN_M^{\Sigma}$. Parallel results also hold for~$\DN_{\Omega}^{\Sigma,B}$ ($\Sigma$ being either one component of boundary or embedded in interior) but we shall not discuss them in order to simplify the presentation. The same applies to everything below this section.

\begin{lemm}\label{lemm-dn-prop}
  Under their respective settings,~$\DN_{\Omega}^{\partial\Omega}$ and~$\DN_M^{\Sigma}$ are such that
  \begin{enumerate}[(i)]
    \item their quadratic forms are given respectively by the Dirichlet energies of their harmonic extensions:
      \begin{align}
	\bank{f,\DN_{\Omega}^{\partial\Omega}f}_{L^2(\partial\Omega)}&=\int_{\Omega}^{}\big(|\nabla\PI_{\Omega}^{\partial\Omega}f|_g^2+m^2(\PI_{\Omega}^{\partial\Omega}f)^2\big) \dd V_{\Omega}, \\
	\bank{f,\DN_{M}^{\Sigma}f}_{L^2(\Sigma)}&=\int_{M}^{}\big(|\nabla\PI_{M}^{\Sigma}f|_g^2+m^2(\PI_{M}^{\Sigma}f)^2\big) \dd V_{M}, 
	\label{eqn-dn-qua-hyp}
      \end{align}
      for~$f\in C^{\infty}(\partial\Omega)$;
    \item they are formally self-adjoint, strictly positive, and~$L^2$-invertible;
    \item they are elliptic~$\Psi$DOs of order~$1$, with principal symbols being~$|\xi|_g$ and~$2|\xi|_g$ respectively;
    \item they afford a finer comparison with~$\mn{D}_{\partial\Omega}=(\Delta_{\partial\Omega}+m^2)^{\frac{1}{2}}$ or~$2\mn{D}_{\Sigma}=2(\Delta_{\Sigma}+m^2)^{\frac{1}{2}}$: the operators
      \begin{equation}
	\DN_{\Omega}^{\partial\Omega}-\mn{D}_{\partial\Omega},\quad \DN_{M}^{\Sigma}-2\mn{D}_{\Sigma},\quad \mn{D}_{\partial\Omega}^{-1}\DN_{\Omega}^{\partial\Omega}-\one, \quad \textrm{and}\quad (2\mn{D}_{\Sigma})^{-1}\DN_{M}^{\Sigma}-\one,
	\label{}
      \end{equation}
      are~$\Psi$DOs of orders at most~$-2$,~$-2$,~$-3$, and~$-3$ respectively. A fortiori, they are all of trace class when~$\dim \Omega=\dim M=2$ and~$\dim \Sigma=\dim \partial\Omega=1$.
  \end{enumerate}
\end{lemm}

\begin{proof}
  See Taylor \cite{Taylor3} and the references therein. For (iv) see Kandel, Mnev and Wernli \cite{KMW} proposition A.3.
\end{proof}

\subsection{Two Consequences of the Green-Stokes Formula}\label{sec-green-stokes}

\noindent Let~$(M,g)$ be a closed Riemannian manifold and~$\Sigma\subset M$ an embedded hypersurface. Formula (\ref{eqn-dn-qua-hyp}) in a slightly more general form allows one to obtain an expression for the ``distributional adjoint'' of the trace map onto~$\Sigma$.

\begin{lemm}\label{lemm-DN-trick}
  Let~$\tau_{\Sigma}:C^{\infty}(M)\to C^{\infty}(\Sigma)$,~$\phi\mapsto \phi|_{\Sigma}$ be the trace map. One has, for any~$\phi\in C^{\infty}(M)$ and~$f\in C^{\infty}(\Sigma)$,
  \begin{equation}
    \ank{\tau_{\Sigma}\phi,f}_{L^2(\partial\Omega)}=\ank{\phi,j_{\Sigma}f}_{L^2(M)}
    \label{eqn-adj-trace}
  \end{equation}
  where~$j_{\Sigma}=(\Delta+m^2)\PI_M(\DN_M^{\Sigma})^{-1}$. Moreover, this equality can be extended to~$\phi\in W^1(M)$ and~$f\in W^{-\frac{1}{2}}(\Sigma)$.
\end{lemm}

\begin{proof}
  First suppose~$h\in C^{\infty}(\Sigma)$, then applying the Green-Stokes formula to~$M\setminus\Sigma$ with boundary~$\Sigma\sqcup\Sigma$ gives
\begin{align*}
  \ank{\tau\phi,\DN_M^{\Sigma}h}_{L^2(\Sigma)}&=\int_{M\setminus\Sigma}^{} (\langle\nabla \phi,\nabla \PI_M h\rangle+ m^2 \phi(\PI_M h)) \dd V_{M} \\
&=\int_{M}^{} (\langle\nabla \phi,\nabla \PI_M h\rangle+ m^2 \phi(\PI_M h)) \dd V_{M} \tag{lemma \ref{lemm-pi-hyp-reg}}\\
&=\ank{\phi,(\Delta_M+m^2)\PI_M^{\Sigma} h}_{L^2(M)}. \tag{\#}
\end{align*}
We remark that step (\#) is the \textit{definition} of the action of~$(\Delta_M+m^2)$ on the distribution~$\PI_M^{\Sigma} h$. By lemma \ref{lemm-trace-prop}, lemma \ref{lemm-pi-hyp-reg}, and (iii) of lemma \ref{lemm-dn-prop}, then, this equality can be extended to~$\phi\in W^1(M)$ and~$h\in W^{\frac{1}{2}}(\Sigma)$. Finally, replacing~$h$ by~$(\DN_{M}^{\Sigma})^{-1}f$, with~$(\DN_M^{\Sigma})^{-1}$ being a~$\Psi$DO of order~$-1$, yields the desired relation (\ref{eqn-adj-trace}) as well as its domain.
\end{proof}

\begin{def7}
  One is advised to compare lemma \ref{lemm-DN-trick} with the fact in one dimensions that the distributional derivative of the Heaviside function~$H_a=a\cdot 1_{(0,\infty)}$ ($a\in\mb{R}$) is the delta function multiplied by the jump of~$H_a$ across~$0$, that is,
  \begin{equation}
  \ank{H_a',\varphi}_{L^2(\mb{R})}=[H_a(0+)-H_a(0-)]\cdot \varphi(0)
    \label{}
  \end{equation}
  for any~$\varphi\in \mathcal{S}(\mb{R})$. In our case the role of the Heaviside function is played by the vector field~$\nabla \PI_M^{\Sigma}(\DN_M^{\Sigma})^{-1}f$. Indeed, following remark \ref{rmk-dn-jump}, the ``jump'' of~$\nabla \PI_M^{\Sigma}(\DN_M^{\Sigma})^{-1}f$ across $\Sigma$ is exactly $f$, as the directions tangential to~$\Sigma$ does not contribute to the jump with~$(\DN_M^{\Sigma})^{-1}f$ being smooth. This comparison in mind, it is also customary to write~$j_{\Sigma}f$ as~$f\otimes \delta_{\Sigma}$, as for example, in Carron \cite{Carron}.
\end{def7}

\begin{corr}\label{cor-dn-is-conj-of-green}
  For~$f$,~$h\in C^{\infty}(\Sigma)$, we have
  \begin{equation}
    \ank{f,(\DN_M^{\Sigma})^{-1}h}_{L^2(\Sigma)}=\ank{j_{\Sigma}f,(\Delta+m^2)^{-1}j_{\Sigma}h}_{L^2(M)}.
    \label{}
  \end{equation}
  In other words,~$(\DN_M^{\Sigma})^{-1}=\tau(\Delta+m^{2})^{-1}j_{\Sigma}~(=\tau(\Delta+m^{2})^{-1}\tau^*)$.
\end{corr}

\begin{proof}
  This is immediate by noting that~$\tau_{\Sigma}\PI_M^{\Sigma}$ is the identity on~$L^2(\Sigma)$.
\end{proof}

\begin{def7}
  Indeed, noting that the Schwartz kernel~$K_{\tau}$ of~$\tau_{\Sigma}$ is the delta distribution on the diagonal~$\{(x,x)\}\subset \Sigma\times M$, and that $j_{\Sigma}$ is the distributional adjoint of $\tau_{\Sigma}$, corollary \ref{cor-dn-is-conj-of-green} allows one to deduce immediately the Schwartz kernel of~$(\DN_M^{\Sigma})^{-1}$, denoted $G_{\DN}^{\Sigma}$:
  \begin{align*}
    \ank{f,(\DN_M^{\Sigma})^{-1}h}_{L^2(\Sigma)}&=\iint_{\Sigma\times M}\iint_{\Sigma\times M} f(x)K_{\tau}(x,z)G_{(\Delta+m^2)}(z,w)K_{\tau}(y,w)h(y)\dd x\dd z\dd y\dd w \\
    &=\iint_{\Sigma\times \Sigma}f(x)G_{(\Delta+m^2)}(x,y)h(y)\dd x\dd y,
  \end{align*}
  that is,~$G_{\DN}^{\Sigma}=G_{(\Delta+m^2)}|_{\Sigma\times \Sigma}$, where $G_{(\Delta+m^2)}$ is the Helmholtz Green function on $M$, which is a well-known result. Of course, assuming this result, one could also work backwards to give lemma \ref{lemm-DN-trick} another proof, using the Poisson integral formula (lemma \ref{lemm-adj-PI-bdy} below) for $\PI_M^{\Sigma}$.
\end{def7}

Now we move to the second consequence of the Green-Stokes formula. Let~$(\Omega,g)$ be a compact Riemannian manifold with boundary~$\partial\Omega$. Recall that~$(\Delta_{\Omega,D}+m^2)^{-1}$ denotes the Helmholtz Green operator with \textit{Dirichlet} conditions on~$\partial\Omega$.

\begin{lemm}[\cite{Taylor2} page 46] \label{lemm-adj-PI-bdy}
  We have, for~$\varphi\in C^{\infty}(\partial\Omega)$ and~$f\in C_c^{\infty}(\Omega^{\circ})$,
  \begin{equation}
    \bank{\PI_{\Omega}^{\partial\Omega}\varphi,f}_{L^2(\Omega)}=
    \bank{\varphi,-\partial_{\nu}(\Delta_{\Omega,D}+m^2)^{-1}f|_{\partial\Omega}}_{L^2(\partial\Omega)},
    \label{}
  \end{equation}
  where~$\nu$, again, denotes the \textit{outward} unit normal vector field along~$\partial\Omega$.\hfill~$\Box$
\end{lemm}

\begin{corr}
  For~$(M,g)$ a closed Riemannian manifold and~$\Sigma\subset M$ an embedded hypersurface, for~$\varphi\in C^{\infty}(\Sigma)$ and~$f\in C^{\infty}_c(M\setminus\Sigma)$,
  \begin{equation}
    \ank{\PI_M^{\Sigma}\varphi,f}_{L^2(M)}=\ank{\varphi,-(\partial_{\nu}u|_{\Sigma_-}+\partial_{-\nu}u|_{\Sigma_+})}_{L^2(\Sigma)},
    \label{}
  \end{equation}
  where~$u=(\Delta_{M\setminus\Sigma,D}+m^2)^{-1}f$, and the notations~$\Sigma_-$ and~$\Sigma_+$ have the same meanings as in definition \ref{def-jp-DN-map}.
\end{corr}

\begin{proof}
  See~$M\setminus\Sigma$ as a manifold with two boundaries~$\Sigma\sqcup\Sigma$, and we note
  \begin{equation}
    \PI_M^{\Sigma}\varphi=\PI_{M\setminus\Sigma}^{\Sigma\sqcup\Sigma}\bnom{\varphi}{\varphi},
    \label{eqn-two-PI-same-law}
  \end{equation}
  as well as
  \begin{equation}
    \Bank{\bnom{\varphi}{\varphi},-\bnom{\partial_{\nu}u|_{\Sigma_-}}{\partial_{-\nu}u|_{\Sigma_+}}}_{L^2(\Sigma\sqcup\Sigma)}=
\bank{\varphi,-(\partial_{\nu}u|_{\Sigma_-}+\partial_{-\nu}u|_{\Sigma_+})}_{L^2(\Sigma)},
    \label{}
  \end{equation}
  while applying lemma \ref{lemm-adj-PI-bdy}.
\end{proof}

\subsection{Induced Laws}\label{sec-ind-law}

Results of this section rely on the possibility of extending linear functionals or operators measurably from the Cameron-Martin space, and the (almost sure) uniqueness of such extensions, a detailed treatment of which we refer to \cite{Bogachev} sections 2.10 and 3.7.

Now let~$(M,g)$ be a closed Riemannian manifold and~$\Sigma\subset M$ an embedded hypersurface. Lemma \ref{lemm-DN-trick} then says that, for each~$f\in C^{\infty}(\Sigma)$, the random variables
\begin{equation}
  \tau_{\Sigma}\phi(f)\quad\textrm{and}\quad \phi(j_{\Sigma}f),
  \label{}
\end{equation}
while~$\phi\sim \mu_{\mm{GFF}}^M$, are (surely) equal on the Cameron-Martin space~$W^1(M)$. By \cite{Bogachev} theorem 2.10.11, they are almost surely equal over~$\mathcal{D}'(M)$. Subsequently from corollary \ref{cor-dn-is-conj-of-green} we deduce
\begin{equation}
  \mb{E}_{\mm{GFF}}^M[\tau_{\Sigma}\phi(f)\tau_{\Sigma}\phi(h)]=\mb{E}_{\mm{GFF}}^M[\phi(j_{\Sigma}f)\phi(j_{\Sigma}h)]=\ank{f,(\DN_M^{\Sigma})^{-1}h}_{L^2(\Sigma)}.
  \label{}
\end{equation}
Taking into account lemma \ref{lemm-GHS-cam-mar-of-cov-C}, (iii) of lemma \ref{lemm-dn-prop} and \cite{Bogachev} theorem 3.7.6, we have proved the following.
\begin{prop}\label{prop-induce-law-DN}
  If~$\phi\in \mathcal{D}'(M)$ follows the law of~$\mu_{\mm{GFF}}^M$, then the random field~$\tau_{\Sigma}\phi\in \mathcal{D}'(\Sigma)$ can equivalently be realized as the (centered) Gaussian field~$\tilde{\varphi}$ on~$\Sigma$ with covariance
  \begin{equation}
    \mb{E}\big[\tilde{\varphi}(f)\tilde{\varphi}(h)\big]=\ank{f,(\DN_{M}^{\Sigma})^{-1}h}_{L^2(\Sigma)},
    \label{eqn-cov-dn-field}
  \end{equation}
  for~$f$,~$h\in C^{\infty}(\Sigma)$. In other words, the measure image~$\wh{\tau_{\Sigma}}_*(\mu_{\mm{GFF}}^M)$ of~$\mu_{\mm{GFF}}^M$ under any measurable linear extension~$\wh{\tau_{\Sigma}}$ of~$\tau_{\Sigma}:W^1(M)\lto W^{\frac{1}{2}}(\Sigma)$ coincides with the measure~$\mu_{\DN}^{\Sigma,M}$ on any~$\mathcal{D}'(\Sigma)$ for the field~$\tilde{\varphi}$ satisfying (\ref{eqn-cov-dn-field}). \hfill~$\Box$
\end{prop}

Next we study induced random fields in the other direction, by the Poisson integral operator. Namely, for~$\Omega$,~$\partial\Omega$ as in lemma \ref{lemm-adj-PI-bdy}, given a Gaussian random field~$\varphi$ on~$\partial\Omega$, what is the law of the field~$\PI_{\Omega}^{\partial\Omega}\varphi$? From another perspective one solves the Helmholtz (Laplace) equation with random boundary conditions. We write in shorthand
\begin{equation}
    (\PI_{\Omega}^{\partial\Omega})^*\defeq -\partial_{\nu}(\Delta_{\Omega,D}+m^2)^{-1}(-)|_{\partial\Omega}.
\end{equation}
Thus lemma \ref{lemm-adj-PI-bdy} says
\begin{equation}
\bank{\PI_{\Omega}^{\partial\Omega}\varphi,f}_{L^2(\Omega)}=
  \bank{\varphi,(\PI_{\Omega}^{\partial\Omega})^* f}_{L^2(\partial\Omega)},
  \label{}
\end{equation}
for~$\varphi\in C^{\infty}(\partial\Omega)$,~$f\in C_c^{\infty}(\Omega^{\circ})$. Suppose~$\varphi$ has covariance operator~$C$ of order~$-s$ ($s>0$). Note~$\PI_{\Omega}^{\partial\Omega}$ is always well-defined on the Cameron-Martin space~$W^s(\partial\Omega)$. By the same token as above, for~$f\in C_c^{\infty}(\Omega^{\circ})$,
\begin{equation}
  (\PI_{\Omega}^{\partial\Omega}\varphi)(f)\quad\textrm{and}\quad \varphi((\PI_{\Omega}^{\partial\Omega})^*f)
  \label{}
\end{equation}
are surely equal on~$W^s(\partial\Omega)$. Moreover,
\begin{equation}
  \mb{E}_{C}^{\partial\Omega}\big[(\PI_{\Omega}^{\partial\Omega}\varphi)(f)(\PI_{\Omega}^{\partial\Omega}\varphi)(h)\big]=\bank{f,\PI_{\Omega}^{\partial\Omega} C(\PI_{\Omega}^{\partial\Omega})^*h}_{L^2(\partial\Omega)}.
\end{equation}
We deduce
\begin{prop}\label{prop-induce-PI-law}
  If~$\varphi\in \mathcal{D}'(\partial\Omega)$ is a (centered) Gaussian random field with covariance operator~$C$, then the random field~$\PI_{\Omega}^{\partial\Omega}\varphi\in \mathcal{D}'(\Omega^{\circ})$ can equivalently be realized as the (centered) Gaussian field~$\tilde{\phi}$ on~$\Omega$ with covariance
  \begin{equation}
    \mb{E}\big[\tilde{\phi}(f)\tilde{\phi}(h)\big]=\bank{f,\PI_{\Omega}^{\partial\Omega} C(\PI_{\Omega}^{\partial\Omega})^*h}_{L^2(\partial\Omega)},
    \label{eqn-cov-pi-ind-field}
  \end{equation}
  for~$f$,~$h\in C^{\infty}(\partial\Omega)$. In other words, the measure image~$\wh{\PI_{\Omega}^{\partial\Omega}}_*(\mu_{C}^{\partial\Omega})$ of~$\mu_{C}^{\partial\Omega}$ under any measurable linear extension~$\wh{\PI_{\Omega}^{\partial\Omega}}$ of~$\PI_{\Omega}^{\partial\Omega}:W^s(\partial\Omega)\lto W^{s+\frac{1}{2}}(\Omega)$ coincides with the measure for the field~$\tilde{\phi}$ satisfying (\ref{eqn-cov-pi-ind-field}). \hfill~$\Box$
\end{prop}

 \begin{def7}\label{rem-induce-law-equality}
    What is strictly needed for showing Segal axioms is not the full proposition \ref{prop-induce-PI-law} but rather this innocent observation: by (\ref{eqn-two-PI-same-law}), if the random field~$\varphi\in \mathcal{D}'(\Sigma)$ follows a fixed probability law, then the induced random fields~$\PI_M^{\Sigma}\varphi$ and~$\PI_{M\setminus\Sigma}^{\Sigma\sqcup\Sigma}[\smx{\varphi\\ \varphi}]$ in~$\mathcal{D}'(M\setminus\Sigma)$ follows the same law.
  \end{def7}

\subsection{A Remark on Reflection Positivity}  
\label{sec-RP-first}

\noindent As the names would suggest, the positivity of the Dirichlet-to-Neumann map (itself the consequence of the positivity of the Dirichlet energy) gives an interesting inequality comparing the resolvants of Laplacians with Dirichlet and Neumann conditions (corollary \ref{cor-rp-op-ineq}), via the Poisson integral formula (lemma \ref{lemm-adj-PI-bdy}). Let~$\Omega$,~$\partial\Omega$ be as in lemma \ref{lemm-adj-PI-bdy}.  We adopt the shorthand notations
\begin{equation}
  \quad C_D\defeq (\Delta_{\Omega,D}+m^2)^{-1},\quad C_N\defeq (\Delta_{\Omega,N}+m^2)^{-1},
  \label{}
\end{equation}
where, similar to~$C_D$,~$C_N f$ for~$f\in C_c^{\infty}(\Omega^{\circ})$ solves the Neumann boundary value problem
\begin{equation}
  \left\{
    \begin{array}{ll}
      (\Delta+m^2)(C_N f)=f,&\textrm{in }\Omega,\\
      \partial_{\nu}(C_N f)|_{\partial\Omega}=0,&\textrm{on }\partial\Omega.
    \end{array}
    \right.
  \label{}
\end{equation}
There is the following simple, elementary relation:
\begin{lemm}
  In the situation as above, we have the operator equality on~$C_c^{\infty}(\Omega^{\circ})$,
  \begin{equation}
    \PI_{\Omega}^{\partial\Omega}(\DN_{\Omega}^{\partial\Omega})^{-1}(\PI_{\Omega}^{\partial\Omega})^*=C_N-C_D.
    \label{}
  \end{equation}
\end{lemm}

\begin{proof}
  Pick~$f\in C_c^{\infty}(\Omega^{\circ})$, let~$u:=C_D f$ and put~$w:=\PI_{\Omega}^{\partial\Omega}(\DN_{\Omega}^{\partial\Omega})^{-1}(-\partial_{\nu} u|_{\partial\Omega})$. Then, by the definition of~$\DN_{\Omega}^{\partial\Omega}$, $w$ solves the following boundary value problem:
  \begin{equation}
    \left\{
    \begin{array}{ll}
      (\Delta+m^2)w=0,&\textrm{in }\Omega,\\
      \partial_{\nu}w|_{\partial\Omega}=-\partial_{\nu}u|_{\partial\Omega},&\textrm{on }\partial\Omega.
    \end{array}
    \right.
    \label{}
  \end{equation}
  However,
  \begin{equation}
    \left\{
    \begin{array}{ll}
      (\Delta+m^2)u=f,&\textrm{in }\Omega,\\
     u|_{\partial\Omega}=0,&\textrm{on }\partial\Omega,
    \end{array}
    \right.\quad\textrm{implying}\quad
    \left\{
    \begin{array}{ll}
      (\Delta+m^2)(u+w)=f,&\textrm{in }\Omega,\\
      \partial_{\nu}(u+w)|_{\partial\Omega}=0,&\textrm{on }\partial\Omega,
    \end{array}
    \right.
    \label{}
  \end{equation}
  namely~$u+w=C_N f$, that is,~$w=C_N f-C_D f$. We obtain the result.
\end{proof}

Now the positivity of~$\DN_{\Omega}^{\partial\Omega}$ (lemma \ref{lemm-dn-prop}) implies~$C_N\ge C_D$, namely~$\ank{f,(C_N-C_D)f}_{L^2(\Omega)}\ge 0$ for all~$f\in C_c^{\infty}(\Omega^{\circ})$. One step further,

\begin{corr}\label{cor-rp-op-ineq}
  We have~$C_N\ge C_D$ as operators on~$L^2(\Omega)$. \hfill~$\Box$
\end{corr}

It is emphasized in Jaffe and Ritter \cite{JR} section 3 that~$C_N\ge C_D$ is the crucial relation that leads to the so-called \textit{reflection positivity} (RP) of the GFF. The geometric set-up is as follows. Let now~$\partial\Omega\subset\Omega$ be totally geodesic,~$\Omega^*$ a copy of~$\Omega$ (reversing the coorientation of~$\partial\Omega$),~$|\Omega|^2:=\Omega^*\cup_{\partial\Omega}\Omega$, the \textsf{isometric double} which is a closed Riemannian manifold, and~$\Theta:|\Omega|^2\lto |\Omega|^2$ an isometric involution fixing~$\partial\Omega$, such that~$\Theta(\Omega)=\Omega^*$ and~$\Theta(\Omega^*)=\Omega$.

\begin{def7}\label{rem-real-tun-geo}
  Such~$\Omega$ is named in Gibbons \cite{Gibbons} (in the 4-dimensional case) as a \textsf{real tunnelling geometry} (see \cite{Gibbons} section 4). The isometric double $|\Omega|^2$ and the involution~$\Theta$ exist, the latter being a \textsf{reflection}, \textit{dissecting} $|\Omega|^2$ into $\Omega$ and $\Omega^*$, its fixed point set being~$\partial\Omega$ (see also Ritter \cite{Ritter} section 2.1.1).
\end{def7}

The action of~$\Theta$ extends in the usual way to~$C^{\infty}(|\Omega|^2)$ and~$\mathcal{D}'(|\Omega|^2)$ by pulling-back. This set-up brings in another resolvant operator which is
\begin{equation}
  C\defeq (\Delta_{|\Omega|^2}+m^2)^{-1}.
  \label{}
\end{equation}
Denote also by~$\Pi_+:L^2(|\Omega|^2)\lto L^2(\Omega)$ the orthogonal projection. Then we have

\begin{lemm}[\cite{JR} lemma 3]
  Let~$|\Omega|^2$,~$\Theta$,~$\Pi_+$ and~$C$ be as above. Then
  \begin{equation}
    \Pi_+ \Theta C=\frac{1}{2}(C_N-C_D)
    \label{}
  \end{equation}
  on~$C_c^{\infty}(\Omega^{\circ})$ and~$L^2(\Omega)$.\hfill~$\Box$
\end{lemm}

In summary,

\begin{corr}
  [equivalent formulations of RP] \label{cor-equiv-form-rp} In the situation as above, we have
  \begin{equation}
    2\bank{f,\Theta Cf}_{L^2(|\Omega|^2)}=\bank{f,(C_N-C_D)f}_{L^2(\Omega)}=\bank{(\PI_{\Omega}^{\partial\Omega})^*f,(\DN_{\Omega}^{\partial\Omega})^{-1}(\PI_{\Omega}^{\partial\Omega})^* f}_{L^2(\partial\Omega)}
    \label{eqn-rp-equiv-form}
  \end{equation}
  for all~$f\in C_c^{\infty}(\Omega^{\circ})$, and all of the above quantities are nonnegative. \hfill~$\Box$
\end{corr}

\begin{def7}
  Note that the 3 quantities in (\ref{eqn-rp-equiv-form}) make sense respectively on the spaces~$L^2(|\Omega|^2)$, $L^2(\Omega)$ and~$L^2(\partial\Omega)$. While the first quantity is the original view of RP, the second quantity offers a \textit{one-sided} view and the third provides a view \textit{within the boundary}~$\partial\Omega$. Nevertheless, the map~$\DN_{\Omega}^{\partial\Omega}$ reflects the geometry of the bulk, see for example Paternain, Salo and Uhlmann \cite{PSU} section 11.5.
\end{def7}

\section{Markov Property and Consequences}\label{sec-markov-main}

\noindent Materials in sections \ref{sec-sobo-decomp}, \ref{sec-stoc-decomp-gff} are largely classical with an excellent source being Dimock \cite{Dimock} which we follow roughly. See also Simon \cite{Sim2} section III.3 and see Powell and Werner \cite{PW} section 4.2 for a probabilistic point of view.

  \subsection{Decompositions of Sobolev Spaces}\label{sec-sobo-decomp}

   \noindent Let~$(M,g)$ be a closed Riemannian manifold. Recall from Appendix \ref{sec-app-sobo} the definitions of the spaces~$W^s_A(M)$ and~$W^s_U(M)$ for~$A\subset M$ closed and~$U\subset M$ open. Recall also (lemma \ref{lemm-sobo-inner-prod}) that we have the isometric isomorphism~$(\Delta+m^2):W^1(M)\xlongrightarrow{\sim}W^{-1}(M)$ as we endow~$W^1(M)$ and~$W^{-1}(M)$ respectively with the inner products~$\sank{-,(\Delta+m^2)-}_{L^2}$ and~$\sank{-,(\Delta+m^2)^{-1}-}_{L^2}$.

\begin{lemm}\label{lemm-sobo-decomp}
  Let~$A\subset M$ be a closed set. Then~$W^{1}(M)$ and~$W^{-1}(M)$ decompose orthogonally as
  \begin{equation}
  \begin{tikzcd}
    W^1(M) \ar[d,"(\Delta+m^2)"',"\sim"] \ar[rr, "p_{M\setminus A}^{\perp}" near start, bend left=60] \ar[rrrr, "p_{M\setminus A}" near end, bend left=30] &[-30pt] = &[-30pt] W^1_{M\setminus A}(M)^{\perp} \ar[d,"(\Delta+m^2)"',"\sim"] &[-30pt] \oplus &[-30pt] W^1_{M\setminus A}(M)\ar[d,"(\Delta+m^2)"',"\sim"]\\ [+20pt]
    W^{-1}(M) \ar[rr, "P_{ A}"' near start, bend right=60] \ar[rrrr, "P_{A}^{\perp}"' near end, bend right=30] &[-30pt] = &[-30pt] W^{-1}_{ A}(M) & \oplus & W^{-1}_{ A}(M)^{\perp} ,
  \end{tikzcd}
  \label{eqn-sobo-decomp}
\end{equation}
which is \textsf{preserved} by the isometric isomorphism~$\Delta+m^2$. In particular, we have
\begin{equation}
  P_{ A}(\Delta+m^2)=(\Delta+m^2)p_{M\setminus A}^{\perp},\quad\textrm{and}\quad P_{ A}^{\perp}(\Delta+m^2)=(\Delta+m^2)p_{M\setminus A},
  \label{eqn-sobo-proj-comm}
\end{equation}
where~$p_{M\setminus A}^{\perp}$,~$p_{M\setminus A}$,~$P_{ A}$, and~$P_{ A}^{\perp}$ are the corresponding orthogonal projections as indicated in the diagram.
\end{lemm}

\begin{proof}
  We just need to show that the image of~$W^1_{M\setminus A}(M)$ under~$\Delta+m^2$ is precisely~$W^{-1}_{ A}(M)^{\perp}$. Indeed, by our definition of the inner products we have
  \begin{equation}
    \sank{(\Delta+m^2)u,v}_{W^{-1}(M)}=\ank{u,v}_{L^2(M)}
    \label{eqn-equal-sobo-distri-pair}
  \end{equation}
  for all~$u\in W^1(M)$,~$v\in W^{-1}(M)$, where the RHS denotes the duality (distributional) pairing. However, the \textsf{annihilator} of~$W^{-1}_{ A}(M)$ under the duality pairing is exactly~$W^1_{M\setminus A}(M)$, see lemma \ref{lemm-sobo-dual}. This translates as
  \begin{equation}
    \big\{u\in W^1(M)~\big|~\sank{(\Delta+m^2)u,v}_{W^{-1}(M)}=0\textrm{ for all }v\in W^{-1}_{ A}(M)\big\}=W^1_{M\setminus A}(M),    \label{}
  \end{equation}
  which is what we desired.
\end{proof}

\begin{corr}
  (\ref{eqn-sobo-decomp}) and (\ref{eqn-sobo-proj-comm}) holds for~$A=\Sigma\subset M$ an embedded closed hypersurface (codimension one submanifold). \hfill~$\Box$
\end{corr}

Indeed, the crucial relation (\ref{eqn-equal-sobo-distri-pair}) together with (\ref{eqn-sobo-proj-comm}) gives

\begin{corr}
  [adjoints] \label{cor-adj-sobo-proj} We have
  \begin{equation}
    \bank{p_{M\setminus A}^{\perp}u,v}_{L^2}=\bank{u,P_A v}_{L^2},\quad\textrm{and}\quad \bank{p_{M\setminus A} u,v}_{L^2}=\bank{u,P_A^{\perp} v}_{L^2},
    \label{}
  \end{equation}
  for~$u\in W^1(M)$ and~$v\in W^{-1}(M)$. \hfill~$\Box$
\end{corr}

\begin{def7}
\label{rem-sobo-proj=dist-res}
Denote~$U:=M\setminus A$. Recall from (\ref{eqn-sobo-open-set-ortho}) that~$W^{-1}(U):=W^{-1}_{ A}(M)^{\perp}$. In fact, corollary \ref{cor-adj-sobo-proj} shows~$P_{M\setminus U}^{\perp}$ coincides with the restriction map~$\rho_{M|U}:\mathcal{D}'(M)\lto \mathcal{D}'(U)$, since if already~$u\in W^1_U(M)$ then~$p_U u=i_*u$,~$i_*:W^1_U(M)\lto W^1(M)$ being inclusion.
\end{def7}

We are thus lead naturally to the following and eventually corollary \ref{cor-sobo-decomp-boundary}.

\begin{corr}\label{cor-sobo-decomp-open}
  Let~$U\subset M$ be an open set and~$F\subset U$ be a closed set. Then~$W^{1}_U(M)$ and~$W^{-1}(U)$ decompose further as
  \begin{equation}
  \begin{tikzcd}
    W^1_U(M) \ar[d,"(\Delta+m^2)"',"\sim"] &[-30pt] = &[-30pt] W^1_{U\setminus F}(M)^{\perp} \ar[d,"(\Delta+m^2)"',"\sim"] &[-30pt] \oplus &[-30pt] W^1_{U\setminus F}(M)\ar[d,"(\Delta+m^2)"',"\sim"]\\ [+20pt]
    W^{-1}(U)&[-30pt] = &[-30pt] P_{M\setminus U}^{\perp}(W^{-1}_{ F}(M)) & \oplus & W^{-1}_{ F}(M)^{\perp} ,
  \end{tikzcd}
\end{equation}
where the orthogonal complements are taken respectively \textsf{inside}~$W^1_U(M)$ and~$W^{-1}(U)$, which is \textsf{preserved} by the isometric isomorphism~$\Delta+m^2$. Also, we have commutation relations similar to (\ref{eqn-sobo-proj-comm}).
\end{corr}

\begin{proof}
  The only point needing explanation is
  \begin{equation}
    (\Delta+m^2)( W^1_{U\setminus F}(M)^{\perp}\cap W_U^1(M))=P_{M\setminus U}^{\perp}(W^{-1}_{ F}(M)).
    \label{}
  \end{equation}
  Indeed, since~$\Delta+m^2$ is bijective, by (\ref{eqn-sobo-decomp}) we have
  \begin{align*}
    \textrm{LHS}&=W^{-1}_{F\cup(M\setminus U)}(M)\cap W_{M\setminus U}^{-1}(M)^{\perp}\\
    &=(W^{-1}_{M\setminus U}(M)\oplus W^{-1}_{F}(M)) \cap W_{M\setminus U}^{-1}(M)^{\perp} \tag{$F$ and~$M\setminus U$ disjoint}\\
    &=\textrm{RHS},
  \end{align*}
  where the direct sum in the second line is non-orthogonal.
\end{proof}

 Now let~$(\Omega,g)$ be compact with (smooth) boundary, smoothly and isometrically embedded in closed~$M$ (if $\Omega$ is a \textit{real tunnelling geometry}, then one choice for~$M$ is the ``isometric double'', see remark \ref{rem-real-tun-geo}) and~$\Sigma$ be embedded in~$ \Omega^{\circ}$. Then in particular

 \begin{corr}\label{cor-sobo-decomp-boundary}
  Corollary \ref{cor-sobo-decomp-open} holds for~$U=\Omega^{\circ}$ and~$\Sigma\subset \Omega^{\circ}$ an embedded closed hypersurface. \hfill~$\Box$
\end{corr}

 \subsection{The Markov Stochastic Decomposition of GFF}\label{sec-stoc-decomp-gff}
 
 \noindent Again suppose~$(M,g)$ is a closed Riemannian manifold and~$\Sigma\subset M$ an embedded closed hypersurface with induced metric. A probabilistic point of view of the results in this section is provided in Powell and Werner \cite{PW} section 4.1.

 \begin{lemm}\label{lemm-sobo-proj-trace-harmonic}
  We have
  \begin{equation}
    p_{M\setminus\Sigma}^{\perp}=\PI_M^{\Sigma}\circ \tau_{\Sigma}
    \label{}
\end{equation}
on~$W^1(M)$, where $p_{M\setminus\Sigma}^{\perp}:W^1(M)\lto W^1_{M\setminus A}(M)^{\perp}$ is the orthogonal projection as in lemma \ref{lemm-sobo-decomp}.
\end{lemm}

\begin{proof}
  We show that for~$f\in C^{\infty}(M)$,~$p_{M\setminus\Sigma}^{\perp} f$ solves the boundary value problem
  \begin{equation}
    \left\{
    \begin{array}{ll}
      (\Delta+m^2)u=0,&\textrm{in }M\setminus\Sigma,\\
      u|_{\Sigma}=\tau_{\Sigma}f,&\textrm{on }\Sigma.
    \end{array}
    \right.
    \label{}
  \end{equation}
  Indeed, the first condition holds since~$\supp((\Delta+m^2)p_{M\setminus\Sigma}^{\perp} f)\subset \Sigma$ by lemma \ref{lemm-sobo-decomp}. To show the second, notice~$\tau_{\Sigma}=0$ on~$W_{M\setminus\Sigma}^1(M)$ by its definition (\ref{eqn-def-sobo-open-closure}). Hence
  \begin{equation}
    \tau_{\Sigma}p_{M\setminus\Sigma}^{\perp}=\tau_{\Sigma}(\one-p_{M\setminus\Sigma})=\tau_{\Sigma},
    \label{}
  \end{equation}
  on~$W^1(M)$. We conclude the proof.
\end{proof}

\begin{prop}\label{prop-stoc-decomp-closed}
  Suppose~$\phi\in \mathcal{D}'(M)$ follows~$\mu_{\mm{GFF}}^M$. Then there is a stochastic decomposition
  \begin{equation}
    \phi=\wh{p_{M\setminus\Sigma}^{\perp}}\phi+\wh{p_{M\setminus\Sigma}}\phi\defeq \phi_{\Sigma}+\phi_{M\setminus\Sigma}^D,
    \label{eqn-stoc-decomp-gff-closed}
  \end{equation}
  into independent random fields~$\phi_{\Sigma}$ and~$\phi_{M\setminus\Sigma}^D$. More precisely, for~$f\in C^{\infty}(M)$ we define the random variables
  \begin{equation}
    \phi_{\Sigma}(f)\defeq \phi(P_{\Sigma}f),\quad\textrm{and}\quad \phi_{M\setminus\Sigma}^D(f)\defeq \phi(P_{\Sigma}^{\perp}f).
    \label{eqn-def-markov-decomp-rv}
  \end{equation}
  Then,~$\phi_{M\setminus\Sigma}^D$ follows the law of~$\mu_{\mm{GFF}}^{M\setminus\Sigma,D}$, while~$\phi_{\Sigma}$ solves the boundary value problem
  \begin{equation}
    \left\{
    \begin{array}{ll}
      (\Delta+m^2)\phi_{\Sigma}=0,&\textrm{in }M\setminus\Sigma,\\
      \phi_{\Sigma}|_{\Sigma}=\tau_{\Sigma}\phi,&\textrm{on }\Sigma,
    \end{array}
    \right.
    \label{}
  \end{equation}
  almost surely.
\end{prop}

\begin{def7}
  The term ``stochastic decomposition'' is borrowed from Bogachev \cite{Bogachev} remark 3.7.7.
\end{def7}

\begin{proof}
  Firstly,~$\phi_{\Sigma}$ and~$\phi_{M\setminus\Sigma}^D$ are independent since by (\ref{eqn-def-markov-decomp-rv}) their Gaussian Hilbert spaces are respectively~$W_{\Sigma}^{-1}(M)$ and~$W_{\Sigma}^{-1}(M)^{\perp}=W^{-1}(M\setminus\Sigma)$. The latter also means~$\phi_{M\setminus\Sigma}^D$ follows~$\mu_{\mm{GFF}}^{M\setminus\Sigma,D}$. The last fact about~$\phi_{\Sigma}$ follows from lemma \ref{lemm-sobo-proj-trace-harmonic} and the uniqueness of the measurable linear extension (\cite{Bogachev} theorem 3.7.6).
\end{proof}

In the same vein using corollary \ref{cor-sobo-decomp-boundary}, we have also a decomposition for the Dirichlet GFF on a domain~$\Omega$ with smooth boundary, with respect to a hypersurface~$\Sigma$ (isometrically) embedded in the interior~$\Omega^{\circ}$.

\begin{prop}\label{prop-stoc-decomp-domain}
  Suppose~$\phi_{\Omega}^D\in \mathcal{D}'(\Omega^{\circ})$ follows~$\mu_{\mm{GFF}}^{\Omega,D}$. Then there is a stochastic decomposition
  \begin{equation}
    \phi=\wh{p_{\Omega^{\circ}\setminus\Sigma}^{\perp}}\phi+\wh{p_{\Omega^{\circ}\setminus\Sigma}}\phi\defeq (\phi_{\Omega}^D)_{\Sigma}+\phi_{\Omega\setminus\Sigma}^D,
    \label{}
  \end{equation}
  into independent random fields~$(\phi_{\Omega}^D)_{\Sigma}$ and~$\phi_{\Omega\setminus\Sigma}^D$. More precisely, for~$f\in W^{-1}(\Omega^{\circ})$ we define the random variables
  \begin{equation}
    (\phi_{\Omega}^D)_{\Sigma}(f)\defeq \phi_{\Omega}^D(P_{\Sigma}f),\quad\textrm{and}\quad \phi_{\Omega\setminus\Sigma}^D(f)\defeq \phi_{\Omega}^D(P_{\Sigma}^{\perp}f).
    \label{eqn-def-markov-decomp-rv-domain-diri}
  \end{equation}
  Then,~$\phi_{\Omega\setminus\Sigma}^D$ follows the law of~$\mu_{\mm{GFF}}^{\Omega\setminus\Sigma,D}$, while~$(\phi_{\Omega}^D)_{\Sigma}$ solves the boundary value problem
  \begin{equation}
    \left\{
    \begin{array}{ll}
      (\Delta+m^2)(\phi_{\Omega}^D)_{\Sigma}=0,&\textrm{in }\Omega^{\circ}\setminus\Sigma,\\
      (\phi_{\Omega}^D)_{\Sigma}|_{\Sigma}=\tau_{\Sigma}\phi_{\Omega}^D,&\textrm{on }\Sigma,\\
      (\phi_{\Omega}^D)_{\Sigma}|_{\partial\Omega}=0, &\textrm{on }\partial\Omega,
    \end{array}
    \right.
    \label{}
  \end{equation}
  almost surely. \hfill~$\Box$
\end{prop}

 The following version of proposition \ref{prop-stoc-decomp-closed} in the case of a dissecting (smooth isometrically embedded) hypersurface~$\Sigma\subset M$ such that~$M\setminus \Sigma=M^{\circ}_+\sqcup M^{\circ}_-$, will be useful in definition \ref{def-inter-over-domain}.

 \begin{lemm}
    We have~$W^{-1}_{M_-}(M)^{\perp}\subset W^{-1}_{M_+}(M)$ and thus~$W^{-1}_{M_+}(M)=W^{-1}_{M_-}(M)^{\perp}\oplus W^{-1}_{\Sigma}(M)$. Similarly~$W^{-1}_{M_-}(M)=W^{-1}_{M_+}(M)^{\perp}\oplus W^{-1}_{\Sigma}(M)$.
  \end{lemm}

  \begin{proof}
    By lemma \ref{lemm-sobo-decomp} with~$A=M_-$ we see that~$W^{-1}_{M_-}(M)^{\perp}=(\Delta+m^2)(W^1_{M_+^{\circ}}(M))$. These distributions are supported in~$M_+$ since~$\Delta+m^2$ is local.
  \end{proof}

\begin{corr}\label{cor-stoc-decomp-disec}
  Suppose~$\phi\in \mathcal{D}'(M)$ follows~$\mu_{\mm{GFF}}^M$. Then there is a stochastic decomposition
 \begin{equation}
   \phi=\wh{p_{M\setminus\Sigma}^{\perp}}\phi+\wh{p_{M_+^{\circ}}}\phi +\wh{p_{M_-^{\circ}}}\phi\defeq \phi_{\Sigma}+\phi_{M_+^{\circ}}^D +\phi_{M_-^{\circ}}^D,
    \label{}
  \end{equation} 
  into independent random fields. More precisely, for~$f\in C^{\infty}(M)$ the random variables are defined by
  \begin{equation}
    \phi_{\Sigma}(f)\defeq \phi(P_{\Sigma}f),\quad \phi_{M_+^{\circ}}^D(f)\defeq \phi(P_{M_-}^{\perp}f), \quad\textrm{and}\quad \phi_{M_-^{\circ}}^D(f)\defeq \phi(P_{M_+}^{\perp}f),
  \end{equation}
  and with~$\phi_{\Sigma}$ the same as in proposition \ref{prop-stoc-decomp-closed},~$\phi_{M_+^{\circ}}^D$ and~$\phi_{M_-^{\circ}}^D$ follows respectively~$\mu_{\mm{GFF}}^{M_+^{\circ},D}$ and~$\mu_{\mm{GFF}}^{M_-^{\circ},D}$. Moreover,~$(\phi_{\Sigma}+\phi_{M_+^{\circ}}^D)(f)=\phi(P_{M_+}f)$,~$(\phi_{\Sigma}+\phi_{M_-^{\circ}}^D)(f)=\phi(P_{M_-}f)$. \hfill~$\Box$
\end{corr}

\subsection{The Bayes Principle Applied to GFF}\label{sec-bayes-gff}

\noindent Let now~$M$ be a closed Riemannian manifold and~$\Sigma_1$,~$\Sigma_2\subset M$ \textit{non-intersecting} isometrically embedded smooth closed hypersurfaces. The goal of this section is to derive the Bayes principle relating the probability laws of the two random fields~$\tau_{\Sigma_1} \phi$ and~$\tau_{\Sigma_2}\phi$ where~$\phi$ is the GFF on~$M$. To avoid convolving with nomenclatures of conditional probabilities we prefer a direct measure theoretic argument, although these are clearly equivalent. Throughout this section we identify continuous linear maps of Cameron-Martin spaces with their measurable extensions to distributional~$Q$-spaces, as well as their induced actions on measures.

To begin with, by proposition \ref{prop-stoc-decomp-closed} we have a stochastic decomposition
\begin{equation}
  \phi=\phi_{\Sigma_1}+\phi_{M\setminus\Sigma_1}^D
  \label{}
\end{equation}
for~$\phi\sim \mu_{\mm{GFF}}^M$, the two components \textit{independent} of each other. Then, apply~$\tau_{\Sigma_2}$ we get
\begin{align}
  \tau_{\Sigma_2}\phi&=\tau_{\Sigma_2}\phi_{\Sigma_1}+\tau_{\Sigma_2}\phi_{M\setminus \Sigma_1}^D=\tau_{\Sigma_2}\PI_{M}^{\Sigma_1}\tau_{\Sigma_1}\phi+\tau_{\Sigma_2}\phi_{M\setminus \Sigma_1}^D \\
  &\defeq \mathcal{M}_{M,2}^1 \tau_{\Sigma_1}\phi +\tau_{\Sigma_2}\phi_{M\setminus \Sigma_1}^D.
  \label{eqn-stoc-decomp-with-M}
\end{align}
Here we define
\begin{equation}
  \mathcal{M}_{M,2}^1\defeq \tau_{\Sigma_2}\PI_{M}^{\Sigma_1}:\left\{
    \begin{array}{rcl}
      C^{\infty}(\Sigma_1)&\lto& C^{\infty}(\Sigma_2),\\
      W^{\frac{1}{2}}(\Sigma_1) &\lto &W^{\frac{1}{2}}(\Sigma_2),
    \end{array}
    \right.
  \label{eqn-trans-op-def}
\end{equation}
called the \textsf{transition operator/propagator}. Define accordingly
\begin{equation}
  \left.
  \begin{array}{rcl}
    \mathcal{G}_{M,2}^1:W^{\frac{1}{2}}(\Sigma_1)\times W^{\frac{1}{2}}(\Sigma_2) &\lto& W^{\frac{1}{2}}(\Sigma_1)\times W^{\frac{1}{2}}(\Sigma_2),\\
    (\varphi_1,h)&\longmapsto & (\varphi_1,h+\mathcal{M}_{M,2}^1\varphi_1).
  \end{array}
  \right.
  \label{}
\end{equation}
called the \textsf{graph operator}.
\begin{lemm}\label{lemm-pre-bayes}
  We have
  \begin{equation}
    \tau_{\Sigma_1\sqcup \Sigma_2}(\mu_{\mm{GFF}}^{M})=(\mathcal{G}_{M,2}^1)_*( \mu_{\DN}^{\Sigma_1,M} \otimes \mu_{\DN,D}^{\Sigma_2,M\setminus \Sigma_1}),
    \label{}
  \end{equation}
  where~$\mu_{\DN}^{\Sigma_1,M}$ and~$\mu_{\DN,D}^{\Sigma_2,M\setminus \Sigma_1}$ are Gaussian measures on~$\mathcal{D}'(\Sigma_1)$ and~$\mathcal{D}'(\Sigma_2)$ with covariances~$(\DN_M^{\Sigma_1})^{-1}$ and~$(\DN_{M\setminus\Sigma_1}^{\Sigma_2,D})^{-1}$, respectively.
\end{lemm}

\begin{proof}
  Following (\ref{eqn-stoc-decomp-with-M}),~$\tau_{\Sigma_2}\phi_{M\setminus \Sigma_1}^D=\tau_{\Sigma_2}\phi-\mathcal{M}_{M,2}^1 \tau_{\Sigma_1}\phi$ is independent of~$\tau_{\Sigma_1}\phi$, and follows the law~$\mu_{\DN,D}^{\Sigma_2,M\setminus \Sigma_1}$, while~$\tau_{\Sigma_1}\phi$ follows the law~$\mu_{\DN}^{\Sigma_1,M}$, both by proposition \ref{prop-induce-law-DN}.
\end{proof}

We could also define the operators~$\mathcal{M}_{M,1}^2$ and~$\mathcal{G}_{M,1}^2$ with the roles of~$\Sigma_1$ and~$\Sigma_2$ switched. Note that lemma \ref{lemm-pre-bayes} is entirely symmetric under the switching of~$\Sigma_1$ and~$\Sigma_2$. Recall the measures~$\mu_{2\mn{D}}^{\Sigma}$ from corollary \ref{corr-rad-niko-dense}.

  \begin{prop}[Bayes Principle for GFF]\label{prop-bayes-gff}
    Let now~$M$ be a closed Riemannian manifold and~$\Sigma_1$,~$\Sigma_2\subset M$ \textit{non-intersecting} isometrically embedded smooth closed hypersurfaces. We have equality of Radon-Nikodym densities
    \begin{align}
      \frac{\dd\tau_{\Sigma_1\sqcup \Sigma_2}(\mu_{\mm{GFF}}^{M}) }{\dd \mu_{2\mn{D}}^{\Sigma_1\sqcup \Sigma_2}}(\varphi_1,\varphi_2)
  &=\frac{\dd  \mu_{\DN}^{\Sigma_1,M}}{\dd\mu_{2\mn{D}}^{\Sigma_1}} (\varphi_1)
  \frac{\dd \big((\mathcal{M}_{M,2}^1 \varphi_1)_*\mu_{\DN,D}^{\Sigma_2,M\setminus \Sigma_1}\big)}{\dd\mu_{2\mn{D}}^{ \Sigma_2}}(\varphi_2) \label{eqn-bayes-gff-1}\\
  &=\frac{\dd  \mu_{\DN}^{\Sigma_2,M}}{\dd\mu_{2\mn{D}}^{\Sigma_2}} (\varphi_2)
  \frac{\dd \big((\mathcal{M}_{M,1}^2 \varphi_2)_*\mu_{\DN,D}^{\Sigma_1,M\setminus \Sigma_2}\big)}{\dd\mu_{2\mn{D}}^{ \Sigma_1}}(\varphi_1). \label{eqn-bayes-gff-2}
    \end{align}
    Here~$(\varphi)_*$ denotes the shift induced by~$\varphi$ on measures as in corollary \ref{cor-cam-mar-shift}.
  \end{prop}

  \begin{proof}
    We just need to prove (\ref{eqn-bayes-gff-1}). In other words,~$(\mathcal{M}_{M,2}^1 \varphi_1)_*\mu_{\DN,D}^{\Sigma_2,M\setminus \Sigma_1}$ is the conditional law of~$\tau_{\Sigma_2}\phi$ provided ``$\tau_{\Sigma_1}=\varphi_1$''. The proof is straightforward. For positive (Borel) measurable functionals~$F\in L^+(\mathcal{D}'(\Sigma_1))$,~$G\in L^+(\mathcal{D}'(\Sigma_2))$, we have by lemma \ref{lemm-pre-bayes}, the change of variables formula, and Fubini's theorem,
\begin{align*}
    \int_{}^{} (F\otimes G)(\varphi_1,\varphi_2)\dd\tau_{\Sigma_1\sqcup \Sigma_2}(\mu_{\mm{GFF}}^{M})&= \iint_{}^{}(F\otimes G)\circ \mathcal{G}_{M,2}^1(\varphi_1,\varphi_2)\dd\mu_{\DN}^{\Sigma_1,M} \otimes \dd \mu_{\DN,D}^{\Sigma_2,M\setminus \Sigma_1} \\
    &=\int_{}^{}F(\varphi_1)\dd  \mu_{\DN}^{\Sigma_1,M}(\varphi_1)\int G(\varphi_2+\mathcal{M}_{M,2}^1 \varphi_1)\dd \mu_{\DN,D}^{\Sigma_2,M\setminus \Sigma_1}(\varphi_2)  \\
  &=\int_{}^{}F(\varphi_1)\dd  \mu_{\DN}^{\Sigma_1,M}\int_{}^{}G(\varphi_2)\dd(\mathcal{M}_{M,2}^1 \varphi_1)_*\mu_{\DN,D}^{\Sigma_2,M\setminus \Sigma_1}. 
\end{align*}
as~$F$ and~$G$ range over all positive measurable functionals, we obtain the result.
  \end{proof}

\begin{def7}
  Applying corollary \ref{corr-rad-niko-dense} and the Cameron-Martin-Girsanov formula (\cite{Bogachev} corollary 2.4.3), one obtains a relation between quadratic forms (the corresponding relation for determinants also works out by BFK),
  \begin{align*}
    &\quad \Bank{\bnom{\varphi_1}{x},\DN_{M}^{\Sigma_1\sqcup\Sigma_2}\bnom{\varphi_1}{x}}_{L^2(\Sigma_1\sqcup\Sigma_2)}\\
    =&\quad \ank{\varphi_1,\DN_M^{\Sigma_1}\varphi_1}_{L^2(\Sigma_1)}+\ank{x,\DN_{M\setminus\Sigma_1}^{\Sigma_2,D}x}_{L^2(\Sigma_2)}\\
    &\quad -2
    \ank{x,\DN_{M\setminus\Sigma_1}^{\Sigma_2,D}\mathcal{M}_{M,2}^1\varphi_1}_{L^2(\Sigma_2)}+\ank{\mathcal{M}_{M,2}^1\varphi_1,\DN_{M\setminus\Sigma_1}^{\Sigma_2,D}\mathcal{M}_{M,2}^1\varphi_1}_{L^2(\Sigma_2)}
  \end{align*}
  which, of course, can also be obtained directly noting
  \begin{align*}
    \ank{\varphi_1,\DN_M^{\Sigma_1}\varphi_1}_{L^2(\Sigma_1)}&=\Bank{\bnom{\varphi_1}{\mathcal{M}_{M,2}^1\varphi_1},\DN_{M}^{\Sigma_1\sqcup\Sigma_2}\bnom{\varphi_1}{\mathcal{M}_{M,2}^1\varphi_1}}_{L^2(\Sigma_1\sqcup\Sigma_2)},\\
\ank{x,\DN_{M\setminus\Sigma_1}^{\Sigma_2,D}x}_{L^2(\Sigma_2)}&=\Bank{\bnom{0}{x},\DN_{M}^{\Sigma_1\sqcup\Sigma_2}\bnom{0}{x}}_{L^2(\Sigma_1\sqcup\Sigma_2)},
  \end{align*}
  et cetera, and expanding the left hand side by writing~$[\smx{\varphi_1\\x}]=[\smx{\varphi_1\\ \mathcal{M}_{M,2}^1\varphi_1}]+[\smx{0\\ x-\mathcal{M}_{M,2}^1\varphi_1}]$.
\end{def7}

\begin{def7}
  In the case~$M=\Sigma\times \mb{R}$ with~$\Sigma_1:=\Sigma\times \{0\}$,~$\Sigma_2:=\Sigma\times \{t\}$, and Dirichlet condition at~$\Sigma\times \{\pm \infty\}$ (namely decay at~$\infty$), we have the explicit expression
  \begin{equation}
    \mathcal{M}^1_{M,2}=\me^{-t\mn{D}_{\Sigma}},
    \label{}
  \end{equation}
  where~$\mn{D}_{\Sigma}=(\Delta_{\Sigma}+m^2)^{1/2}$. Also, in this case~$\mn{D}_{\Sigma}=\DN_{M_+}^{\Sigma}$,~$M_+=\Sigma\times [0,+\infty)$.
\end{def7}

\subsection{Locality of the $P(\phi)$ interaction and the Restricted GFF}\label{sec-locality}

In this subsection we pronounce what is meant by the \textsf{locality} of the $P(\phi)$-interaction and prove it. Let~$M$ be a closed Riemannian surface and~$\Omega^{\circ}\subset M$ a domain with smooth boundary~$\partial\Omega$. Denote~$\Omega=\ol{\Omega^{\circ}}$. The formulation will be based on the domain $\sigma$-algebras defined by the GFF random varibales $\ank{-,f}_{L^2}$, $f\in C^{\infty}(M)$. That is, for~$A\subset M$ a closed set, we define the~$\sigma$-algebra
\begin{equation}
  \mathcal{O}_A\defeq \sigma\left( \phi(f)~|~f\in W_A^{-1}(M) \right).
  \label{}
\end{equation}

\begin{def7}
  In general one could consider \textsf{regular} domains~$\Omega^{\circ}$. A domain is called \textsf{regular} if~$W^{-1}_{\Omega^{\circ}}(M)=W^{-1}_{\Omega}(M)$, the two spaces defined by (\ref{eqn-def-sobo-open-closure}) and (\ref{eqn-def-sobo-closed-support}). See also Simon \cite{Sim2} page 267.
\end{def7}
Suppose~$\chi \in C^{\infty}(M)$. We define
\begin{equation}
  S_{\Omega,\varepsilon,\chi}(\phi)\defeq \int_{M}^{}\chi(x) {:}P(\phi_{\Omega,\varepsilon}^{D}(x)+\phi_{\partial\Omega,\varepsilon}(x)){:}\dd V_M(x),
  \label{}
\end{equation}
where
\begin{equation}
  \phi_{\Omega,\varepsilon}^{D}(x)\defeq \phi(P_{\ol{\Omega^c}}^{\perp} E_{\varepsilon}\delta_x),\quad \phi_{\partial\Omega,\varepsilon}(x)\defeq\phi(P_{\partial\Omega}E_{\varepsilon}\delta_x),\quad \phi\sim \mu_{\mm{GFF}}^M.
  \label{}
\end{equation}

\begin{prop}[locality]\label{prop-locality}
  Whenever~$\chi\in C_c^{\infty}(\Omega^{\circ})$, we have, for any fixed Wick ordering~${:}\bullet{:}$,
  \begin{equation}
    S_{M,\chi}(\phi)=\int_{M}^{}\chi(x) {:}P(\phi(x)){:}\dd V_M(x)=\lim_{\varepsilon\to 0}S_{\Omega,\varepsilon,\chi}(\phi),
    \label{eqn-decoupling-eqv-rv}
  \end{equation}
  and hence~$S_{M,\chi}$ is~$\mathcal{O}_{\Omega}$-measurable. In particular,
  \begin{equation}
    \int_{\Omega}^{}{:}P(\phi(x)){:}\dd V_{\Omega}= \int_{\Omega}^{} {:}P(\phi^{D}_{\Omega}(x)+\phi_{\partial\Omega}(x)){:}\dd V_{\Omega}
    \label{}
  \end{equation}
  and is~$\mathcal{O}_{\Omega}$-measurable.
\end{prop}

\begin{proof}
  Without loss of generality we suppose~$\supp \chi$ stays~$\delta$-away from~$\partial\Omega$ for some~$\delta>0$. Then corollary \ref{cor-nelson-all-Lp-cutoff-limit} and approximation in~$L^4$ of a general~$\chi$ as well as~$1_{\Omega}$ gives the result. We note that the support of~$E_{\varepsilon}\delta_x$ is contained in the~$\varepsilon$-ball around~$x$. Then whenever~$\varepsilon<\delta$ for~$x\in \supp\chi$ we have~$\supp(E_{\varepsilon}\delta_x)\subset \Omega$, hence
  \begin{equation}
    P_{\Omega}^{\perp}E_{\varepsilon}\delta_x =0,
    \label{}
  \end{equation}
  and
  \begin{equation}
    \phi_{\varepsilon}(x)\equiv \phi(P_{\ol{\Omega^c}}^{\perp} E_{\varepsilon}\delta_x) +\phi(P_{\partial\Omega}E_{\varepsilon}\delta_x),
    \label{}
  \end{equation}
  both~$\mathcal{O}_{\Omega}$-measurable. Thus
  \begin{equation}
    {:}P(\phi_{\Omega,\varepsilon}^{D}(x)+\phi_{\partial\Omega,\varepsilon}(x)){:}\equiv {:}P(\phi_{\varepsilon}(x)){:}
    \label{}
  \end{equation}
  for~$x\in \supp \chi$ and~$\varepsilon<\delta$, and is~$\mathcal{O}_{\Omega}$-measurable. Therefore (\ref{eqn-decoupling-eqv-rv}) holds by proposition \ref{prop-nelson-main}.
\end{proof}

We emphasize here that for our purposes we employ the Wick ordering~${:}\bullet{:}_0$ of section \ref{sec-change-wick} in defining~$S_{M,\chi}$. This is to say
\begin{equation}
  {:}\phi_{\varepsilon}(x)^{2n}{:}=\sum_{j=0}^{n} \frac{(-1)^j (2n)!}{(2n-2j)! j! 2^j}C_{\varepsilon}(x)^j \phi_{\varepsilon}(x)^{2n-2j},
  \label{}
\end{equation}
where
\begin{equation}
  C_{\varepsilon}(x)=\iint E_{\varepsilon}(x,y)C_0(y,z)E_{\varepsilon}(z,x)\dd V_M(y)\dd V_M(z).
  \label{}
\end{equation}
Since~$E_{\varepsilon}(x,-)$ is supported in an~$\varepsilon$-ball around~$x$, and~$C_0(y,z)$ depends only on the geometry of~$M$ resctricted to a convex neighborhood of~$y$,~$z$, the term~$C_{\varepsilon}(x)$ depends only on the geometry of~$M$ locally near~$x$. This means that under~${:}\bullet{:}_0$, once~$\supp \chi\subset \Omega$, the limiting (integrated) random variable~$S_{M,\chi}$, in addition to being~$\mathcal{O}_{\Omega}$-measurable, is in fact fully determined with knowledge of the metric~$g|_{\Omega}$ restricted to~$\Omega$. This allows the freedom of choosing the ambient manifold~$M$ where~$\Omega$ isometrically embeds in defining the interaction over~$\Omega$ (see definition \ref{def-inter-over-domain}).

Now we repeat proposition \ref{prop-locality} in a different way:
\begin{corr}\label{cor-interact-locality}
   Whenever~$\chi\in C_c^{\infty}(\Omega^{\circ})$, for any fixed Wick ordering~${:}\bullet{:}$,
   \begin{equation}
     \mb{E}_{\mm{GFF}}^M[S_{M,\chi}|\mathcal{O}_{\Omega}]=S_{M,\chi},
     \label{}
   \end{equation}
   and~$S_{M,\chi}$ is independent of~$\phi_{\Omega^c}^D$, that is, of~$\phi(P_{\Omega}^{\perp} f)$ for all~$f\in C^{\infty}(M)$ or~$\phi((\Delta+m^2)f)$ for all~$f\in C_c^{\infty}(\Omega^c)$. Moreover, since~$\Omega$ is regular,~$S_{M,\chi}$ is the limit in~$L^2(\mu_{\mm{GFF}}^M)$ of polynomials of random variables of the form~$\phi(f)$ with~$f\in C_c^{\infty}(\Omega^{\circ})$ via the It\^o-Wiener-Segal isomorphism.\hfill~$\Box$
\end{corr}

This result motivates

\begin{deef}\label{def-res-gff}
  Let~$\Omega^{\circ}\subset M$ be a regular domain. The \textsf{(massive) Gaussian Free Field} over~$\Omega^{\circ}$ \textsf{restricted from~$M$} is the centered Gaussian process indexed by~$C_c^{\infty}(\Omega^{\circ})$, with covariance
  \begin{equation}
    \mb{E}[\phi(f)\phi(h)]=\ank{f,h}_{W^{-1}_{\Omega}(M)}
    \label{}
  \end{equation}
  for any~$f$,~$h\in C_c^{\infty}(\Omega^{\circ})$. We denote it by $\phi|_{\Omega}$.
\end{deef}

Since the inner product of~$W^{-1}_{\Omega}(M)$ depends on~$M$, the restricted GFF does not make sense on~$\Omega^{\circ}$ alone. However, by corollary \ref{cor-stoc-decomp-disec}, it is equal in law to~$\phi_{\partial\Omega}+\phi_{\Omega}^D$, where~$\phi_{\Omega}^D$ does make sense over~$\Omega^{\circ}$ and the law of~$\phi_{\partial\Omega}=\PI_{\Omega}^{\partial\Omega}(\tau_{\partial\Omega}\phi)$ is determined via~$\PI_{\Omega}^{\partial\Omega}$ by that of a boundary data over~$\partial\Omega$.

Going back to corollary \ref{cor-interact-locality} we see that for~$\chi\in C_c^{\infty}(\Omega^{\circ})$, the interaction~$S_{M,\chi}$ can now be defined as a random variable of the sample~$\phi|_{\Omega}\in\mathcal{D}'(\Omega^{\circ})$. Equally, it is also a random variable over~$\mathcal{D}'(\Omega^{\circ})\times \mathcal{D}'(\partial\Omega)$ equipped with~$\mu_{\mm{GFF}}^{\Omega,D}\otimes \mu_C^{\partial\Omega}$ where~$\mu_C^{\partial\Omega}$ is some probability measure on~$\mathcal{D}'(\partial\Omega)$ mutually absolutely continuous with respect to~$\tau_{\partial\Omega}(\mu_{\mm{GFF}}^M)$ and on the same~$\sigma$-algebra. We consider only a few very specific candidates for~$\mu_C^{\partial\Omega}$.

Now suppose we start with a Riemannian surface~$(\Omega,g)$ with totally geodesic boundary~$\partial\Omega$. With the help of the isometric double~$|\Omega|^2$ discussed in remark \ref{rem-real-tun-geo} we define
\begin{deef}\label{def-inter-over-domain}
  The \textsf{interaction over~$\Omega$} is
  \begin{equation}
    S_{\Omega}(\phi_{\Omega}^D|\varphi)=S_{\Omega}(\phi_{\Omega}^D+\PI_{\Omega}^{\partial\Omega}\varphi)\defeq S_{|\Omega|^2,1_{\Omega}}(\phi_{\Omega}^D+\PI_{\Omega}^{\partial\Omega}\varphi)
    \label{}
  \end{equation}
  as an~$L^2$-random variable over~$\mathcal{D}'(\Omega^{\circ})\times \mathcal{D}'(\partial\Omega)$ equipped with~$\mu_{\mm{GFF}}^{\Omega,D}\otimes \mu_{2\DN}^{\partial\Omega,\Omega}$, where the latter is the Gaussian measure with covariance operator~$\frac{1}{2}(\DN_{\Omega}^{\partial\Omega})^{-1}$.
\end{deef}

\begin{def7}\label{rem-interp-inter-over-omega}
    We emphasize here that the variable~$S_{\Omega}(\phi_{\Omega}^D|\varphi)$ must be understood for~$\phi_{\Omega}^D$ and~$\varphi$ \textit{both} random, with the law of~$\varphi$ mutually absolutely continuous with respect to~$\mu_{2\DN}^{\partial\Omega,\Omega}$ (thus~$\mu_{2\mn{D}}^{\partial\Omega}$ is a valid candidate). Thus if one asks how much can one ``fix'' $\varphi$ as one interprets~$S_{\Omega}(\phi_{\Omega}^D|\varphi)$, the answer is that it makes sense as a random variable of~$\phi_{\Omega}^D$ alone for almost every fixed~$\varphi$ under $\mu_{2\mn{D}}^{\partial\Omega}$ (but this full-measure set with respect to $\mu_{2\mn{D}}^{\partial\Omega}$ is generally unknown). Said from a slightly different perspective, the expression~$S_{\Omega}(\phi)$ for a generic~$\phi\in \mathcal{D}'(\Omega^{\circ})$ makes sense as a random variable only when~$\phi$ follows a law mutually absolutely continuous with respect to~$\mu_{\mm{GFF}}^{\Omega,D} * (\PI_{\Omega}^{\partial\Omega})_*\mu_{2\mn{D}}^{\partial\Omega}$, where ``$*$'' denotes convolution product. Statements involving~$S_{\Omega}(\phi_{\Omega}^D|\varphi)$ in the sequel should be understood in this sense. 
  \end{def7}

  The following result will be useful in lemma \ref{lemm-amplitude-express}.

\begin{prop}\label{prop-decoupling}
  For any fixed Wick ordering~${:}\bullet{:}$ we have
  \begin{equation}
    \int_{M}^{}{:}P(\phi(x)){:}\dd V_M=\int_{\Omega}^{} {:}P(\phi^{D}_{\Omega}(x)+\phi_{\partial\Omega}(x)){:}\dd V_{\Omega} +\int_{\Omega^c}^{} {:}P(\phi_{\Omega^c}^{D}(x)+\phi_{\partial\Omega}(x)){:}\dd V_{\Omega^c}
    \label{eqn-decoup-inter}
  \end{equation}
  in~$L^p(\mu_{\mm{GFF}}^M)$ for all~$1\le p<\infty$ and in particular pointwise almost surely, as a result
  \begin{equation}
    \me^{- \int_{M}^{}{:}P(\phi(x)){:}\dd V_M}=\me^{-\int_{\Omega}^{} {:}P(\phi^{D}_{\Omega}(x)+\phi_{\partial\Omega}(x)){:}\dd V_{\Omega}} \me^{-\int_{\Omega^c}^{} {:}P(\phi_{\Omega^c}^{D}(x)+\phi_{\partial\Omega}(x)){:}\dd V_{\Omega^c}}
    \label{}
  \end{equation}
  pointwise almost surely and thus also in~$L^1(\mu_{\mm{GFF}}^M)$.
\end{prop}

\begin{def7}
  Note that the equality (\ref{eqn-decoup-inter}) says nothing about pointwise almost sure convergence of the mollified sequences~$S_{M,\varepsilon}$,~$S_{\Omega,\varepsilon}$, and~$S_{\Omega^c,\varepsilon}$. In fact, only separate subsequences in~$\varepsilon$ will converge pointwise almost surely to the three respective terms in (\ref{eqn-decoup-inter}). 
\end{def7}

\begin{proof}
  Consider a smooth partition of unity
  \begin{equation}
    1_M=\chi_{\Omega}+\chi_{\partial\Omega}+\chi_{\Omega^c},
    \label{eqn-decoup-pf-part-unity}
  \end{equation}
  where~$\chi_{\Omega}$ and~$\chi_{\Omega^c}$ are supported in the interiors of~$\Omega$ and~$\Omega^c$ respectively and~$\chi_{\partial\Omega}$ is supported near~$\partial\Omega$. Note then that (\ref{eqn-decoup-pf-part-unity}) holds as well in~$L^4(M)$ so by corollary \ref{cor-nelson-all-Lp-cutoff-limit} we have
  \begin{equation}
    \int_{M}^{}{:}P(\phi(x)){:}\dd V_M= \int_{M}^{}\chi_{\Omega}{:}P(\phi(x)){:}\dd V_M+\int_{M}^{}\chi_{\partial\Omega}{:}P(\phi(x)){:}\dd V_M+\int_{M}^{}\chi_{\Omega^c}{:}P(\phi(x)){:}\dd V_M
    \label{}
  \end{equation}
  in~$L^p(\mu_{\mm{GFF}}^M)$ for~$1\le p<\infty$. Now take
  \begin{equation}
    \chi_{\Omega}\to 1_{\Omega},\quad \chi_{\Omega^c}\to 1_{\Omega^c},\quad\textrm{and }\chi_{\partial\Omega}\to 0,
    \label{}
  \end{equation}
  in~$L^4(M)$, keeping the equality (\ref{eqn-decoup-pf-part-unity}) in the process. Then
  \begin{align*}
    \int_{M}^{}{:}P(\phi(x)){:}\dd V_M &=\int_{M}^{}1_{\Omega}{:}P(\phi(x)){:}\dd V_M+\int_{M}^{}1_{\Omega^c}{:}P(\phi(x)){:}\dd V_M \\
    &=\int_{\Omega}^{} {:}P(\phi^{D}_{\Omega}(x)+\phi_{\partial\Omega}(x)){:}\dd V_{\Omega} +\int_{\Omega^c}^{} {:}P(\phi_{\Omega^c}^{D}(x)+\phi_{\partial\Omega}(x)){:}\dd V_{\Omega^c}
  \end{align*}
  in~$L^p(\mu_{\mm{GFF}}^M)$ for~$1\le p<\infty$, and we finish the proof.
\end{proof}

\section{Presenting the $P(\phi)_2$ Model as a Segal Theory}\label{sec-segal-main}

\subsection{Description of Segal's Rules as Pertains to the Model}\label{sec-segal-descript}

\noindent The prototype of what we propose here is definition (4.4) on page 460 of \cite{Segal}, then we relate to the ``transfer operator'' formalism on page 455-456 (lemma \ref{lemm-segal-transfer}). We will consider real (separable) Hilbert spaces. We denote the category of such spaces with Hilbert-Schmidt operators as morphisms by~$\mn{Hilb}_{\mb{R}}$. The reason for Hilbert-Schmidt is \cite{Sim3} theorem 3.8.5.

As a preliminary consideration, let $\mathcal{C}$ denote the ``category'' whose objects are finite disjoint unions $\Sigma_j=\bigsqcup_i \mb{S}^1_i$ of Riemannian circles, each of which determined by its radius~$R_i$, and for any finite collection of such objects~$(\Sigma_j)_j$ we say that an \textsf{cobordism/unoriented morphism} among them is simply an orientable Riemannian surface with a \textit{totally geodesic} boundary~$\partial\Omega$ which is identified with~$\bigsqcup_j \Sigma_j$ via an isometry (see also remark \ref{rem-smooth-gluing}). With this in mind we can consider the following definition. Morphisms will be oriented once we make the transfer operator connection.

\begin{deef}
  [\cite{Segal} page 460] \label{def-segal-2} A 2d Riemannian QFT is a correspondence~$\mathcal{C}\lto\mn{Hilb}_{\mb{R}}$, such that
  \begin{enumerate}[(i)]
    \item to each oriented Riemannian circle~$\mb{S}^1_R$ of radius~$R$ there is associated a Hilbert space~$\mathcal{H}_R$, and to the disjoint union~$\bigsqcup_{i\in I} \mb{S}_i^1$ there is associated the tensor product~$\bigotimes_{i\in I} \mathcal{H}_i$ of the corresponding single-circle Hilbert spaces; in the degenerate case~$R=0$, we let~$\mathcal{H}_0=\mb{R}$;
    \item (\textsf{operator-reflection}) to each Riemannian surface~$\Omega$ such that~$\partial\Omega=\Sigma$, without distinguishing the orientations on the components, we associate an element~$\mathcal{A}_{\Omega}\in \mathcal{H}_{\Sigma} $; in the degenerate case where~$\Omega=M$ is closed, we associate a real number~$Z_{M}\in \mb{R}$, called the \textsf{partition function};
    \item (\textsf{sewing-trace}) suppose~$\Sigma_i$ and~$\Sigma_j$ are two connected components of~$\Sigma$ that are isometric, and let~$\rho:\Sigma_i\lto\Sigma_j$ be an (orientation reversing) isometry. Then
      \begin{equation}
	\mathcal{A}_{\Omega/\rho}=\ttr_{\rho}(\mathcal{A}_{\Omega}),
	\label{}
      \end{equation}
      where~$\Omega/\rho$ is the surface obtained from~$\Omega$ by gluing~$\Sigma_i$ with~$\Sigma_j$ along~$\rho$, and~$\ttr_{\rho}$ is the \textsf{trace map} such that writing~$\Sigma=\Sigma_i\sqcup\Sigma_j\sqcup\Sigma'$ (possibly~$\Sigma'=\{\mm{pt}\}$),~$\ttr_{\rho}$ is the map
      \begin{equation}
	\left.
	\begin{array}{rcl}
	  \ttr_{\rho}:\mathcal{H}_{\Sigma_i} \otimes \mathcal{H}_{\Sigma_j} \otimes \mathcal{H}_{\Sigma'}  &\lto&\mathcal{H}_{\Sigma'}  ,\\
	  F\otimes G\otimes H&\longmapsto &\ank{\rho_* F, G}_{\Sigma_j} H,
	\end{array}
	\right.
	\label{eqn-segal-general-trace}
      \end{equation}
      with~$\ank{-,-}_{\Sigma_j} $ being the (real) inner product on~$\mathcal{H}_{\Sigma_j} $.
    \end{enumerate}
\end{deef}

\begin{def7}
  A word of caution should be said immediately concerning (\ref{eqn-segal-general-trace}). As written,~$\ttr_{\rho}$ would not extend continuously to the whole of~$\mathcal{H}_{\Sigma_i} \otimes \mathcal{H}_{\Sigma_j} \otimes \mathcal{H}_{\Sigma'} $ for ``infinitely entangled'' states in~$\mathcal{H}_{\Sigma_i} \otimes \mathcal{H}_{\Sigma_j}$. However, for us each~$\mathcal{H}_{\Sigma}$ is~$L^2(Q_{\Sigma},\mu_{\Sigma})$ for a probability space~$(Q_{\Sigma},\mu_{\Sigma})$, and thus by discussion on \cite{RSim} page 52,~$\mathcal{H}_{\Sigma_i} \otimes \mathcal{H}_{\Sigma_j} \otimes \mathcal{H}_{\Sigma'} $ is represented as~$L^2(Q\times Q\times Q',\mu\otimes \mu\otimes \mu')$, where we identify~$(Q_{\Sigma_i},\mu_{\Sigma_i})$ and~$(Q_{\Sigma_j},\mu_{\Sigma_j})$ with~$(Q,\mu)$ via the isometry. The action of~$\ttr_{\rho}$ should be understood as
  \begin{equation}
    \ttr_{\rho}:\mathcal{A}_{\Omega}(x,y,z) \longmapsto \int_{}^{}\mathcal{A}_{\Omega}(x,x,z)\dd \mu(x),
    \label{eqn-def-flat-trace}
  \end{equation}
  when the latter integral converges, which we will show to happen for our model in sections \ref{sec-trace-axiom} and \ref{sec-gluing-final}. This is analogous to the ``flat trace'' of Atiyah and Bott \cite{AB}. It can be shown easily that (\ref{eqn-def-flat-trace}) coincides with (\ref{eqn-segal-general-trace}) on ``finitely entangled'' states of~$\mathcal{H}_{\Sigma_i} \otimes \mathcal{H}_{\Sigma_j} \otimes \mathcal{H}_{\Sigma'} $, that is, finite linear combinations~$\sum_i F_i\otimes G_i\otimes H_i$. However, the finiteness of neither (\ref{eqn-def-flat-trace}) nor (\ref{eqn-segal-general-trace}) would imply the operator~$U_{\Omega}$ of (\ref{eqn-def-segal-transfer}) is trace class proper.
\end{def7}

 \begin{def7}\label{rem-trace-and-diagonal}
   There is in \textit{general} no relation between the trace of an integral operator and the integral along diagonal of its Schwartz kernel. For there to exist a relation the Schwartz kernel needs to satisfy some continuity condition (as the diagonal has measure zero in the product space), or we must be in the case of \cite{Sim3} theorem 3.8.5. See Simon \cite{Sim3} section 3.11, also Vershik, Petrov and Zatitskiy \cite{PVZ} section 3.3.
 \end{def7}

Now we manufacture Hilbert-Schmidt operators in~$\mn{Hilb}_{\mb{R}}$ out of definition \ref{def-segal-2} with the help of \cite{Sim3} theorem 3.8.4. Let~$\Sigma_{\mm{in}}$,~$\Sigma_{\mm{out}}$ be two objects in~$\mathcal{C}$, now \textit{oriented}. We say that an unoriented morphism~$\Omega$ among~$\Sigma_{\mm{in}}$,~$\Sigma_{\mm{out}}$ is in~$\Mor(\Sigma_{\mm{in}},\Sigma_{\mm{out}})$ if the isometry~$\Sigma_{\mm{in}}\sqcup \Sigma_{\mm{out}}\lto \partial\Omega$ identifies the orientation of~$\Sigma_{\mm{in}}$ with that induced by an \textit{inward pointing} normal on~$\partial\Omega$, and that of~$\Sigma_{\mm{out}}$ an \textit{outward pointing} normal.

 Theorem 3.8.4 of \cite{Sim3} and definition \ref{def-segal-2} now implies
\begin{lemm}\label{lemm-segal-transfer}
  If~$\Omega\in \Mor(\Sigma_{\mm{in}},\Sigma_{\mm{out}})$, then the \textsf{Segal transfer operator}~$U_{\Omega}$ defined by
 \begin{equation}
    \left.
    \begin{array}{rcl}
      U_{\Omega}: \mathcal{H}_{\Sigma_{\mm{in}}} &\lto& \mathcal{H}_{\Sigma_{\mm{out}}},\\
      F &\longmapsto & \ttr_{\rho}(F\otimes \mathcal{A}_{\Omega}),
    \end{array}
    \right.
    \label{eqn-def-segal-transfer}
  \end{equation} 
  where~$\ttr_{\rho}$ is the map of (\ref{eqn-def-flat-trace}) with~$\Sigma_i=\Sigma_j=\Sigma_{\mm{in}}$,~$\Sigma'=\Sigma_{\mm{out}}$, and~$\rho=\one$, is Hilbert-Schmidt. Moreover,
  \begin{enumerate}[(i)]
    \item if~$\Omega_1\in\Mor(\Sigma_1,\Sigma_2)$,~$\Omega_2\in \Mor(\Sigma_2,\Sigma_3)$,~and $\Omega_2\cup_2 \Omega_1$ the Riemannian connected sum by gluing components corresponding to~$\Sigma_2$, then
      \begin{equation}
    U_{\Omega_2\cup_2 \Omega_1}=U_{\Omega_2}\circ U_{\Omega_1}.
    \label{eqn-def-segal-sew}
  \end{equation}
  \item suppose~$\Omega\in \Mor(\Sigma,\Sigma)$, with $\Sigma$ identified with itself along an isometry $\rho$, and denote by~$\check{\Omega}$ the closed surface obtained by gluing these components together, then
  \begin{equation}
    Z_{\check{\Omega}}=\ttr_{\rho} (U_{\Omega}).
    \label{}
  \end{equation}
\item if~$\Omega\in \Mor(\Sigma_1,\Sigma_2)$, denote by~$\Omega^*\in \Mor(\Sigma_2,\Sigma_1)$ the surface obtained by reversing the orientations of boundaries of~$\Omega$ without changing the orientation of~$\Omega$, then
  \begin{equation}
    U_{\Omega^*}=U_{\Omega}^{\dagger},
    \label{}
  \end{equation}
  with~$U_{\Omega}^{\dagger}$ denoting the (real) adjoint of~$U_{\Omega}$. \hfill~$\Box$
  \end{enumerate}
\end{lemm}

\begin{figure}
    \centering
    \includegraphics[width=0.8\linewidth]{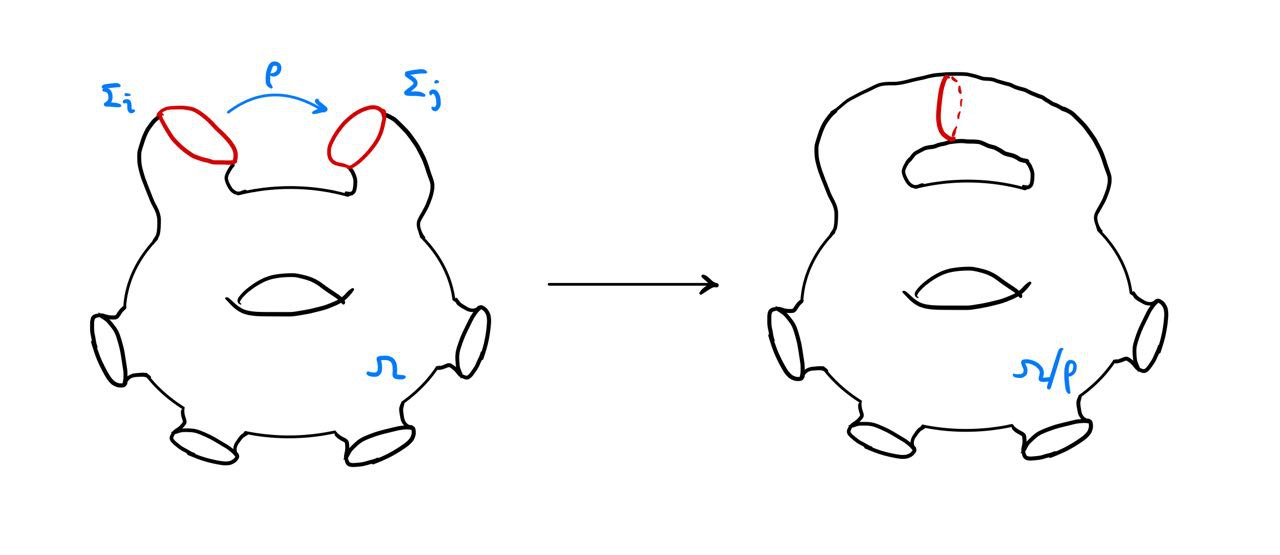}
    \caption{sewing}
    \label{fig-segal-gluing}
\end{figure}

\begin{def7}\label{rem-segal-proj}
  As alluded to at the bottom of subsection \ref{sec-BFK}, one could allow each~$\mathcal{A}_{\Omega}$ to be a \textit{ray} instead of a fixed vector in $\mathcal{H}_{\partial\Omega}$, which eventually enables an arbitrary (nonzero) constant to be included in (\ref{eqn-def-segal-sew}), obtaining the so-called \textsf{projective} version of the axiom. This is reasonable since quantum states are rays, not vectors.
\end{def7}

\begin{def7}
  An interesting remark appears at the bottom of page 457 of \cite{Segal}, namely the conjecture of Friedan that a theory in the sense of definition \ref{def-segal-2} is in fact completely determined by its restriction to \textit{closed surfaces}, namely the knowledge of the numbers~$Z_M$ for all closed~$M$. While far from proving this conjecture, we shall see in section \ref{sec-amplitude} that the correct definition of the amplitude ``derives'' very naturally from a special case of the trace axiom ((iii) of definition \ref{def-segal-2}). It seems certain that if one has knowledge not just of the numbers~$Z_M$ but the \textit{measures}~$\mu_{P(\phi)}^M$ as in (\ref{eqn-def-mes-gibbs-heu}), then the theory would be determined.
\end{def7}

\begin{def7}\label{rem-smooth-gluing}
To have safety of smooth gluing while taking into account at the same time our reflection constructions (the isometric double) required in section \ref{sec-amplitude}, one should enhance the objects~$\Sigma$ with \textsf{symmetric two-sided collars}~$\Sigma\times (-\varepsilon,\varepsilon)$ equipped with symmetric metrics making~$\Sigma\times \{0\}$ geodesic, and a cobordism~$\Omega$ as described above must allow each~$\Sigma_j$ as a component of~$\partial\Omega$ to have a tubular neighborhood isometric to~$\Sigma_j\times (-\varepsilon,0]$. Here \textsf{symmetric} means having an isometric involution~$\Theta$ that exchanges~$\Sigma\times\{\pm t\}$ and fixes exactly~$\Sigma\times\{0\}$ (actually the reflection would ensure~$\Sigma\times\{0\}$ being geodesic, see \cite{Ritter} section 2.1.1 and Alekseevsky et al \cite{AKLM}). In fact, one should consider all proper ingredients that make up what is called the \textsf{Riemannian bordism category} as described in detail in section 6.2 of H. Hohnhold, S. Stolz, and P. Teichner \cite{HST}. Perhaps a better consideration is not fixed collars but \textit{germs} of collars; see also the discussion in Kontsevich and Segal \cite{KS21} following definition 3.1. However, the problem of finding the right underlying category is somewhat orthogonal to the problems considered in this article.
\end{def7}

All these said, we now restate theorem \ref{thrm-intro-main-1} in the introduction in the following precise manner.

\begin{thrmbisbis}{thrm-intro-main-1}
\label{thrm-main-1-precise}
  The Hilbert spaces~$\mathcal{H}_R$ given by (\ref{eqn-def-hilb-space-circ}), the amplitudes~$\mathcal{A}_{\Omega}^P$ given by (\ref{eqn-def-amp-omega}) and the partition functions~$Z_M$ given by (\ref{eqn-def-part-func-rigor}) together gives a 2d Riemannian QFT satisfying the requirements of definition \ref{def-segal-2}.
\end{thrmbisbis}

The proof of this theorem constitutes the main body of subsections \ref{sec-trace-axiom} and \ref{sec-gluing-final}.

\subsection{The Hilbert Spaces}\label{sec-segal-hilb-space}

\noindent Now we associate a real Hilbert space~$\mathcal{H}_{R}$ to the Riemannian circle~$\mb{S}^1_R$ of radius~$R$. Let
\begin{equation}
    \mn{D}_R\defeq(\Delta_{\mb{S}^1_R}+m^2)^{1/2}
\end{equation}
be the positive square root of the positive Helmholtz operator on~$\mb{S}^1_R$. The circle having radius~$R$ would mean that we take~$\Delta_{\mb{S}^1_R}:=-\partial_{\theta}^2$ when the metric is~$R\dd\theta$, parametrized by the arc length~$\theta$. Denote by~$\mu_{(2\mn{D}_R)^{-1}}$ the Gaussian measure on~$\mathcal{D}'(\mb{S}^1_R)$ with covariance~$(2\mn{D}_R)^{-1}$ (for convenience, we just write~$\mu_{2\mn{D}}$). Then we define
\begin{equation}
  \mathcal{H}_R \defeq L^2_{\mb{R}}(\mathcal{D}'(\mb{S}^1_R),\wh{\mu}_{2\mn{D}}).
  \label{eqn-def-hilb-space-circ}
\end{equation}

 Here we take~$\wh{\mu}_{2\mn{D}}:=\mu_{2\mn{D}}\cdot\detz(2\mn{D}_{R})^{-\frac{1}{2}}$, namely the probability measure scaled by the positive finite constant~$\detz(2\mn{D}_{R})^{-\frac{1}{2}}$. Indeed, this same method could be applied directly to (finite) disjoint unions~$\Sigma=\bigsqcup_i \mb{S}_i^1$ and produce a measure on the corresponding~$\mathcal{D}'(\Sigma)$, and this is compatible with (i) of definition \ref{def-segal-2}. Indeed,
\begin{equation}
  \mathcal{D}'(\mb{S}^1_1\sqcup \mb{S}^1_2)=\mathcal{D}'(\mb{S}^1_1)\oplus \mathcal{D}'(\mb{S}^1_2)
  \label{}
\end{equation}
and the operator~$\mn{D}_{\Sigma}$ acts diagonally, and~$\mu_{2\mn{D}}^{\mb{S}_1\sqcup\mb{S}_2}=\mu_{2\mn{D}}^{\mb{S}_1}\otimes \mu_{2\mn{D}}^{\mb{S}_2}$ in view of lemma \ref{lemm-gaus-field-disj-indep} (for Gaussian fields). The real~$L^2$ space would then be the tensor product (\cite{RSim} page 52).

\subsection{Amplitudes = Schwartz Kernels} \label{sec-amplitude}

\noindent Let~$\Omega$ be an unoriented morphism among~$(\Sigma_j)_j$. We denote~$\Sigma:=\bigsqcup_j \Sigma_j$, and seek to define~$\mathcal{A}_{\Omega}\in \mathcal{H}_{\Sigma}$. We consider gluing~$\Omega$ with itself along an orientation reversing isometry~$\rho:\Sigma\lto \Sigma$, forming the isometric double~$|\Omega|^2=\Omega^*\cup_{\Sigma}\Omega$ (here~$\Omega^*$, as in (iii) of lemma \ref{lemm-segal-transfer}, denote the copy of~$\Omega$ with coorientation of~$\Sigma$ reversed). Suppose we have defined~$\mathcal{A}_{\Omega}$ satisfying definition \ref{def-segal-2}, then (iii) would imply
\begin{align*}
  Z_{|\Omega|^2}&=\int_{}^{}|\mathcal{A}_{\Omega}(\varphi)|^2 \dd\mu_{2\mn{D}}^{\Sigma}(\varphi)\detz (2\mn{D}_{\Sigma})^{-\frac{1}{2}} \\
  &=\int_{\mathcal{D}'(M)}^{}\me^{-S_{|\Omega|^2}(\phi)}\dd\mu_{\mm{GFF}}^{|\Omega|^2}(\phi)\detz(\Delta_{|\Omega|^2}+m^2)^{-\frac{1}{2}} \\
  &=\int_{\mathcal{D}'(\Sigma)}^{}\dd\tau_{\Sigma}\big(\me^{-S_{|\Omega|^2}}\cdot\mu_{\mm{GFF}}^{|\Omega|^2}\big)(\varphi) \detz(\Delta_{|\Omega|^2}+m^2)^{-\frac{1}{2}}.
\end{align*}
In the case~$\Omega\in \Mor(\Sigma_1,\Sigma_2)$ this corresponds to the fact that~$Z_{|\Omega|^2}=\ttr_{\rho}(U_{\Omega^*}U_{\Omega})$. This motivates the following definition.

\begin{deef}
  Let~$\Omega$ be an unoriented morphism among~$\Sigma$ and $P$ the polynomial defining the interaction. We define the \textsf{amplitude} associated to~$\Omega$ to be the quantity
\begin{equation}
  \mathcal{A}^P_{\Omega}(\varphi)\defeq \bigg[\frac{\detz(\Delta_{|\Omega|^2}+m^2)^{-\frac{1}{2}}}{\detz (2\mn{D}_{\Sigma})^{-\frac{1}{2}}}\cdot \frac{\dd\big[\tau_{\Sigma}\big(\me^{-S_{|\Omega|^2}}\cdot\mu_{\mm{GFF}}^{|\Omega|^2}\big)\big]}{\dd\mu_{2\mn{D}}^{\Sigma}}(\varphi) \bigg] ^{\frac{1}{2}},
  \label{eqn-def-amp-omega}
\end{equation}
where the second ratio denotes the Radon-Nikodym density.
\end{deef}

This is well-defined since~$\tau_{\Sigma}(\me^{-S_{|\Omega|^2}}\cdot\mu_{\mm{GFF}}^{|\Omega|^2}) \ll \mu_{2\mn{D}}^{\Sigma}$ as~$\tau_{\Sigma}(\mu_{\mm{GFF}}^{|\Omega|^2}) =\mu_{2\DN}^{\Sigma,\Omega} \ll \mu_{2\mn{D}}^{\Sigma}$ by corollary \ref{corr-rad-niko-dense}, and both are finite positive measures. We see also that~$\mathcal{A}_{\Omega}>0$ almost surely with respect to~$\wh{\mu}_{2\mn{D}}^{\Sigma}$. We have automatically~$\mathcal{A}_{\Omega}\in L^2(\mathcal{D}'(\Sigma),\wh{\mu}_{2\mn{D}}^{\Sigma})=\mathcal{H}_{\Sigma}$ since~$|\mathcal{A}_{\Omega}|^2\in L^1(\mathcal{D}'(\Sigma),\wh{\mu}_{2\mn{D}}^{\Sigma})$ by definition.

\begin{def7}
This definition (\ref{eqn-def-amp-omega}) is inspired by \cite{Pickrell} (definition 3). In fact it hints at an infinite-dimensional analogue of the notion of ``half-densities'', and is actually what \cite{Pickrell} considered. This latter notion has the advantage of being intrinsic. But due to insufficient literature discussing half-densities in the infinite-dimensional setting\footnote{D. Pickrell confirmed with the author that he did not know of a reference which discuss infinite-dimensional half-densities other than the appendix of \cite{Pickrell}. The author agrees with him that this notion might be a helpful one worthy of a detailed development.}, we adopt here the equivalent approach of fixing a background measure. The reader could also see below that the expression we obtain from this definition agrees, for example, with the one of \cite{GKRV} in its form.
\end{def7}

\begin{exxx}
  Let us derive the amplitude~$\mathcal{A}_{\Omega}^0$ for the free field, that is, with $S_{|\Omega|^2}=0$. Indeed,
  \begin{equation}
    |\mathcal{A}^0_{\Omega}(\varphi)|^2=\frac{\dd\tau_{\Sigma}\big(\mu_{\mm{GFF}}^{|\Omega|^2}\detz(\Delta_{|\Omega|^2}+m^2)^{-\frac{1}{2}}\big)}{\dd\mu_{2\mn{D}}^{\Sigma}\detz (2\mn{D})^{-\frac{1}{2}}}(\varphi).
    \label{}
  \end{equation}
  Since
  \begin{equation}
    \tau_{\Sigma}(\mu_{\mm{GFF}}^{|\Omega|^2})=\mu_{2\DN}^{\Sigma,\Omega},
    \label{}
  \end{equation}
  namely the Gaussian measure on~$\mathcal{D}'(\Sigma)$ with covariance~$\frac{1}{2}(\DN_{\Omega}^{\Sigma})^{-1}$, taking into account the BFK formula
  \begin{equation}
    \detz(\Delta_{|\Omega|^2}+m^2)=[\detz(\Delta_{\Omega,D}+m^2)]^2\detz(2\DN_{\Omega}^{\Sigma}),
    \label{}
  \end{equation}
  we obtain, by corollary \ref{corr-rad-niko-dense},
  \begin{equation} \mathcal{A}_{\Omega}^0(\varphi)=\detz(\Delta_{\Omega,D}+m^2)^{-\frac{1}{2}}\me^{-\frac{1}{2}\sank{\varphi,(\DN_{\Omega}^{\Sigma}-\mn{D})\varphi}_{L^2(\Sigma)}}.
    \label{}
  \end{equation}
\end{exxx}

\begin{lemm}\label{lemm-amplitude-express}
  We have
  \begin{equation}
    \mathcal{A}_{\Omega}^P(\varphi)=\mathcal{A}_{\Omega}^0(\varphi)\int_{}^{}\me^{-S_{\Omega}(\phi^{D}_{\Omega}|\varphi)}\dd\mu_{\mm{GFF}}^{\Omega,D}(\phi^{D}_{\Omega}),
    \label{}
  \end{equation}
  for almost every~$\varphi$ under~$\mu_{2\mn{D}}^{\Sigma}$. Here~$S_{\Omega}(\phi^{D}_{\Omega}|\varphi)$ is as defined in definition \ref{def-inter-over-domain}.
\end{lemm}

\begin{proof}
  Indeed, by definition of the measure image,
  \begin{equation}
    \mathcal{A}_{\Omega}^P(\varphi)^2=\mathcal{A}_{\Omega}^0(\varphi)^2\iint_{}^{}\me^{-S_{|\Omega|^2}(\phi^{D}_{\Omega}+\phi^D_{\Omega^*}+\PI_{|\Omega|^2}^{\Sigma}\varphi)}\dd\mu_{\mm{GFF}}^{\Omega,D}(\phi^{D}_{\Omega})\otimes \dd\mu_{\mm{GFF}}^{\Omega^*,D}(\phi^{D}_{\Omega^*}).
    \label{}
  \end{equation}
  However, by proposition \ref{prop-decoupling} which decouples the interaction into a sum over complementary regions,
  \begin{equation}
    \me^{-S_{|\Omega|^2}(\phi)}=\me^{-S_{\Omega}(\phi^{D}_{\Omega}|\varphi)}\me^{-S_{\Omega^*}(\phi^D_{\Omega^*}|\varphi)},
    \label{}
  \end{equation}
  almost surely againt~$\mu_{\mm{GFF}}^{\Omega,D}\otimes \mu_{\mm{GFF}}^{\Omega^*,D}\otimes \tau_{\Sigma}(\mu_{\mm{GFF}}^{|\Omega|^2})$, and hence also against~$\mu_{\mm{GFF}}^{\Omega,D}\otimes \mu_{\mm{GFF}}^{\Omega^*,D}\otimes \mu_{2\mn{D}}^{\Sigma}$. Thus by reflection symmetry between $\Omega$ and $\Omega^*$ and independence between~$\phi^{D}_{\Omega}$ and~$\phi^D_{\Omega^*}$ we obtain the result.
\end{proof}

\subsection{Trace Axiom and its Consequences}\label{sec-trace-axiom}

\noindent In this subsection we treat separately (ii) of lemma \ref{lemm-segal-transfer} as part of (iii) of definition \ref{def-segal-2}. Let~$\Omega\in\Mor(\Sigma,\Sigma)$, \textit{not necessarily connected}. We will consider two closed surfaces:~$\check{\Omega}$ and~$|\Omega|^2:=(\Omega^*\cup_{\Sigma\sqcup\Sigma} \Omega)^{\vee}$.
\begin{prop}
  [pre-trace]\label{prop-pre-trace}
  We have
  \begin{equation}
    \frac{\dd\tau_{\Sigma}(\mu_{\mm{GFF}}^{\check{\Omega}})}{\dd\mu_{2\mn{D}}^{\Sigma}}(\varphi)=\frac{\detz(\DN_{\check{\Omega}}^{\Sigma})^{\frac{1}{2}}}{\detz(2\DN_{\Omega}^{\Sigma\sqcup\Sigma})^{\frac{1}{4}}}\bigg( 
    \frac{\dd\tau_{\Sigma\sqcup\Sigma}(\mu_{\mm{GFF}}^{|\Omega|^2})}{\dd\mu_{2\mn{D}}^{\Sigma\sqcup\Sigma}}(\varphi,\varphi)
    \bigg)^{\frac{1}{2}}.
    \label{eqn-pre-trace}
  \end{equation}
\end{prop}

\begin{proof}
  This boils down to comparing the explicit expressions for the densities as given by corollary \ref{corr-rad-niko-dense}. Indeed, the relation that we need is
  \begin{equation}
    \Bank{\bnom{\varphi}{\varphi},(2\DN_{\Omega}^{\Sigma\sqcup\Sigma} -2\mn{D}_{\Sigma\sqcup\Sigma})\bnom{\varphi}{\varphi}}_{L^2(\Sigma\sqcup\Sigma)}=
    2\bank{\varphi,(\DN_{\check{\Omega}}^{\Sigma}-2\mn{D}_{\Sigma})\varphi}_{L^2(\Sigma)}.
    \label{}
  \end{equation}
  This is true for~$\varphi\in W^{\frac{1}{2}}(\Sigma)$ because~$\mn{D}_{\Sigma\sqcup\Sigma}=\mn{D}_{\Sigma}\oplus \mn{D}_{\Sigma}$ and~$\sank{[\smx{\varphi\\ \varphi}], \DN_{\Omega}^{\Sigma\sqcup\Sigma}[\smx{\varphi\\ \varphi}]}_{L^2(\Sigma\sqcup\Sigma)}$ and~$\sank{\varphi,\DN_{\check{\Omega}}^{\Sigma}\varphi}_{L^2(\Sigma)}$ are both the Dirichlet energy of the harmonic extension over~$\Omega$ with boundary condition~$(\varphi,\varphi)$. The equality then extends to~$\varphi\in W^{-\delta}(\Sigma)$ for small enough~$\delta$ by continuity (see lemma \ref{lemm-supp-gaus-meas}).
\end{proof}

\begin{corr}
  [trace for free field] \label{cor-free-trace} In the situation as above, we have
  \begin{equation}
    \int_{}^{} \mathcal{A}_{\Omega}^0(\varphi,\varphi)\dd\mu_{2\mn{D}}^{\Sigma}(\varphi) \detz(2\mn{D}_{\Sigma})^{-\frac{1}{2}}=\detz(\Delta_{\check{\Omega}}+m^2)^{-\frac{1}{2}}.
    \label{}
  \end{equation}
\end{corr}

\begin{proof}
  Indeed,
  \begin{align*}
    \textrm{LHS}&=\frac{\detz(\Delta_{|\Omega|^2}+m^2)^{-\frac{1}{4}}}{\detz(2\mn{D}_{\Sigma}\oplus 2\mn{D}_{\Sigma})^{-\frac{1}{4}}}
    \int_{}^{}\bigg( 
    \frac{\dd\tau_{\Sigma\sqcup\Sigma}(\mu_{\mm{GFF}}^{|\Omega|^2})}{\dd\mu_{2\mn{D}}^{\Sigma\sqcup\Sigma}}(\varphi,\varphi)
    \bigg)^{\frac{1}{2}} \dd\mu_{2\mn{D}}^{\Sigma}(\varphi) \detz(2\mn{D}_{\Sigma})^{-\frac{1}{2}} \\
    &=\frac{\detz(2\DN_{\Omega}^{\Sigma\sqcup\Sigma})^{\frac{1}{4}}}{\detz(\DN_{\check{\Omega}}^{\Sigma})^{\frac{1}{2}}}
    \frac{\detz(\Delta_{|\Omega|^2}+m^2)^{-\frac{1}{4}}}{\detz(2\mn{D}_{\Sigma})^{-\frac{1}{2}}}
    \underbrace{\int_{}^{}\frac{\dd\tau_{\Sigma}(\mu_{\mm{GFF}}^{\check{\Omega}})}{\dd\mu_{2\mn{D}}^{\Sigma}}(\varphi) \dd\mu_{2\mn{D}}^{\Sigma}(\varphi)}_{=~ 1} \detz(2\mn{D}_{\Sigma})^{-\frac{1}{2}} \\
    &=\detz(\DN_{\check{\Omega}}^{\Sigma})^{-\frac{1}{2}}\detz(\Delta_{\Omega,D}+m^2)^{-\frac{1}{2}} \tag{BFK for~$|\Omega|^2$}\\
    &=\textrm{RHS}, \tag{BFK for~$\check{\Omega}$}
  \end{align*}
  finishing the proof.
\end{proof}

Note that we do not insist that~$\Omega$ be connected. This leads to the following important consequence of proposition \ref{prop-pre-trace}.

\begin{corr}\label{cor-disec-dens}
  Let~$\Omega_1$,~$\Omega_2$ be two dualizable surfaces such that~$\partial\Omega_1=\partial\Omega_2=\Sigma$, all with the same co-orientation. Denote~$M:=\Omega_1\cup_{\Sigma}\Omega_2$, namely~$\Omega_1$ and~$\Omega_2$ glued along~$\Sigma$. Then
  \begin{equation}
    \frac{\dd\tau_{\Sigma}(\mu_{\mm{GFF}}^{M})}{\dd\mu_{2\mn{D}}^{\Sigma}}(\varphi)=
    \frac{\detz(\DN_{M}^{\Sigma})^{\frac{1}{2}}}{\detz(2\DN_{\Omega_1}^{\Sigma})^{\frac{1}{4}}\detz(2\DN_{\Omega_2}^{\Sigma})^{\frac{1}{4}}}\bigg( 
    \frac{\dd\tau_{\Sigma}(\mu_{\mm{GFF}}^{|\Omega_1|^2})}{\dd\mu_{2\mn{D}}^{\Sigma}}(\varphi)
    \bigg)^{\frac{1}{2}}
    \bigg( 
    \frac{\dd\tau_{\Sigma}(\mu_{\mm{GFF}}^{|\Omega_2|^2})}{\dd\mu_{2\mn{D}}^{\Sigma}}(\varphi)
    \bigg)^{\frac{1}{2}}.
    \label{eqn-disect-density}
  \end{equation}
\end{corr}

\begin{proof}
  Indeed, in this case~$\Omega:=\Omega_1\sqcup\Omega_2$ can be seen as an element of~$\Mor(\Sigma,\Sigma)$. Then~$|\Omega|^2=|\Omega_1|^2\sqcup |\Omega_2|^2$ and the GFFs on the two components are independent. In addition,~$\DN_{\Omega}^{\Sigma\sqcup\Sigma}$ is the direct sum~$\DN_{\Omega_1}^{\Sigma}\oplus \DN_{\Omega_2}^{\Sigma}$. Thus (\ref{eqn-pre-trace}) gives (\ref{eqn-disect-density}).
\end{proof}

\begin{corr}
  [dissection gluing] \label{cor-disect-free} In the same situation as corollary \ref{cor-disec-dens}, we have
  \begin{equation}
    \int_{}^{}\mathcal{A}_{\Omega_1}^0(\varphi)\mathcal{A}_{\Omega_2}^0(\varphi)\dd\mu_{2\mn{D}}^{\Sigma}(\varphi) \detz(2\mn{D}_{\Sigma})^{-\frac{1}{2}}=\detz(\Delta_M+m^2)^{-\frac{1}{2}}.
    \label{}
  \end{equation}
\end{corr}

\begin{proof}
  Again, we apply corollary \ref{cor-free-trace} directly to the case~$\Omega=\Omega_1\sqcup\Omega_2$ and~$|\Omega|^2=|\Omega_1|^2\sqcup |\Omega_2|^2$. We note that~$\mathcal{A}_{\Omega_1}^0(\varphi)\mathcal{A}_{\Omega_2}^0(\varphi)= \mathcal{A}_{\Omega}^0(\varphi,\varphi)$ because
  \begin{equation}
    \frac{\dd\tau_{\Sigma\sqcup\Sigma}(\mu_{\mm{GFF}}^{|\Omega|^2})}{\dd\mu_{2\mn{D}}^{\Sigma\sqcup\Sigma}}(\varphi,\varphi)=\frac{\dd\tau_{\Sigma}(\mu_{\mm{GFF}}^{|\Omega_1|^2})}{\dd\mu_{2\mn{D}}^{\Sigma}}(\varphi) 
    \frac{\dd\tau_{\Sigma}(\mu_{\mm{GFF}}^{|\Omega_2|^2})}{\dd\mu_{2\mn{D}}^{\Sigma}}(\varphi)
    \label{}
  \end{equation}
  and~$\detz(\Delta_{|\Omega|^2}+m^2)=\detz(\Delta_{|\Omega_1|^2}+m^2)\detz(\Delta_{|\Omega_2|^2}+m^2)$.
\end{proof}

Next we deal with the trace axiom in the interacting case. Again let~$\Omega\in\Mor(\Sigma,\Sigma)$ and consider~$\check{\Omega}$. One has in this case the decomposition
\begin{equation}
  \mu_{\mm{GFF}}^{\check{\Omega}}=\mu_{\mm{GFF}}^{\Omega,D}\otimes \tau_{\Sigma}(\mu_{\mm{GFF}}^{\check{\Omega}})
  \label{}
\end{equation}
and against which
\begin{equation}
  S_{\check{\Omega}}(\phi)=S_{\check{\Omega}}(\phi^{D}_{\Omega}+\PI_{\check{\Omega}}^{\Sigma}\varphi)=S_{\Omega}(\phi^{D}_{\Omega}|\varphi,\varphi)
  \quad\textrm{and}\quad
  \me^{-S_{\check{\Omega}}(\phi^{D}_{\Omega}+\PI_{\check{\Omega}}^{\Sigma}\varphi)}=\me^{-S_{\Omega}(\phi^{D}_{\Omega}|\varphi,\varphi)}
  \label{}
\end{equation}
for~$\phi^{D}_{\Omega}\sim \mu_{\mm{GFF}}^{\Omega,D}$ and~$\varphi\sim \tau_{\Sigma}(\mu_{\mm{GFF}}^{\check{\Omega}})$ almost surely (see also remark \ref{rem-induce-law-equality} and \ref{rem-interp-inter-over-omega}).

\begin{corr}
  [trace for~$P(\phi)$ field] \label{cor-P(phi)-trace} In the situation as above, we have
  \begin{equation}
    \int_{}^{} \mathcal{A}_{\Omega}^P(\varphi,\varphi)\dd\mu_{2\mn{D}}^{\Sigma}(\varphi) \detz(2\mn{D}_{\Sigma})^{-\frac{1}{2}}=Z_{\check{\Omega}}.
    \label{}
  \end{equation}
\end{corr}

\begin{proof}
  Indeed,
  \begin{equation}
    \mathcal{A}_{\Omega}^P(\varphi,\varphi)=\mathcal{A}_{\Omega}^0(\varphi,\varphi)\int_{}^{}\me^{-S_{\Omega}(\phi^{D}_{\Omega}|\varphi,\varphi)}\dd\mu_{\mm{GFF}}^{\Omega,D}(\phi^{D}_{\Omega}),
    \label{}
  \end{equation}
  therefore, with the constants involving determinants working out in exactly the same way as corollary \ref{cor-free-trace}, one has
  \begin{align*}
    \textrm{LHS}&=\detz(\Delta_{\check{\Omega}}+m^2)^{-\frac{1}{2}} \int_{}^{}\frac{\dd\tau_{\Sigma}(\mu_{\mm{GFF}}^{\check{\Omega}})}{\dd\mu_{2\mn{D}}^{\Sigma}}(\varphi) \dd\mu_{2\mn{D}}^{\Sigma}(\varphi)\int_{}^{}\me^{-S_{\Omega}(\phi^{D}_{\Omega}|\varphi,\varphi)}\dd\mu_{\mm{GFF}}^{\Omega,D}(\phi^{D}_{\Omega}) \\
    &=\detz(\Delta_{\check{\Omega}}+m^2)^{-\frac{1}{2}}
    \int_{}^{}\me^{-S_{\Omega}(\phi^{D}_{\Omega}|\varphi,\varphi)} \dd\mu_{\mm{GFF}}^{\Omega,D}\otimes \dd\tau_{\Sigma}(\mu_{\mm{GFF}}^{\check{\Omega}})(\phi^{D}_{\Omega},\varphi)\\
    &=\detz(\Delta_{\check{\Omega}}+m^2)^{-\frac{1}{2}} \mb{E}_{\mm{GFF}}^{\check{\Omega}}[\me^{-S_{\check{\Omega}}(\phi)}]=\textrm{RHS}.
  \end{align*}
  We arrive at the proof.
\end{proof}

\begin{corr}
  [dissection gluing for $P(\phi)$ field] In the same situation as corollary \ref{cor-disec-dens}, we have
  \begin{equation}
    \int_{}^{}\mathcal{A}_{\Omega_1}^P(\varphi)\mathcal{A}_{\Omega_2}^P(\varphi)\dd\mu_{2\mn{D}}^{\Sigma}(\varphi) \detz(2\mn{D}_{\Sigma})^{-\frac{1}{2}}=Z_M.
    \label{}
  \end{equation}
\end{corr}

\begin{proof}
  One verifies that for~$\Omega=\Omega_1\sqcup \Omega_2$,
  \begin{equation}
    \mathcal{A}_{\Omega}^P(\varphi,\varphi)=\mathcal{A}_{\Omega_1}^P(\varphi)\mathcal{A}_{\Omega_2}^P(\varphi).
    \label{}
  \end{equation}
  This in fact also comes directly from the definition as
  \begin{equation}
    S_{|\Omega|^2}(\phi_{\mm{GFF}}^{|\Omega|^2})=S_{|\Omega_1|^2}(\phi_{\mm{GFF}}^{|\Omega_1|^2})+S_{|\Omega_2|^2}(\phi_{\mm{GFF}}^{|\Omega_2|^2})
    \label{}
  \end{equation}
  pointwise almost surely with~$\phi_{\mm{GFF}}^{|\Omega|^2}=\phi_{\mm{GFF}}^{|\Omega_1|^2}+\phi_{\mm{GFF}}^{|\Omega_2|^2}$. One then follows a verbatim reasoning as for corollary \ref{cor-disect-free}.
\end{proof}

\subsection{Bayes-Type Density Formulae and Gluing with Free Ends}\label{sec-gluing-final}

\noindent In fact let us first put ourselves in the general situation of (iii) of definition \ref{def-segal-2}. Suppose the Riemannian surface~$\Omega$ is such that~$\partial\Omega=\Sigma_2\sqcup \Sigma_2^*\sqcup\Sigma_1$ and that~$\rho:\Sigma_2\lto \Sigma_2^*$ is an orientation reversing isometry. Denote by~$\Omega/\rho$ the surface obtained from~$\Omega$ by gluing~$\Sigma_2$ with~$\Sigma_2^*$ along~$\rho$. Denote by~$\Sigma_4$ the reflected copy of~$\Sigma_2$ in~$|\Omega/\rho|^2$.

\begin{notation}
  In general, for~$M$ a closed manifold and~$\Sigma_1$, \dots,~$\Sigma_{\ell}\subset M$ finitely many nonintersecting closed embedded hypersurfaces, we denote
  \begin{align*}
    \tau_{i\sqcup j\sqcup\cdots}&\defeq \textrm{the trace map }C^{\infty}(M)\lto C^{\infty}(\Sigma_i)\times C^{\infty}(\Sigma_j)\times\cdots,\textrm{ and its extensions,}\\
    \PI_M^{i\sqcup j\sqcup\cdots}&\defeq \textrm{the Poisson integral operator }W^{\frac{1}{2}}(\Sigma_i)\times W^{\frac{1}{2}}(\Sigma_j)\times\cdots \lto W^1(M), \\
    \mathcal{M}_{M,k\sqcup\cdots}^{i\sqcup j\sqcup\cdots}&\defeq \tau_{k\sqcup\cdots}\PI_M^{i\sqcup j\sqcup\cdots},\textrm{ the transition operator }\mathcal{D}'(\Sigma_i)\times \mathcal{D}'(\Sigma_j)\times\cdots \lto \mathcal{D}'(\Sigma_k)\times\cdots\textrm{ as in (\ref{eqn-trans-op-def}).}
  \end{align*}
  We also remind the reader of the notations set up in definitions \ref{def-DN-boundary} and \ref{def-pi-more-general}.
\end{notation}

\begin{prop}
  [Bayes with free ends] \label{prop-bayes-free-end} We have
  \begin{equation}
    \frac{\dd\tau_{1}(\mu_{\mm{GFF}}^{|\Omega/\rho|^2})}{\dd \mu_{2\mn{D}}^{\Sigma_1}}(\varphi_1)
    \bigg( 
    \frac{\dd \big((\mathcal{M}_{|\Omega/\rho|^2,2}^{1}\varphi_1)_*\mu_{\DN}^{\Sigma_2,\Omega/\rho,D}\big)}{\dd\mu_{2\mn{D}}^{ \Sigma_2}}(x)
    \bigg)^2 =\frac{\detz(\DN_{|\Omega/\rho|^2}^{\Sigma_2\sqcup\Sigma_4})^{\frac{1}{2}}}{\detz(\DN_{|\Omega|^2}^{\Sigma_2\sqcup\Sigma_2^*})^{\frac{1}{2}}}
    \frac{\dd\tau_{2\sqcup 2^*\sqcup 1}(\mu_{\mm{GFF}}^{|\Omega|^2})}{\dd \mu_{2\mn{D}}^{2\sqcup 2^*\sqcup 1}}(x,x,\varphi_1)
  \end{equation}
\end{prop}

\begin{figure}
    \centering
    \includegraphics[width=0.8\linewidth]{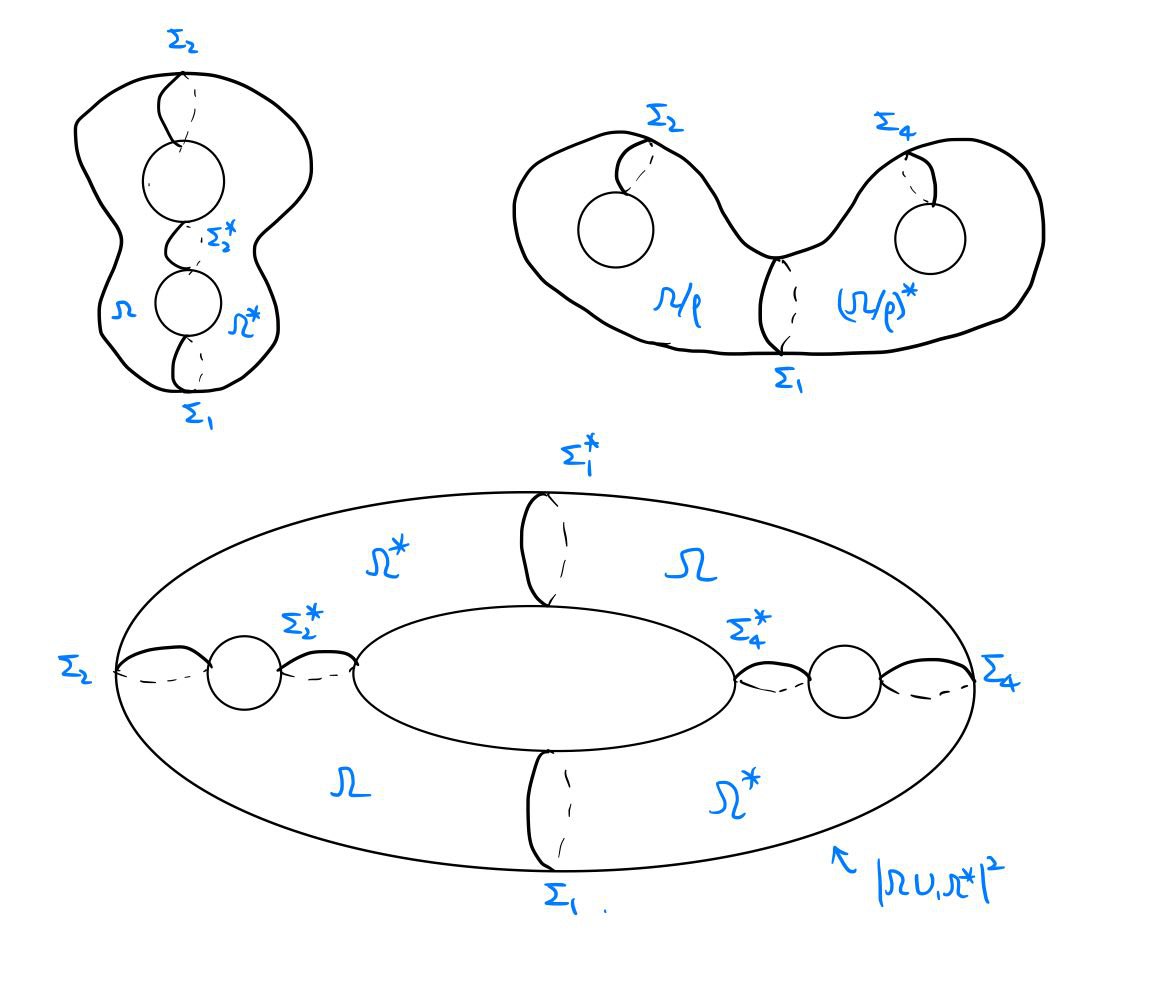}
    \caption{two ways of taking trace}
    \label{fig-general-trace}
\end{figure}

\begin{proof}
  Indeed, by reflection symmetry across~$\Sigma_1$ where~$\Sigma_2$ gets identified with~$\Sigma_4$, we have
  \begin{equation}
  \mathcal{M}_{|\Omega/\rho|^2,2}^{1}\varphi_1 =
\mathcal{M}_{|\Omega/\rho|^2,4}^{1}\varphi_1,\quad\textrm{and}\quad
\mu_{\DN}^{\Sigma_2,\Omega/\rho,D}=\mu_{\DN}^{\Sigma_4,\Omega/\rho,D},
    \label{}
  \end{equation}
  under the symmetry. Also by symmetry (method of images)   
 , we have
  \begin{equation}
  \mathcal{M}_{|\Omega/\rho|^2,1}^{2\sqcup 4}\bnom{x}{ x}
    =\mathcal{M}_{|\Omega|^2,1}^{2\sqcup 2^*}\bnom{x}{ x},
    \label{}
  \end{equation}
  because both of them expresses the associated Dirichlet data on~$\Sigma_1$ of the solution of the Helmholtz boundary value problem over $\Omega$ with Dirichlet data equal to~$x$ on both~$\Sigma_2$,~$\Sigma_2^*$ and Neumann condition (zero normal derivative) on~$\Sigma_1$. Thus,
  \begin{align*}
    \textrm{LHS}&=\frac{\dd\tau_{2\sqcup 4\sqcup 1}(\mu_{\mm{GFF}}^{|\Omega/\rho|^2})}{\dd \mu_{2\mn{D}}^{2\sqcup 4\sqcup 1}}(x,x,\varphi_1) \tag{formula (\ref{eqn-bayes-gff-1}) backward}\\
    &=\frac{\dd\tau_{2\sqcup 4}(\mu_{\mm{GFF}}^{|\Omega/\rho|^2})}{\dd \mu_{2\mn{D}}^{\Sigma_2\sqcup \Sigma_4}}(x,x)
    \frac{\dd \big((\mathcal{M}_{|\Omega/\rho|^2,1}^{2\sqcup 4}[\smx{x\\ x}])_*\mu_{2\DN}^{\Sigma_1,\Omega,D}\big)}{\dd\mu_{2\mn{D}}^{ \Sigma_1}}(\varphi_1) \tag{formula (\ref{eqn-bayes-gff-2}) forward}\\
    &=\frac{\detz(\DN_{|\Omega/\rho|^2}^{\Sigma_2\sqcup\Sigma_4})^{\frac{1}{2}}}{\detz(\DN_{|\Omega|^2}^{\Sigma_2\sqcup\Sigma_2^*})^{\frac{1}{2}}}
   \frac{\dd\tau_{2\sqcup 2^*}(\mu_{\mm{GFF}}^{|\Omega|^2})}{\dd \mu_{2\mn{D}}^{\Sigma_2\sqcup \Sigma_2^*}}(x,x)
   \frac{\dd \big((\mathcal{M}_{|\Omega|^2,1}^{2\sqcup 2^*}[\smx{x\\ x}])_*\mu_{2\DN}^{\Sigma_1,\Omega,D}\big)}{\dd\mu_{2\mn{D}}^{ \Sigma_1}}(\varphi_1) \tag{\#}\\
  &=\textrm{RHS}.
  \end{align*}
We see that with the help of Bayes principle we ``glued away'' the free end and we are reduced to the situation of proposition \ref{prop-pre-trace}. Indeed,  at step (\#) we use proposition \ref{prop-pre-trace} twice with respect to~$|\Omega\cup_1\Omega^*|^2$ which is~$\Omega$ ``reflected twice'' by first gluing~$\Omega$ with a reflected copy along~$\Sigma_1$ to get~$\Omega\cup_1\Omega^*$, and then reflect and glue again along~$\Sigma_2\sqcup\Sigma_2^*\sqcup\Sigma_4\sqcup\Sigma_4^*$, with~$\Sigma_4$,~$\Sigma_4^*$ being the reflected copies of~$\Sigma_2$,~$\Sigma_2^*$ across~$\Sigma_1$. From this we get
\begin{align*}
  \frac{\dd\tau_{2\sqcup 4}(\mu_{\mm{GFF}}^{|\Omega/\rho|^2})}{\dd \mu_{2\mn{D}}^{\Sigma_2\sqcup \Sigma_4}}(x,x)&=
  \frac{\detz(\DN_{|\Omega/\rho|^2}^{\Sigma_2\sqcup\Sigma_4})^{\frac{1}{2}}}{\detz(2\DN_{|\Omega\cup_1 \Omega^*|^2}^{2\sqcup 2^*\sqcup 4\sqcup 4^*})^{\frac{1}{4}}}
  \bigg( 
    \frac{\dd\tau_{2\sqcup 2^*\sqcup 4\sqcup 4^*}(\mu_{\mm{GFF}}^{|\Omega\cup_1 \Omega^*|^2})}{\dd\mu_{2\mn{D}}^{2\sqcup 2^*\sqcup 4\sqcup 4^*}}(x,x,x,x)
    \bigg)^{\frac{1}{2}}\\ 
    &=\frac{\detz(\DN_{|\Omega/\rho|^2}^{\Sigma_2\sqcup\Sigma_4})^{\frac{1}{2}}}{\detz(\DN_{|\Omega|^2}^{\Sigma_2\sqcup\Sigma_2^*})^{\frac{1}{2}}}
   \frac{\dd\tau_{2\sqcup 2^*}(\mu_{\mm{GFF}}^{|\Omega|^2})}{\dd \mu_{2\mn{D}}^{\Sigma_2\sqcup \Sigma_2^*}}(x,x),
\end{align*}
  finishing the proof.
\end{proof}

\begin{corr}
  [free-end gluing for free field] \label{cor-real-glue-free} We have
  \begin{equation}
    \int_{}^{}\mathcal{A}_{\Omega}^0(\varphi,\varphi,\psi)\dd\mu_{2\mn{D}}^{\Sigma_2}(\varphi)\detz(2\mn{D}_{\Sigma_2})^{-\frac{1}{2}}=\mathcal{A}_{\Omega/\rho}^0(\psi),
    \label{}
  \end{equation}
  for almost every $\psi\sim \mu_{2\mn{D}}^{\Sigma_1}$.
\end{corr}

\begin{proof}
  The key point is to disintegrate $\mathcal{A}_{\Omega}^0(\varphi,\varphi,\psi)$ into a part involving only $\varphi$, which one could ``integrate out'' cleanly, multiplied by a part independent of $\varphi$, using proposition \ref{prop-bayes-free-end}. Indeed,
  \begin{align*}
    \textrm{LHS}&=\frac{\detz(\Delta_{|\Omega|^2}+m^2)^{-\frac{1}{4}}}{\detz(2\mn{D}_{\Sigma_2})^{-\frac{1}{2}}\detz( 2\mn{D}_{\Sigma_1})^{-\frac{1}{4}}}
    \int_{}^{}\bigg( 
    \frac{\dd\tau_{2\sqcup 2^*\sqcup 1}(\mu_{\mm{GFF}}^{|\Omega|^2})}{\dd \mu_{2\mn{D}}^{2\sqcup 2^*\sqcup 1}}(\varphi,\varphi,\psi)
    \bigg)^{\frac{1}{2}} \dd\mu_{2\mn{D}}^{\Sigma_2}(\varphi) \detz(2\mn{D}_{\Sigma_2})^{-\frac{1}{2}} \\
    &=\frac{\detz(\DN_{|\Omega|^2}^{\Sigma_2\sqcup\Sigma_2^*})^{\frac{1}{4}}}{\detz(\DN_{|\Omega/\rho|^2}^{\Sigma_2\sqcup\Sigma_4})^{\frac{1}{4}}}
    \frac{\detz(\Delta_{|\Omega|^2}+m^2)^{-\frac{1}{4}}}{\detz(2\mn{D}_{\Sigma_1})^{-\frac{1}{4}}}
    \bigg( \frac{\dd\tau_{1}(\mu_{\mm{GFF}}^{|\Omega/\rho|^2})}{\dd \mu_{2\mn{D}}^{\Sigma_1}}(\psi) \bigg)^{\frac{1}{2}}
    \underbrace{\int_{}^{}\frac{\dd \big((\mathcal{M}_{|\Omega/\rho|^2,2}^{1}\psi)_*\mu_{\DN}^{\Sigma_2,\Omega/\rho,D}\big)}{\dd\mu_{2\mn{D}}^{ \Sigma_2}}(\varphi)
    \dd\mu_{2\mn{D}}^{\Sigma_2}(\varphi)}_{=~ 1}  \\
    &=\detz(\DN_{|\Omega/\rho|^2}^{\Sigma_2\sqcup\Sigma_4})^{-\frac{1}{4}} \detz(\Delta_{|\Omega|^2\setminus \Sigma_2\sqcup\Sigma_2^*,D}+m^2)^{-\frac{1}{4}}
    \detz(2\mn{D}_{\Sigma_1})^{\frac{1}{4}} 
    \bigg( \frac{\dd\tau_{1}(\mu_{\mm{GFF}}^{|\Omega/\rho|^2})}{\dd \mu_{2\mn{D}}^{\Sigma_1}}(\psi) \bigg)^{\frac{1}{2}}\tag{BFK for~$|\Omega|^2$ and~$\Sigma_2\sqcup\Sigma_2^*$}\\
    &=\textrm{RHS}, \tag{BFK for~$|\Omega/\rho|^2$ and~$\Sigma_2\sqcup\Sigma_4$}
  \end{align*}
  finishing the proof.
\end{proof}

Before proceeding to the interacting case we remind the reader of remark \ref{rem-interp-inter-over-omega}.

\begin{lemm}\label{lemm-equal-law-2-PI}
  For fixed~$\psi\in W^{\frac{1}{2}}(\Sigma_1)$, then the random fields
  \begin{equation}
    \phi\defeq\PI_{\Omega/\rho}^{1}\psi+\PI_{\Omega/\rho}^{2,D}\tau_{2}\phi_{\Omega/\rho}^D,\quad\textrm{and}\quad \tilde{\phi}\defeq\PI_{\Omega}^{2\sqcup 2^*\sqcup 1}[\varphi,\varphi,\psi]
    \label{}
  \end{equation}
  with~$\phi_{\Omega/\rho}^D\sim \mu_{\mm{GFF}}^{\Omega/\rho,D}$ and~$\varphi\sim (\mathcal{M}_{|\Omega/\rho|^2,2}^{1}\psi)_*\mu_{\DN}^{\Sigma_2,\Omega/\rho,D}$, follow the same law on~$\mathcal{D}'(\Omega^{\circ})$.
\end{lemm}

\begin{proof}
  Here~$\Omega^{\circ}$ is also~$(\Omega/\rho)\setminus\Sigma_2$. Indeed, both~$\phi$ and~$\tilde{\phi}$ solves the stochastic boundary value problem
  \begin{equation}
    \left\{
    \begin{array}{ll}
      (\Delta+m^2)u=0&\textrm{in }\Omega^{\circ},\\
      u|_{\Sigma_2}=u|_{\Sigma_2^*}\sim(\mathcal{M}_{|\Omega/\rho|^2,2}^{1}\psi)_*\mu_{\DN}^{\Sigma_2,\Omega/\rho,D}, & \textrm{on }\Sigma_2\textrm{ or }\Sigma_2\sqcup\Sigma_2^* \\
      u|_{\Sigma_1}=\psi,&\textrm{on }\Sigma_1,
    \end{array}
    \right.
    \label{}
  \end{equation}
  with equalities holding almost surely.
\end{proof}

\begin{corr}
  [free-end gluing for~$P(\phi)$ field]  In the situation as above, we have
  \begin{equation}
    \int_{}^{} \mathcal{A}_{\Omega}^P(\varphi,\varphi,\psi)\dd\mu_{2\mn{D}}^{\Sigma_2}(\varphi) \detz(2\mn{D}_{\Sigma_2})^{-\frac{1}{2}}=\mathcal{A}_{\Omega/\rho}^P(\psi),
    \label{}
  \end{equation}
  for almost every $\psi\sim \mu_{2\mn{D}}^{\Sigma_1}$.
\end{corr}

\begin{proof}
  Indeed,
  \begin{equation}
    \mathcal{A}_{\Omega}^P(\varphi,\varphi,\psi)=\mathcal{A}_{\Omega}^0(\varphi,\varphi,\psi)\int_{}^{}\me^{-S_{\Omega}(\phi^{D}_{\Omega}|\varphi,\varphi,\psi)}\dd\mu_{\mm{GFF}}^{\Omega,D}(\phi^{D}_{\Omega}),
    \label{}
  \end{equation}
  therefore, with the constants involving determinants working out in exactly the same way as corollary \ref{cor-real-glue-free}, one has
  \begin{align*}
    \textrm{LHS}&=\frac{\detz(\Delta_{|\Omega/\rho|^2}+m^2)^{-\frac{1}{4}}}{\detz(2\mn{D}_{\Sigma_1})^{\frac{1}{4}}}
    \bigg( \frac{\dd\tau_{1}(\mu_{\mm{GFF}}^{|\Omega/\rho|^2})}{\dd \mu_{2\mn{D}}^{\Sigma_1}}(\psi) \bigg)^{\frac{1}{2}}
    \int_{}^{}\dd \big((\mathcal{M}_{|\Omega/\rho|^2,2}^{1}\psi)_*\mu_{\DN}^{\Sigma_2,\Omega/\rho,D}\big)(\varphi) 
    \int_{}^{} \me^{-S_{\Omega}(\phi^{D}_{\Omega}|\varphi,\varphi,\psi)}\dd\mu_{\mm{GFF}}^{\Omega,D}(\phi^{D}_{\Omega}) \\
    &=\mathcal{A}_{\Omega/\rho}^0(\psi)
    \int_{}^{}\me^{-S_{\Omega/\rho}(\phi^{D}_{\Omega}+(\phi_{\Omega/\rho}^D)_{\Sigma_2}|\psi)} \dd\mu_{\mm{GFF}}^{\Omega,D}\otimes \dd \big[(\PI_{\Omega/\rho}^{2,D}\tau_{2})_*(\mu_{\mm{GFF}}^{\Omega/\rho,D})\big](\phi^{D}_{\Omega},(\phi_{\Omega/\rho}^D)_{\Sigma_2}) \tag{lemma \ref{lemm-equal-law-2-PI}}\\
    &=\mathcal{A}_{\Omega/\rho}^0(\psi)
    \int_{}^{}\me^{-S_{\Omega/\rho}(\phi_{\Omega/\rho}^D|\psi)} \dd\mu_{\mm{GFF}}^{\Omega/\rho,D}(\phi_{\Omega/\rho}^D)=\textrm{RHS},
  \end{align*}
  where the notation $(\phi_{\Omega/\rho}^D)_{\Sigma_2}$ is as in proposition \ref{prop-stoc-decomp-domain}. We arrive at the proof.
\end{proof}

Now let~$\Omega_1\in\Mor(\Sigma_1,\Sigma_2)$,~$\Omega_2\in\Mor(\Sigma_2,\Sigma_3)$ and~$|\Omega_2\Omega_1|^2:=(\Omega_1^*\cup_{4}\Omega_2^*\cup_{3}\Omega_2\cup_2\Omega_1)^{\vee}$, where we denote by~$\Sigma_4$ the ``glued'' outgoing and incoming boundaries of~$\Omega_2^*$ and~$\Omega_1^*$. In this case the result could be seen as a special case of the previous one, where $\Omega=\Omega_1\sqcup\Omega_2$ has two disjoint components. We shall re-state these results without proofs.

\begin{corr}\label{prop-pre-glu}
  We have
  \begin{align*}
    &\quad\frac{\dd\tau_{1\sqcup 3}(\mu_{\mm{GFF}}^{|\Omega_2\Omega_1|^2})}{\dd \mu_{2\mn{D}}^{\Sigma_1\sqcup \Sigma_3}}(\varphi_1,\varphi_3)
    \left( \frac{\dd \big((\mathcal{M}_{|\Omega_2\Omega_1|^2,2}^{1\sqcup 3}[\smx{\varphi_1\\ \varphi_3}])_*\mu_{2\DN}^{\Sigma_2,\Omega_2\Omega_1,D}\big)}{\dd\mu_{2\mn{D}}^{ \Sigma_2}}(x) \right)^2 \\
    =&\quad \frac{\detz(\DN_{|\Omega_2\Omega_1|^2}^{\Sigma_2\sqcup\Sigma_4})^{\frac{1}{2}}}{\detz(2\DN_{\Omega_1}^{\Sigma_2,N})^{\frac{1}{2}}\detz(2\DN_{\Omega_2}^{\Sigma_2,N})^{\frac{1}{2}}}
    \frac{\dd\tau_{\Sigma_1\sqcup \Sigma_2}(\mu_{\mm{GFF}}^{|\Omega_1|^2}) }{\dd \mu_{2\mn{D}}^{\Sigma_1\sqcup \Sigma_2}}(\varphi_1,x) 
    \frac{\dd\tau_{\Sigma_2\sqcup \Sigma_3}(\mu_{\mm{GFF}}^{|\Omega_2|^2}) }{\dd \mu_{2\mn{D}}^{\Sigma_2\sqcup \Sigma_3}}(x,\varphi_3). \quad\Box
  \end{align*}
\end{corr}

\begin{corr}
  [composition for~$P(\phi)$ field]  In the situation as above, we have
  \begin{equation}
    \int_{}^{} \mathcal{A}_{\Omega_1}^P(\psi_1,\varphi)\mathcal{A}_{\Omega_2}^P(\varphi,\psi_3)\dd\mu_{2\mn{D}}^{\Sigma_2}(\varphi) \detz(2\mn{D}_{\Sigma_2})^{-\frac{1}{2}}=\mathcal{A}_{\Omega_2\cup_2\Omega_1}^P(\psi_1,\psi_3),
    \label{}
  \end{equation}
  for almost every $(\psi_1,\psi_3)\sim \mu_{2\mn{D}}^{\Sigma_1\sqcup\Sigma_3}$. \hfill~$\Box$
\end{corr}

 \newpage
 
\section{Periodic Cover = Spin Chain} \label{sec-period-cover}

\subsection{Geometric setting}
\label{s:abeliancovers}

\noindent Our set-up corresponds to example 1 in Bergeron \cite{Bergeron}. Let~$M$ be a closed oriented Riemannian surface\footnote{In particular its~$\mb{Z}$-homology groups are non-torsion.} of genus~$g\ge 1$ and~$\Sigma\subset M$ an embedded \textit{primitive} closed geodesic whose $\mb{Z}$-homology class is non-trivial (exists by a classical theorem of E. Cartan). Necessarily,~$\Sigma$ is nondissecting.\footnote{If~$\Sigma$ dissects~$M$ into~$M_+^{\circ}\sqcup M_-^{\circ}$, then~$\Sigma=\partial M_+$; now closed 1-forms integrate to zero over $\Sigma$ by Stokes theorem, thus $\Sigma$ is null-homologous via de Rham's theorem. Alternatively, note~$\partial:H_2(M_+,\Sigma;\mb{Z})\lto H_1(\Sigma;\mb{Z})$ in the exact sequence of the pair $(M_+,\Sigma)$ is surjective, producing the fundamental class.} We consider the covering space~$M_{\infty}^{\Sigma}\lto M$ over~$M$ given by a (normal) subgroup~$\ker\rho$ of the fundamental group~$\pi_1(M)$ where~$\rho$ is the map
\begin{equation}
  \begin{tikzcd}
    \rho:\pi_1(M) \ar[r,"\mm{Ab}"] &[+10pt] H_1(M;\mb{Z}) \ar[r, " I(-{,}{[}\Sigma{]})"] &[+10pt] \mb{Z},
  \end{tikzcd}
  \label{}
\end{equation}
where the first is Abelianization and~$I(-,[\Sigma])$ is the \textsf{oriented intersection number}\footnote{$I([\gamma],[\Sigma])=D(D^{-1}[\gamma]\smile D^{-1}[\Sigma])=\int_{\gamma}^{}\eta_{\Sigma}\in\mb{Z}$, where~$D^{-1}$ is the Poincar\'e dual map and~$\eta_{\Sigma}$ is a smooth bump 1-form supported in a tubular neighborhood of~$\Sigma$ such that~$\int_{\Sigma}^{}\alpha=\int_{M}^{}\alpha\wedge\eta_{\Sigma}$ for any 1-form~$\alpha$.} with~$\Sigma$, which is surjective (since~$\Sigma$ is primitive). In other words, we put~$M_{\infty}^{\Sigma}=\ker \rho\setminus \tilde{M}$ where~$\tilde{M}$ is the universal cover of~$M$ and~$\ker\rho$ acts on~$\tilde{M}$ as deck transformations. Equip~$M_{\infty}^{\Sigma}$ with the covering metric (thus deck transformations act by isometries).

Geometrically,~$M_{\infty}^{\Sigma}$ can be understood as first cutting~$M$ along~$\Sigma$ and obtaining the surface~$\Omega:=M\setminus\Sigma$ with boundaries~$\Sigma_{\mm{in}}\sqcup \Sigma_{\mm{out}}$ where~$\Sigma_{\mm{in}} \cong \Sigma_{\mm{out}}\cong \Sigma$, and gluing~$\Omega$ periodically where each~$\Sigma_{\mm{out}}$ is glued to the ``next''~$\Sigma_{\mm{in}}$. Indeed, the class~$[\gamma]\in \pi_1(M)$ of a loop~$\gamma$ is in~$\ker \rho$ iff~$I([\gamma],[\Sigma])=0$; in other words, these loops are exactly those which are ``not cut'', i.e. lifts to a loop on~$M_{\infty}^{\Sigma}$, and loops which do intersect~$\Sigma$ are lifted to segments whose end points are related by a deck transformation, i.e. they are ``cut''.

Now, for~$N\in\mb{N}$, compose~$\rho$ further with the mod~$N$ map~$\mb{Z}\lto \mb{Z}_N=:\mb{Z}/N\mb{Z}$ and denote it by~$\rho_N$, and let the covering space of~$M$ corresponding to~$\ker\rho_N$ be~$M_N^{\Sigma}$. Since~$\ker\rho\subset \ker\rho_N$,~$M_{\infty}^{\Sigma}$ also covers~$M_N^{\Sigma}$. Geometrically, this corresponds to closing the surface after gluing~$N$ copies of~$\Omega$ --- loops that intersect~$\Sigma$~$N$-times are now lifted to a ``big loop'' in~$M_N^{\Sigma}$.

\begin{figure}
    \centering
    \includegraphics[width=0.8\linewidth]{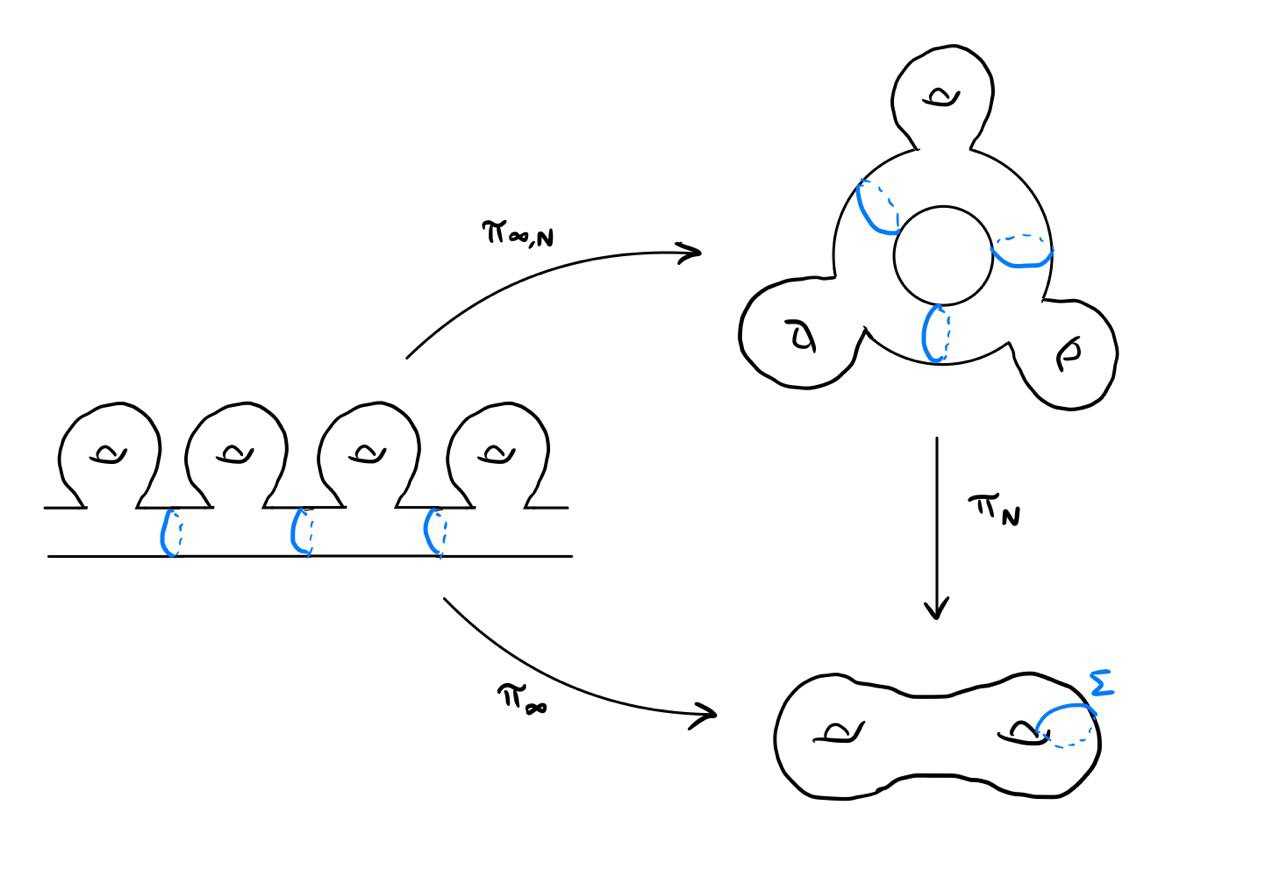}
    \caption{periodic covering and cyclic covering}
    \label{fig-period-cover}
\end{figure}

\begin{def7}
  We also say that the sequence of covers~$(M_N^{\Sigma})_N$ \textsf{converges} to~$M_{\infty}^{\Sigma}$.
\end{def7}

That~$\Omega:=M\setminus\Sigma$ should be understood for what follows.

\subsection{Continued Introduction of Spin Chain Example}\label{sec-spin-chain-cont}

\noindent Here we continue our discussion of the circular spin chain proposed in section \ref{sec-intro-spin-chain}, in particular the equation (\ref{eqn-spin-chain-part-func-trace}).  More generally, one can insert ``nice'' functionals~$F_1$, \dots,~$F_k$ in between at the sites~$1\le i_1<\cdots <i_k\le N$, then
\begin{equation}
  \int_{\mb{R}^N}^{}F_k(\sigma(i_k))\cdots F_1(\sigma(i_1))\me^{-S(\sigma)} \dd^N\sigma =\ttr_{L^2(\mb{R})}\left( T^{N+1-i_k}F_k T^{i_k-i_{k-1}}\cdots F_1 T^{i_1-1} \right),
  \label{}
\end{equation}
where~$F_j$ are thought of as multiplication operators. The evaluation
\begin{equation}
  \left.
  \def\arraystretch{2.5}
  \begin{array}{rcl}
     F_k\otimes\cdots\otimes F_1&\longmapsto& \ddp \frac{1}{\mathcal{Z}(N)}\int_{\mb{R}^N}^{}F_k(\sigma(i_k))\cdots F_1(\sigma(i_1))\me^{-S(\sigma)} \dd^N\sigma \\
     &= & \ddp \frac{\ttr_{L^2(\mb{R})}\left( T^{N+1-i_k}F_k T^{i_k-i_{k-1}}\cdots F_1 T^{i_1-1} \right)}{\ttr_{L^2(\mathbb{R})}( T^{N} )}
  \end{array}
  \right.
  \label{eqn-spin-chain-gibbs-state}
\end{equation}
is said to define a \textsf{Gibbs state} of our spin chain on~$\mb{Z}_N$. Alternatively this is the expected value functional under the discrete Gibbs measure (\ref{eqn-spin-chain-gibbs-part-func}).

\begin{deef}\label{def-thermo-limit}
  We say that a Gibbs state exists in the \textsf{thermodynamic limit} if the second expression in (\ref{eqn-spin-chain-gibbs-state}) tends to a limit as~$N\to \infty$ for bounded functionals~$F_1$, \dots,~$F_k\in L^{\infty}(\mb{R})$.
\end{deef}

We can see from (\ref{eqn-spin-chain-trans-kernel}) that the operator~$T$ is not just bounded, smoothing, but its kernel is strictly positive. Such an operator has a special property which we call the \textsf{Perron-Frobenius} property, referring to the consequence of proposition \ref{prop-perron-frob} below. In particular, this would ensure that the thermodynamic limit does exist, as we shall see in corollaries \ref{cor-segal-trans-exp-decay}, \ref{cor-log-trace-asymp}, and \ref{cor-gibbs-state-limit}.

 As announced in the introduction,
   Corollary \ref{cor-segal-trans-exp-decay} could be equivalently understood as saying that for compactly supported observables $(F,G)$, and $(\tau\sigma)(i):=\sigma(i+1) $ the shift operator, then
\begin{eqnarray*}
  \mathbb{E}[F(\tau^k \sigma) G(\sigma)]= \mathbb{E}[F] \mathbb{E}[G] + \mathcal{O}(\alpha^k)  
\end{eqnarray*}
with $\alpha<1$ the same as in corollary \ref{cor-segal-trans-exp-decay}. Here the expected value should be thought of as coming from a Gibbs measure on the infinite path space~$\mb{R}^{\mb{Z}}$ over~$\mb{Z}$. Indeed, by definition \ref{def-thermo-limit}, this is exactly the vague limit of the finite dimensional (periodic) Gibbs measures over~$\mb{Z}_N$. We say that with respect to this Gibbs measure the shift operator  is \textsf{exponentially mixing} which is exactly saying the $P(\phi)_2$ Gibbs state constructed has a \textbf{mass gap}.

With transfer operator being Perron-Frobenius, the partition functions as in (\ref{eqn-spin-chain-gibbs-part-func}) also enjoy explicit asymptotics. We will explore a consequence in the case of periodic surfaces in the last section.

\begin{exxx}
  In the case $P(\sigma)=m^2\sigma^2$, the spin chain is the discrete massive GFF. We have an exact formula for the partition function $\mathcal{Z}(N)=\prod_{k=0}^{N-1} (1+m^2 -\cos(2\pi\frac{k}{N}  ))^{-\frac{1}{2}}$
hence
\begin{eqnarray*}
  \lim_{N\rightarrow+\infty}    \frac{1}{N} \log \left( \mathcal{Z}(N)\right) = -\frac{1}{2} \int_0^1 \log(1+m^2-\cos(2\pi x))\dd x.
\end{eqnarray*}
\end{exxx}

\subsection{Perron-Frobenius Property and Gibbs State}\label{sec-perron-frob-gibbs}

  A large part of this section follow from the ``properties of a Perron-Frobenius operator''. But we include them here as they form an integral part of the discussion of physical phenomena.

We remind the reader of remark \ref{rem-sym-time-trans}.

\begin{deef}
  An operator~$A$ on~$L^2(Q,\mu)$ of some measure space~$(Q,\mu)$ has \textsf{strictly positive kernel} if for any nonnegative~$F\in L^2(Q,\mu)$ such that~$\nrm{F}_{L^2}\ne 0$ we have~$AF>0$ almost surely.
\end{deef}

\begin{prop}
  [Perron-Frobenius, \cite{GJ} page 51] \label{prop-perron-frob} If~$A$ on~$L^2(Q,\mu)$ has strictly positive kernel, and~$\lambda=\nrm{A}$ is an eigenvalue of~$A$, then~$\lambda$ is simple, and the corresponding eigenvector can be chosen to be strictly positive almost surely. \hfill~$\Box$
\end{prop}

By our definition of the amplitudes~$\mathcal{A}_{\Omega}$ as the square-root of the Radon-Nikodym density between two mutually absolutely continuous positive finite measures ($\me^{-S_{\Omega}}$ is almost surely positive since~$S_{\Omega}$ is a real-valued random variable, recall also remark \ref{rem-zeta-det-positive} that the zeta-determinants are positive), we get immediately

\begin{lemm}
  For any traceable cobordism~$\Omega\in \Mor(\Sigma,\Sigma)$, the Segal transfer operator~$U_{\Omega}:L^2(\mathcal{D}'(\Sigma),\mu_{2\mn{D}})\lto L^2(\mathcal{D}'(\Sigma),\mu_{2\mn{D}})$ has strictly positive kernel. \hfill~$\Box$
\end{lemm}

We deduce that~$U_{\Omega}$ has a simple top eigenvalue~$\lambda_0=\nrm{U_{\Omega}}$ spanned by a normalized, almost surely strictly positive eigenvector~$\Omega_0\in L^2(\mathcal{D}'(\Sigma),\mu_{2\mn{D}})$. Alternatively speaking $U_{\Omega}$ has a \textsf{spectral gap}. We get
\begin{corr}\label{cor-segal-trans-exp-decay}
  Denote~$\wh{U}_{\Omega}:=\lambda_0^{-1}U_{\Omega}$, let~$\lambda_1$ be the eigenvalue of~$U_{\Omega}$ with next largest modulus, thus~$\lambda_0>|\lambda_1|$, and put~$\alpha:=|\lambda_1|/\lambda_0 <1$. Then for any~$F$,~$G\in L^2(\mathcal{D}'(\Sigma),\mu_{2\mn{D}})$, we have
  \begin{equation}
    \big|\bank{F,\wh{U}_{\Omega}^N G}-\bank{F,\Omega_0}\bank{\Omega_0,G}\big| \le \alpha^N \nrm{F}\nrm{G}.
    \label{}
  \end{equation}
\end{corr}

\begin{proof}
  Note~$\sank{F,\wh{U}_{\Omega}^N G}-\ank{F,\Omega_0}\ank{\Omega_0,G}=\sank{\Pi_0^{\perp}F,\wh{U}_{\Omega}^N \Pi_0^{\perp}G}$ where~$\Pi_0^{\perp}$ is the orthogonal projection onto the complement of~$\spn\{\Omega_0\}$, where~$\wh{U}_{\Omega}$ has norm~$\alpha<1$.
\end{proof}

The above corollary is exactly stating that the Gibbs state we defined has a mass gap, in perfect analogy with $1$D spin chains with local interactions which are well known to have no phase transitions.

\begin{corr}\label{cor-log-trace-asymp}
  We have
  \begin{equation}
    \lim_{N\to\infty}\frac{1}{N}\log\ttr (U_{\Omega}^N) =\log\lambda_0.
    \label{}
  \end{equation}
\end{corr}

\begin{proof}
  Without loss of generality let~$N\ge 2$ so each~$U_{\Omega}^N$ is trace class. On one hand
  \begin{equation}
    \ttr(U_{\Omega}^N)\le \bnrm{U_{\Omega}^2}_{\mm{tr}}\bnrm{U_{\Omega}^{N-2}}\le \bnrm{U_{\Omega}^2}_{\mm{tr}}\lambda_0^{N-2}.
    \label{eqn-log-trace-upper-bound}
  \end{equation}
  On the other, we decompose~$\mathcal{H}_{\Sigma}=\spn\{\Omega_0\}\oplus \spn\{\Omega_0\}^{\perp}$ where
  \begin{align*}
    \ttr(U_{\Omega}^N)&=\lambda_0^N +\ttr(U_{\Omega}^N|_{\spn\{\Omega_0\}^{\perp}}) \\
    &\ge \lambda_0^N -\bnrm{U_{\Omega}^2|_{\spn\{\Omega_0\}^{\perp}}}_{\mm{tr}}|\lambda_1|^{N-2}\\
    &=\lambda_0^N\left( 1-C_1 \alpha^{N-2}\lambda_0^{-2} \right).
  \end{align*}
  This and (\ref{eqn-log-trace-upper-bound}) gives the result after taking~$N\to\infty$.
\end{proof}

Corollary \ref{cor-segal-trans-exp-decay} and the proof of corollary \ref{cor-log-trace-asymp} implies

\begin{corr}\label{cor-gibbs-state-limit}
  For any bounded operator~$F\in \mathcal{L}(\mathcal{H}_{\Sigma})$ we have
  \begin{equation}
    \lim_{N\to \infty}\frac{\ttr(U_{\Omega}^{N-L}F)}{\ttr(U_{\Omega}^N)}=\frac{1}{\lambda_0^L}\ank{\Omega_0,F\Omega_0}=\frac{\ank{\Omega_0,F\Omega_0}}{\ank{\Omega_0,U_{\Omega}^L \Omega_0}}.
    \label{}
  \end{equation}
\end{corr}
\begin{proof}
  Remember that~$N\gg L$. Thus by~$\mathcal{H}_{\Sigma}=\spn\{\Omega_0\}\oplus \spn\{\Omega_0\}^{\perp}$ and \cite{Sim1} theorem 2.14 we have
  \begin{equation}
    \lambda_0^{-N+L} \ttr\big(U_{\Omega}^{N-L}F\big)=\bank{\Omega_0,\wh{U}_{\Omega}^{N-L}F\Omega_0}+ \ttr\big(\wh{U}_{\Omega}^{N-L}F|_{\spn\{\Omega_0\}^{\perp}}\big)
  \end{equation}
  with
  \begin{equation}
    \big|\ttr\big(\wh{U}_{\Omega}^{N-L}F|_{\spn\{\Omega_0\}^{\perp}}\big)\big| \le \nrm{F} \bnrm{\wh{U}_{\Omega}^{N-L}|_{\spn\{\Omega_0\}^{\perp}}} \to 0.
    \label{}
  \end{equation}
  Since $\sank{\Omega_0,\wh{U}_{\Omega}^{N-L}F\Omega_0}\to \sank{\Omega_0,F\Omega_0}$ by corollary \ref{cor-segal-trans-exp-decay} and~$\ttr(U_{\Omega}^N)\asymp \lambda_0^N$ by corollary \ref{cor-log-trace-asymp}, we obtain the result.
\end{proof}

We arrive at the conclusion that, for functionals~$F_1$, \dots,~$F_k\in L^{\infty}(\mathcal{D}'(\Sigma),\mu_{2\mn{D}})$ and integers~$1\le i_1<\cdots<i_k$, the evaluation
\begin{equation}
  \left.
  \def\arraystretch{2.5}
  \begin{array}{rcl}
     F_k\otimes\cdots\otimes F_1 &\longmapsto& \ddp \lim_{N\to \infty}\frac{\ttr\big( U_{\Omega}^{N+1-i_k}F_k U_{\Omega}^{i_k-i_{k-1}}\cdots F_1 U_{\Omega}^{i_1-1} \big)}{\ttr\big( U_{\Omega}^{N} \big)} \\
     &= & \ddp \frac{\bank{\Omega_0,F_k U_{\Omega}^{i_k-i_{k-1}}\cdots F_1 U_{\Omega}^{i_1-1}\Omega_0}}{\bank{\Omega_0,U_{\Omega}^{i_k-1} \Omega_0}}
  \end{array}
  \right.
  \label{}
\end{equation}
defines a Gibbs state in the thermodynamic limit on a~$\mathcal{D}'(\Sigma)$-valued~$\mb{Z}$-spin chain.

\begin{def7}
  In fact, a valid functional~$F\in L^{\infty}(\mathcal{D}'(\Sigma),\mu_{2\mn{D}})$ could be given by
  \begin{equation}
    F(\psi)\defeq \int_{}^{}\mathcal{A}_{\Omega}^0(\psi,\varphi)\dd\mu_{2\mn{D}}^{\Sigma}(\varphi)\int_{}^{}\tilde{F}(\phi^{D}_{\Omega},\psi,\varphi)\me^{-S_{\Omega}(\phi^{D}_{\Omega}|\psi, \varphi)}\dd\mu_{\mm{GFF}}^{\Omega,D}(\phi^{D}_{\Omega}),
    \label{}
  \end{equation}
  for~$\tilde{F}\in L^{\infty}(\mathcal{D}'(\Omega^{\circ}),\mu_{\mm{GFF}}^{\Omega,D}\otimes \mu_{2\mn{D}}^{\Sigma\sqcup\Sigma})$, with~$\mu_{2\mn{D}}^{\Sigma\sqcup\Sigma}$ considered as the induced measure on~$\mathcal{D}'(\Omega^{\circ})$ via~$\PI_{\Omega}^{\Sigma\sqcup\Sigma}$. Hence the Gibbs state actually extends to the continuum~$M_{\infty}^{\Sigma}$.
\end{def7}

\subsection{Asymptotic of Partition Function}

\noindent For the technical reason of remark \ref{rem-trace-and-diagonal} we shall assume the surface~$\Omega\in\Mor(\Sigma,\Sigma)$ be \textsf{reflection symmetric}, which in simplest words means~$\Omega=\tilde{\Omega}^*\cup_{\Sigma'}\tilde{\Omega}$ for some~$\tilde{\Omega}\in\Mor(\Sigma,\Sigma')$, in the notations of lemma \ref{lemm-segal-transfer}. In other words there is an isometric involution~$\Theta:\Omega\lto \Omega$ whose fixed point set is exactly~$\Sigma'$, and exchanges the two components of~$\partial\Omega$. This is not much of a restriction. In this case~$U_{\Omega}=U_{\tilde{\Omega}^*}\circ U_{\tilde{\Omega}}$ and \cite{Sim3} theorem 3.8.5 applies, namely~$\ttr(U_{\Omega})=\ttr_{\rho}(U_{\Omega})$, as well as for~$U_{\Omega}^N$. From corollary \ref{cor-P(phi)-trace} and corollary \ref{cor-log-trace-asymp} we then deduce that
\begin{equation}
  \lim_{N\to\infty}\frac{1}{N}\log(Z_{M_N^{\Sigma}})=\log \lambda_0.
  \label{}
\end{equation}
Thus we have proved theorem \ref{thrm-intro-main-2}. This applies in particular to the free case where~$Z_{M_N^{\Sigma}}=\detz(\Delta_{M_N^{\Sigma}}+m^2)^{-\frac{1}{2}}$.

\begin{def7}
  There arises the interesting question of how~$\lambda_0$ (or~$\log \lambda_0$) would depend on the geometry of~$\Omega$. Very crudely one would expect~$\log\lambda_0 \propto \vol(\Omega)$, since in our case the corresponding~$\tilde{\lambda}_0$ for~$N\Omega$ ($N$ copies of~$\Omega$ glued) is just~$\lambda_0^N$. A more precise formula in the general, non-periodic case seems desirable.
\end{def7}

\appendix

\newpage
 \section{Some Recollections}

    \subsection{Wiener Chaos and Wick's Theorem}
  
  \begin{prop}
  [Wiener Chaos decomposition, \cite{Janson} theorem 2.6]
  Let~$(Q,\mathcal{O},\mb{P})$ be a probability space and~$\mathcal{H}\subset L^2(Q,\mathcal{O},\mb{P})$ a Gaussian Hilbert space. Then there is an orthogonal decomposition
  \begin{equation}
    L^2(Q,\mathcal{O}(\mathcal{H}),\mb{P})\cong \bigoplus_{n=0}^{\infty}\mathcal{H}^{{:}n{:}},
    \label{}
  \end{equation}
  where~$\mathcal{O}(\mathcal{H})$ is the~$\sigma$-algebra generated by variables in~$\mathcal{H}$,~$\mathcal{H}^{{:}n{:}}=\ol{\mathcal{P}}_n(\mathcal{H})\cap \ol{\mathcal{P}}_{n-1}(\mathcal{H})^{\perp}$, where~$\mathcal{P}_j(\mathcal{H})$ denotes the span of polynomials of random variables in~$\mathcal{H}$ of degree~$\le j$; in particular~$\mathcal{H}^{{:}0{:}}$ denotes the constants.
\end{prop}

If~$F\in \ol{\mathcal{P}}_n(\mathcal{H})$, denote by~${:}F{:}$ the projection of~$F$ onto~$\mathcal{H}^{{:}n{:}}$, and is called a \textsf{Wick ordered polynomial}; for~$F\in L^2(\Omega,\mathcal{O}(\mathcal{H}),\mb{P})$, denote by~$I_n(F)$ its projection onto~$\mathcal{H}^{{:}n{:}}$. Define the \textsf{Hermite polynomials}~$h_n(x)$ by
\begin{equation}
  \exp\Big( zx-\frac{1}{2}z^2 \Big)=\sum_{n=0}^{\infty}\frac{z^n}{n!}h_n(x).
  \label{}
\end{equation}
We have~$h_0(x)=1$,~$h_2(x)=x^2-1$,~$h_4(x)=x^4-6x^2+3$, etc. 

\begin{lemm}
  [Wick's theorem/Feynman rules, \cite{Sim2} propositions I.2, I.3, I.4, \cite{Janson} theorems 1.28, 3.9, 3.19] \label{lemm-wick-feyn} ~
  \begin{enumerate}[(i)]
  \item For~$X_1$, \dots,~$X_n\in \mathcal{H}$ (not necessarily distinct) jointly Gaussian random variables,
  \begin{equation}
    \mb{E}[X_1\cdots X_n]=\sum_{\mss{P}}\prod_k \mb{E}[X_{i_k}X_{j_k}],
    \label{}
  \end{equation}
  where the sum is over all partitions~$\mss{P}$ of the set~$\{1,\dots,n\}$ into disjoint pairs~$\{i_k,j_k\}$,~$1\le k\le n/2$. In particular, the result is zero if~$n$ is odd.
    \item For~$X\in \mathcal{H}$,
      \begin{equation}
	{:}X^n{:}=\sum_{j=0}^{\lfloor n/2 \rfloor} \frac{(-1)^j n!}{(n-2j)! j! 2^j}\mb{E}[X^2]^j X^{n-2j}=\mb{E}[X^2]^{\frac{n}{2}} h_n\big(X\big/\mb{E}[X^2]^{\frac{1}{2}}\big).
	\label{}
      \end{equation}
    \item Let~$X_1$, \dots,~$X_n\in \mathcal{H}$ and~$Y_1$, \dots,~$Y_m\in \mathcal{H}$ be jointly Gaussian random variables. Then
   \begin{equation}
     \mb{E}\big[{:}X_1\dots X_n{:}~{:}Y_1\dots Y_m{:}\big]=\left\{
       \begin{aligned}
	 &\sum_{\sigma\in\fk{S}_n}\prod_{i=1}^n \mb{E}[X_i Y_{\sigma(i)}], && m=n,\\
	 &~~~0, && m\ne n.
       \end{aligned}
     \right.
     \label{eqn-wick-theorem-exp}
   \end{equation}
   In particular,~$\mb{E}[{:}X^n{:}~{:}Y^m{:}]=\delta_{nm}n!\mb{E}[XY]^n$.
  \end{enumerate}
\end{lemm}

 There is a general way of associating random variables (or numbers) to \textsf{Feynman diagrams}. A \textsf{Feynman diagram} consists of vertices, legs (segments with only one end attached to a vertex), and edges (contracted legs). Two legs are \textsf{contracted} means they are connected to form an edge. A Feynman diagram is called \textsf{fully contracted} if there is no unconnected legs. Given~$n$ random variables~$X_1$, \dots,~$X_n$, a Feynman diagram \textsf{labelled by}~$(X_1,\dots,X_n)$ is simply any Feynman diagram whose vertices are in bijection with~$\{X_1,\dots,X_n\}$. A Feynman diagram labelled by~$(X_1,\dots,X_n)$ can now be associated with \textit{either} a random variable \textit{or} a number using the following rules:
\begin{enumerate}[(i)]
  \item for each leg attached to a vertex~$j$, write down the corresponding random variable~$X_j$, and multiply them all together;
  \item whenever two legs are contracted, enclose the corresponding two random variables by~$\mb{E}[\bullet]$.
\end{enumerate}
Thus fully contracted diagrams are always associated with numbers. If~$\gamma$ is a Feynman diagram labelled by~$X_1$, \dots,~$X_n$, denote by~$v(\gamma)$ the associated object following the above rules.

\begin{exxx}[\cite{PS} section 4.4]
  In a physical context Feynman diagrams are used in rather formal calculations. We would like to find
  \begin{equation}
    \mb{E}_{\mm{GFF}}^M\big[\phi(x)\phi(y)\me^{-\frac{\lambda}{4!}\int_{M}^{}\phi(x)^4\dd V_M(x)}\big]
    \label{eqn-feynman-calc}
  \end{equation}
  by Taylor expanding~$\exp(-\frac{\lambda}{4!}\int_{M}^{}\phi(x)^4\dd V_M(x))$. This gives
  \begin{align*}
    (\textrm{\ref{eqn-feynman-calc}})&\heueq \mb{E}_{\mm{GFF}}^M\big[\phi(x)\phi(y)\big] +
    \mb{E}_{\mm{GFF}}^M\Big[\phi(x)\phi(y)\Big( -\frac{\lambda}{4!} \Big)\int_{}^{}\phi(z)^4\dd V_M(z) \Big] \\
    &\quad\quad + \mb{E}_{\mm{GFF}}^M\Big[\phi(x)\phi(y)\frac{1}{2}\Big( -\frac{\lambda}{4!} \Big)^2\iint_{}^{}\phi(z)^4\phi(w)^4\dd V_M(z)\dd V_M(w) \Big]  +\cdots\\
    &=\mb{E}_{\mm{GFF}}^M\big[\phi(x)\phi(y)\big]  -\frac{\lambda}{4!}
   \underbrace{ \int_{}^{}\mb{E}_{\mm{GFF}}^M\big[\phi(x)\phi(y)\phi(z)^4\big]\dd V_M(z) }_{A} \\
    &\quad\quad + \frac{1}{2}\Big( -\frac{\lambda}{4!} \Big)^2 \iint_{}^{} \mb{E}_{\mm{GFF}}^M\big[\phi(x)\phi(y)\phi(z)^4\phi(w)^4\big]\dd V_M(z)\dd V_M(w) +\cdots
  \end{align*}
  Treating~$\phi(x)$,~$\phi(y)$,~$\phi(z)$ as (jointly Gaussian!) random variables, by (i) of lemma \ref{lemm-wick-feyn} and the above rules we write
  \begin{align*}
    A&\heueq 3\int_{}^{} \dd V_M(z) v\bigg( \raisebox{-.5\height}{\includegraphics[height=2cm]{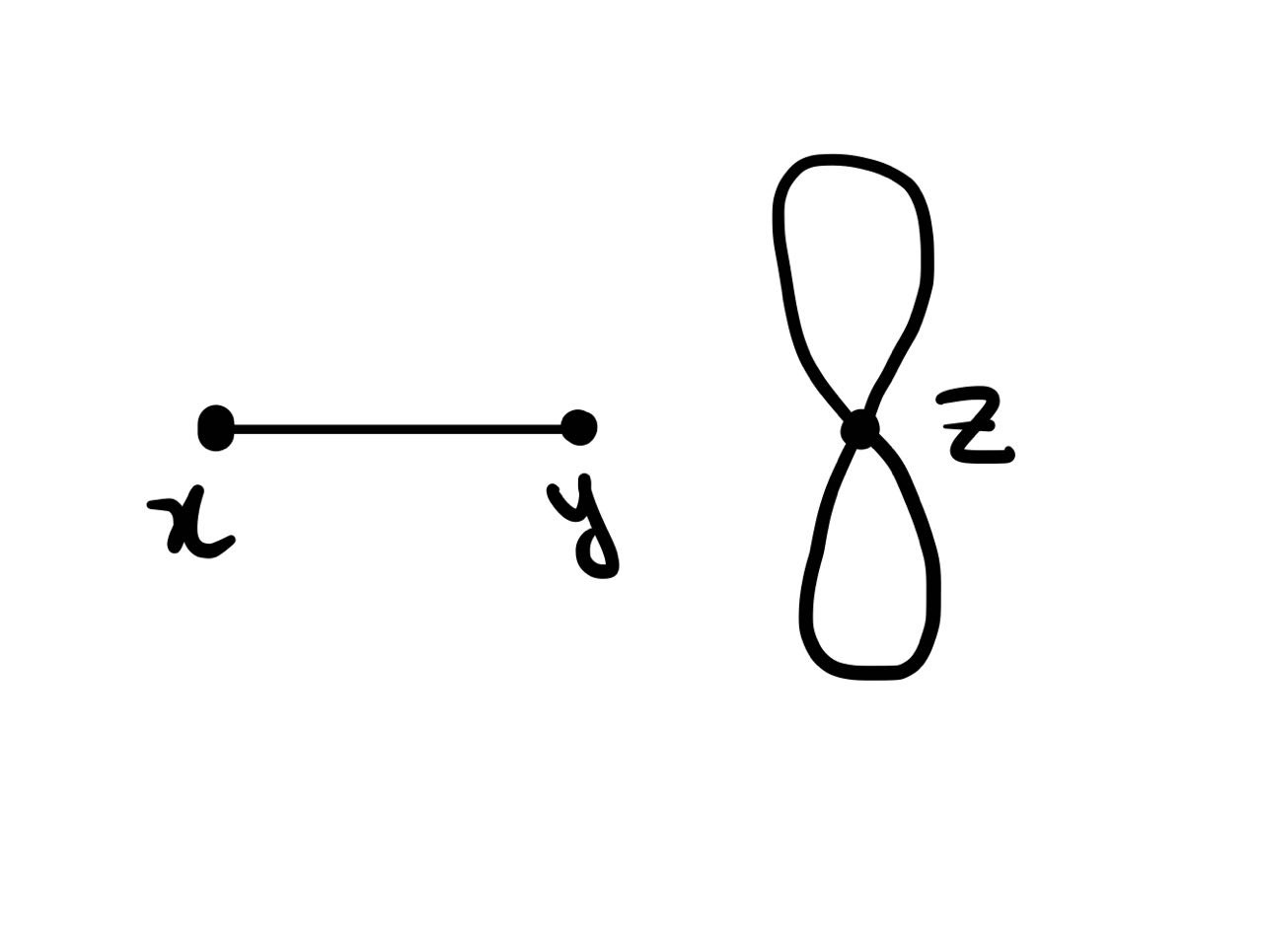}} \bigg) +12 \int_{}^{}\dd V_M(z) v\bigg( \raisebox{-.5\height}{\includegraphics[height=2cm]{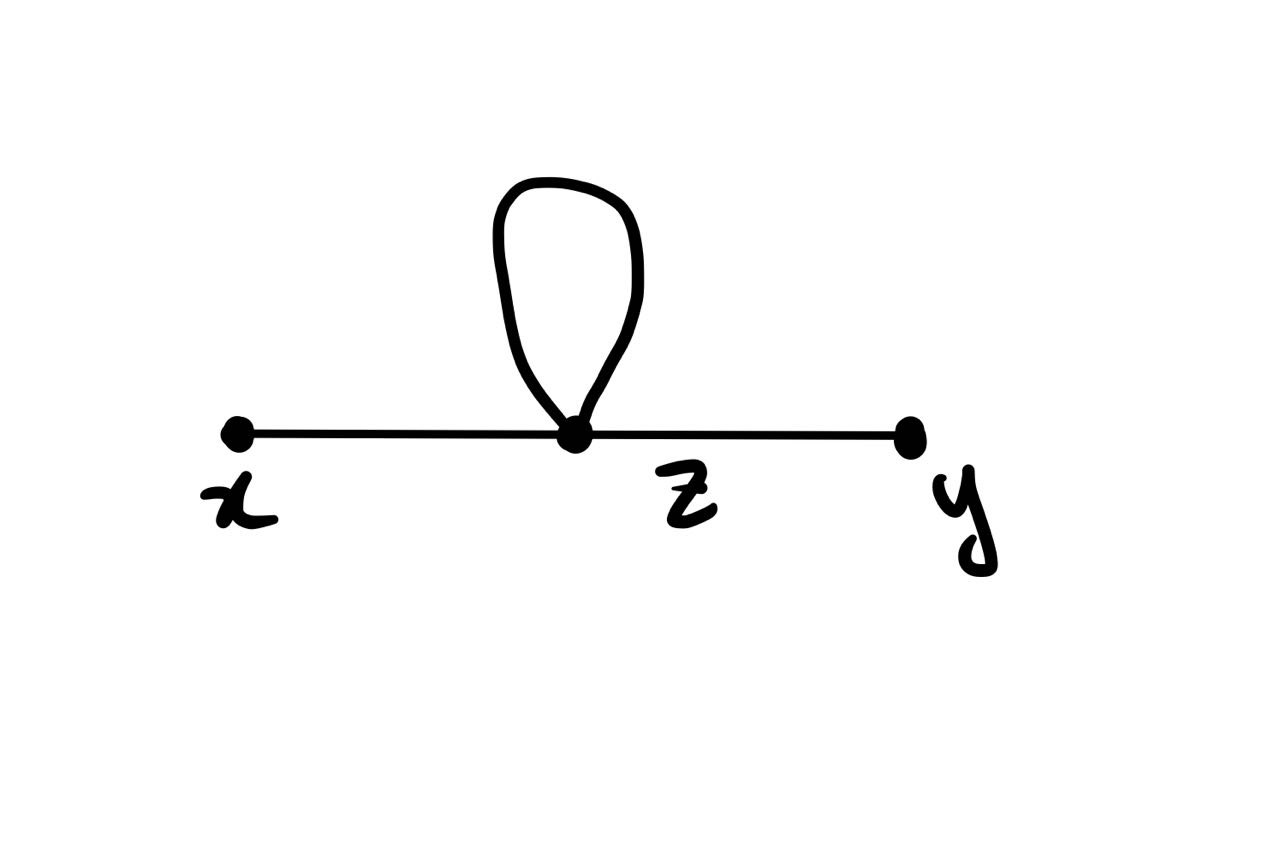}} \bigg) \\
    &=3\int_{}^{}G(x,y)G(z,z)^2 \dd V_M(z) +12 \int_{}^{} G(x,z)G(y,z)G(z,z)\dd V_M(z) ,
  \end{align*}
  where the factors correspond to the number of ways of getting the same contraction starting from 6 legs (4 on~$z$, 1 on~$x$,~$y$ each), and~$G$ denotes the Green function (of~$\Delta+m^2$). One can represent higher order terms using these diagrams in a similar manner. More than that, the diagrams also represent actual physical processes. See Peskin and Schroeder \cite{PS}.
\end{exxx}

 \subsection{Sobolev Spaces over Domains}\label{sec-app-sobo}

\noindent In this paper we make essential use of the usual~$L^2$ Sobolev spaces over Riemannian manifolds. First let~$(M,g)$ be a closed Riemannian manifold and~$s\in \mb{R}$. Then the Sobolev space~$W^s(M)$ of order~$s$ is defined generally as the closure of~$C^{\infty}(M)$ under a norm~$\nrm{\cdot}_{W^s(M)}$, where the norm~$\nrm{\cdot}_{W^s(M)}$ could be defined in various equivalent ways. We refer to Taylor \cite{Taylor1} chapter 4 for a general discussion. For us,~$s=\pm 1$,~$\pm\frac{1}{2}$. We rely heavily on the following fact.
\begin{lemm}\label{lemm-sobo-inner-prod}
  Let~$\Lambda_{2s}$ be an elliptic \textsf{strictly positive} formally self-adjoint pseudodifferential operator on~$M$ with order~$2s$. Then the inner product
  \begin{equation}
    \ank{-,-}_{W^s}\defeq \ank{-,\Lambda_{2s}-}_{L^2}
    \label{eqn-sobo-inner-prod}
  \end{equation}
  induces an equivalent norm for~$W^s(M)$.
\end{lemm}

In particular, the real power~$(\Delta_M+m^2)^s$ of the Helmholtz operator (massive Laplacian)~$\Delta_M+m^2$ provides such a candidate for~$\Lambda_{2s}$. \textbf{Convention:} whenever we use the space~$W^s(M)$, the inner product (\ref{eqn-sobo-inner-prod}) with~$\Lambda_{2s}=(\Delta_M+m^2)^s$ is understood, unless otherwise specified.

\begin{def7}
  Various regimes of functional calculus can be used to define~$(\Delta_M+m^2)^s$. One of them is presented in section \ref{sec-det-def} which in fact defines \textit{complex} powers. We also mention a smooth functional calculus presented in Sogge \cite{Sogge} theorem 4.3.1.
\end{def7}

Next we discuss important subspaces of~$W^s(M)$. Let~$A\subset M$ be a closed set and~$U\subset M$ an open set. Define
\begin{align}
  W^s_A(M)&\defeq \{u\in W^s(M)~|~\supp u\subset A\textrm{ as a distribution}\},\label{eqn-def-sobo-closed-support}\\
  W^s_U(M) &\defeq \textrm{closure of }C_c^{\infty}(U)\textrm{ inside }W^s(M), \label{eqn-def-sobo-open-closure}\\
  W^s(U)&\defeq W^s_{M\setminus U}(M)^{\perp}\subset W^s(M).
  \label{eqn-sobo-open-set-ortho}
\end{align}
These are closed subspaces of~$W^s(M)$. 

\begin{def7}
  We point out right away that by definition, then,
  \begin{equation}
    W^s(U)\cong W^s(M)/W^s_{M\setminus U}(M),
    \label{}
  \end{equation}
  the latter equipped with the quotient norm, which is a more familiar characterization of~$W^s(U)$, see Taylor \cite{Taylor1} page 339. Our definition as in (\ref{eqn-sobo-open-set-ortho}) poses the obvious problem that in general~$C_c^{\infty}(U)\not\subset W^s(U)$, at least for~$s\not\in \mb{Z}_+$. We emphasize therefore that what is important in this definition is not the space~$W^s(U)$ \textit{per se} but the following choice for its inner product:
  \begin{equation}
    \bank{f,h}_{W^s(U)}\defeq \bank{P_{M\setminus U}^{\perp} f,P_{M\setminus U}^{\perp} h}_{W^s(M)},
    \label{eqn-def-inner-prod-sobo-open}
  \end{equation}
  for any~$f$,~$h\in W^s(M)$, in particular for~$f$,~$h\in C_c^{\infty}(U)$, which produces a norm equivalent to the quotient norm, where~$P_{M\setminus U}^{\perp}:W^s(M)\lto W^s(U)$ denotes the orthogonal projection.
\end{def7}

\begin{def7}
  Clearly~$W^s_U(M)\subset W^s_{\ol{U}}(M)$ by definition. In general the inclusion is strict (certainly if $U\ne (\ol{U})^{\circ}$!). See Taylor \cite{Taylor1} page 339 and section 4.7 for interesting discussions on conditions for~$s$ and~$U$ for which equality holds. In particular,~$W^k_{\Omega}(M)=W^k_{\Omega^{\circ}}(M)= W^k_{\ol{\Omega}}(M)$ if~$\Omega\subset M$ is a domain with smooth boundary~$\partial\Omega$ (a closed Riemannian manifold with one dimension less) and~$k\in \mb{Z}_+$. In this case, we use these notations interchangeably. 
\end{def7}

The rest of this appendix could be read along with section \ref{sec-sobo-decomp}. Let~$s=-1$. Although~$C_c^{\infty}(U)\not\subset W^{-1}(U)$, we have

\begin{lemm}\label{lemm-sobo-test-dense-in-open}
  Let~$U\subset M$ be an open set. Then~$P_{M\setminus U}^{\perp}(C_c^{\infty}(U))$ is dense in~$W^{-1}(U)$.
\end{lemm}

\begin{proof}
  We note~$\Delta_M+m^2$ is local and therefore~$(\Delta+m^2)(C_c^{\infty}(U))\subset C_c^{\infty}(U)$. It follows from lemma \ref{rem-sobo-proj=dist-res} and our definition of~$W_U^1(M)$ that~$(\Delta+m^2)(C_c^{\infty}(U))\subset W^{-1}_{M\setminus U}(M)^{\perp}$ and is dense there, proving the result.
\end{proof}

\begin{def7}\label{rem-natural-sobo-cam-mar}
  Clearly, the map~$P_{M\setminus U}^{\perp}$ is also injective on~$C_c^{\infty}(U)$; together with lemma \ref{lemm-sobo-test-dense-in-open} this shows~$P_{M\setminus U}^{\perp}$ is a good embedding of~$C_c^{\infty}(U)$ in~$W^{-1}(U)$. In fact, this is the same as the embedding of~$C_c^{\infty}(U)$ in~$\mathcal{D}'(U)$, by remark \ref{rem-sobo-proj=dist-res}. Nevertheless, the smaller class~$(\Delta+m^2)(C_c^{\infty}(U))$, as it is already dense in~$W^{-1}(U)$, suffices as a class of test functions to define the GFF with Dirichlet condition over a domain (see remark below lemma \ref{lemm-ghs-for-gff}). This reflects the fact that the Cameron-Martin pairing~$\ank{-,-}_{W^1}$ is more natural than~$\ank{-,-}_{L^2}$ in treating the GFF (see remark \ref{rem-use-cam-mar-pairing-not-l2}). We have stuck to~$\ank{-,-}_{L^2}$ only because this is more practical with functional analysis.
\end{def7}

\begin{def7}
For general~$s\in\mb{R}$, one could also define
  \begin{equation}
    W^s_0(U)\defeq \textrm{closure of }C_c^{\infty}(U)\textrm{ under (\ref{eqn-def-inner-prod-sobo-open})}.
    \label{}
  \end{equation}
  Then~$W^s_U(M)\subset W^s_0(U)$. But it cannot generally be compared with~$W^s_{\ol{U}}(M)$ (to the author's knowledge). See the exercises in Taylor \cite{Taylor1} pages 343-344 for more information.
\end{def7}

Next we state the duality results for the various spaces. Recall that~$\ank{-,-}_{L^2(M)}$ denotes both the inner product of~$L^2(M)$ and the distributional pairing between~$\mathcal{D}'(M)$ and~$C^{\infty}(M)$. Below, we extend it to denote also the pairing between dual Sobolev spaces (see (i) of the lemma below).
\begin{lemm}\label{lemm-sobo-dual}
  Let~$M$ be a closed Riemannian manifold,~$U\subset M$ an open set,~$A\subset M$ a closed set, and~$s\in\mb{R}$.
  \begin{enumerate}[(i)]
    \item $W^{-s}(M)$ is the dual Banach space, denoted~$W^s(M)^*$, of~$W^s(M)$ under~$\ank{-,-}_{L^2}$;
    \item the \textsf{annihilator} of~$W^s_U(M)$ under~$\ank{-,-}_{L^2}$ is~$W^{-s}_{M\setminus U}(M)$, that is,
      \begin{equation}
	W^{-s}_{M\setminus U}(M)=\{u\in W^{-s}(M)~|~\ank{u,f}_{L^2}=0\textrm{ for all }f\in W^s_U(M)\};
	\label{}
      \end{equation}
    the \textsf{annihilator} of~$W^s_A(M)$ is accordingly~$W^{-s}_{M\setminus A}(M)$;
  \item $W^s(U)^*\cong W^{-s}_U(M)$,~$W^s_U(M)^*\cong W^{-s}(U)$, these spaces being therefore reflexive.
  \end{enumerate}
\end{lemm}

Finally, when~$\Omega\subset M$ is a domain with smooth boundary~$\partial\Omega$, we define, in view of lemma \ref{lemm-sobo-decomp}, the \textsf{Dirichlet Green operator}~$(\Delta_{\Omega,D}+m^2)^{-1}:=(\Delta+m^2)^{-1}P_{M\setminus \Omega^{\circ}}^{\perp}:W^{-1}(\Omega^{\circ})\lto W_{\Omega^{\circ}}^1(M)$. Clearly this agrees with the usual definition. In terms of quadratic forms,
\begin{lemm}[\cite{Sim2} theorem VII.1] \label{lemm-diri-green-op-quad-form}
  Let~$\Omega\subset M$ be a domain with smooth boundary~$\partial\Omega$. We have
  \begin{equation}
    \bank{f,(\Delta_{\Omega,D}+m^2)^{-1}h}_{L^2}=\bank{P_{M\setminus \Omega^{\circ}}^{\perp}f, P_{M\setminus \Omega^{\circ}}^{\perp}h}_{W^{-1}},
    \label{}
  \end{equation}
  for~$f$,~$h\in C_c^{\infty}(\Omega^{\circ})$.
\end{lemm}

 \subsection{Symbol Convergence Lemma and Heat Kernel}\label{app-symbol}

 \begin{proof}[Proof of lemma \ref{lemm-dyat-zwor-symbol-conv}.]
  By coordinate invariance of the definition of~$\Psi^m(M)$ it suffices to pick~$x\in M$ and prove the result for a chart around~$x$ and~$\chi(x)=1$. Denote the kernel of~$\chi E_{\varepsilon}\chi$ by~$E_{\chi,\varepsilon}$ then in this chart we could write
  \begin{equation}
    E_{\chi,\varepsilon}(x,y)=\tilde{E}_{\chi,\varepsilon}(x,h)=\frac{1}{F_{\varepsilon}(x)}\tilde{\psi}\left( \frac{h}{\varepsilon} \right)\tilde{\chi}(h),
    \label{}
  \end{equation}
  where~$h=x-y$. Indeed, by definition of our function~$\psi$ and freedom of choosing~$\chi$ we could further assume that for small enough~$\varepsilon$ one has~$\tilde{\chi}(h)\equiv 1$ on the support of~$\tilde{\psi}(\cdot/\varepsilon)$. Thus under this condition
  \begin{equation}
    \sigma_{\chi E_{\varepsilon}\chi}(x,\xi)=\int_{\mb{R}^d}^{} \me^{-\ii h\cdot \xi}\frac{1}{F_{\varepsilon}(x)} \tilde{\psi}\left( \frac{h}{\varepsilon} \right) \dd h =\frac{\varepsilon^d}{F_{\varepsilon}(x)}\underbrace{F_1(x) \sigma_{\chi E_1\chi}(x,\varepsilon\xi)}_{\textrm{indep. of }x}.
    \label{eqn-reg-symbol-expression}
  \end{equation}
  Note that~$\sigma_{\chi E_1\chi}(x,\eta)$ is Schwartz in~$\eta$ and~$\sigma_{\chi E_1\chi}(x,0)=1$. On the other hand clearly~$\sigma_{\chi\one\chi}(x,\xi)\equiv 1$. Thus for some~$U'\subset U$ depending only on the chart and~$\chi$, one has
  \begin{equation}
    \sup_{x\in K\Subset U'}\sup_{\xi}\frac{|\sigma_{\chi (E_{\varepsilon}-\one)\chi}(x,\xi)|}{\ank{\xi}^{\delta}} \le \left\{
      \def\arraystretch{1.3}
    \begin{array}{l}
      \ddp C\sup_{x\in K\Subset U'}\sup_{|\xi|\le R} \ank{\xi}^{-\delta} |\sigma_{\chi E_{1}\chi}(x,\varepsilon\xi)-1| \le C_{K,\chi} \sqrt{\varepsilon},\\
      \ddp C\sup_{x\in K\Subset U'}\sup_{|\xi|\ge R} (\cdots)\le C_{K,\chi}\varepsilon^{\delta/2}\sup_{\eta}|\sigma_{\chi E_1\chi}(x,\eta)|,
    \end{array}
    \right.
    \label{eqn-reg-symbol-bound}
  \end{equation}
  with~$R=\varepsilon^{-1/2}$. Next we deal with derivatives. Note that by (\ref{eqn-reg-symbol-expression}) all the~$x$-derivatives fall on~$1/F_{\varepsilon}(x)$ and all~$\xi$-derivatives fall on~$\sigma_{\chi E_1\chi}(x,\varepsilon\xi)$. Indeed, one has~$|\partial_x^{\alpha}(1/F_{\varepsilon}(x))|\le C_{\alpha} \varepsilon^{-d}$ (see Dyatlov and Zworski \cite{DZ} page 28), and so when there are only~$x$-derivatives we obtain the same bounds as (\ref{eqn-reg-symbol-bound}) only with new constants depending on~$\alpha$. When there is at least one~$\xi$-derivative,
  \begin{equation}
    |\partial_{\xi}^{\beta}\left( \sigma_{\chi E_1\chi}(x,\varepsilon\xi) \right)|=|\varepsilon^{|\beta|}(\partial_{\xi}^{\beta} \sigma_{\chi E_1\chi})(x,\varepsilon\xi)|\le C_{\beta,K,\chi}\varepsilon^{|\beta|}\ank{\varepsilon\xi}^{\delta-|\beta|},\quad |\beta|\ge 1.
    \label{}
  \end{equation}
  Hence, on account of (\ref{eqn-reg-symbol-expression}) again,
  \begin{equation}
    \sup_{x\in K\Subset U'}\sup_{\xi}\frac{|\partial_x^{\alpha}\partial_{\xi}^{\beta}\sigma_{\chi (E_{\varepsilon}-\one)\chi}(x,\xi)|}{\ank{\xi}^{\delta-|\beta|}}\le C\varepsilon^d C_{\alpha}\varepsilon^{-d}C_{\beta,K,\chi}\varepsilon^{|\beta|}=C_{\alpha,\beta,K,\chi}\varepsilon^{|\beta|}.
    \label{}
  \end{equation}
  Consequently, all the~$\mathcal{S}^{\delta}_{1,0}$ seminorms of~$\sigma_{\chi(E_{\varepsilon}-\one)\chi}$ goes to zero as~$\varepsilon\to 0$. We obtain the result.
\end{proof}

In what follows we sum up some properties of the \textsf{heat operator}~$\me^{-t(\Delta+m^2)}$ of the massive Laplacian (Helmholtz operator) and its Schwartz kernel~$p_t(x,y)$ called the \textsf{heat kernel}.

 \begin{lemm}[\cite{BGV} theorems 2.30, 2.38 and pages 92-94] \label{lemm-heat} We have
    \begin{enumerate}[(i)]
      \item $p_t(x,y)\in C^{\infty}((0,\infty)\times M\times M)$;
      \item we have
  \begin{equation}
    (\Delta+m^2)^{-1}=\int_{0}^{\infty}\me^{-t(\Delta+m^2)}\dd t.
    \label{eqnB8}
  \end{equation}
  In particular, the kernel~$G_t$ of~$\me^{-t(\Delta+m^2)}(\Delta+m^2)^{-1}=(\Delta+m^2)^{-1}\me^{-t(\Delta+m^2)}$ is
  \begin{equation}
    G_{t}(x,y)=\int_{t}^{\infty} p_s(x,y)\dd s,
    \label{eqnB9}
  \end{equation}
  for~$x$,~$y\in M$.
\item Let~$\dim M=n$. There are asymptotic expansions
  \begin{align}
  p_t(x,y)&\sim \frac{1}{(4\pi t)^{n/2}}\me^{-\frac{1}{4t}d(x,y)^2}\sum_{i=0}^{\infty} f_i(x,y) t^i,\\
    \ttr_{L^2(M)}\big(\me^{-t(\Delta+m^2)}\big)&\sim \frac{1}{(4\pi t)^{n/2}}\sum_{i=0}^{\infty} a_i t^i
    \label{}
  \end{align}
  as~$t\to 0+$, for some real numbers~$a_i$ and functions $f_i\in C^{\infty}(M\times M)$,~$i=0$,~$1$,~$2$, \dots.
\item For~$t$ large and each~$\ell\in\mb{N}$,
  \begin{equation}
    \nrm{p_t(x,y)}_{C^{\ell}}\le C_{\ell}\me^{-tm^2/2}
    \label{}
  \end{equation}
  for some constant~$C_{\ell}$.
    \end{enumerate}
\end{lemm}

\end{document}